\documentclass[reprint,prx,superscriptaddress,floatfix]{revtex4-2}%
\usepackage{amssymb}  

\usepackage[dvipsnames]{xcolor}  
\usepackage{caption}
\usepackage{subcaption}
\usepackage{amsthm}
\usepackage{thmtools,thm-restate}
\usepackage{dsfont}  
\usepackage{physics}
\usepackage[compat=0.6]{yquant}
\usepackage{makecell}
\usepackage{diagbox}
\usepackage{tikz}

\usepackage{gridstyles}

\newcommand{\midswap}{Mid-SWAP}
\newcommand{\midswapSC}{Mid-SWAP syndrome extraction}

\newcommand{\swapSC}{SWAP syndrome extraction}
\newcommand{\averageMLEDecoder}{Average-MLE decoder}
\newcommand{\marginalMatchingDecoder}{Marginal-Matching decoder}
\newcommand{\envelopeMLEDecoder}{Envelope-MLE decoder}
\newcommand{\envelopeMatchingDecoder}{Envelope-Matching decoder}

\newcommand{\ident}[0]{\mathds{1}}

\usepackage{mathtools}

\usepackage{booktabs}
\usepackage{multirow}
\usepackage{colortbl}
\usepackage{todonotes}
\usepackage{tcolorbox}
\usepackage[colorlinks=true, allcolors=blue, hypertexnames=false]{hyperref}
\usepackage{zref-clever}
\zcsetup{cap=true}  
\zcRefTypeSetup{figure}{
  Name-sg = Fig.,
  Name-pl = Figs.,
  name-sg = Fig.,
  name-pl = Figs.
}

\newtheorem{theorem}{Theorem}
\newtheorem{definition}{Definition}
\newtheorem{lemma}{Lemma}

\newtheorem{example}{Example}
\newtheorem{algorithm}{Algorithm}

\begin{document}

\title{Low-Overhead Surface Code with Delayed Atom-Loss Detection
via Pauli Envelope}

\author{Pengyu Liu}
\affiliation{QuEra Computing Inc., 1284 Soldiers Field Road, Boston,
MA, 02135, US}
\affiliation{Carnegie Mellon University, 5000 Forbes Avenue, Pittsburgh,
PA, 15213, US}
\author{Shi Jie Samuel Tan}
\affiliation{QuEra Computing Inc., 1284 Soldiers Field Road, Boston,
MA, 02135, US}
\affiliation{Joint Center for Quantum Information and Computer Science,
NIST/University of Maryland, College Park, Maryland 20742, US}
\author{Eric Huang}
\affiliation{QuEra Computing Inc., 1284 Soldiers Field Road, Boston,
MA, 02135, US}
\affiliation{Joint Center for Quantum Information and Computer Science,
NIST/University of Maryland, College Park, Maryland 20742, US}
\author{Umut A. Acar}
\affiliation{Carnegie Mellon University, 5000 Forbes Avenue, Pittsburgh,
PA, 15213, US}
\author{Hengyun Zhou}
\affiliation{QuEra Computing Inc., 1284 Soldiers Field Road, Boston,
MA, 02135, US}
\author{Chen Zhao}
\email{czhao@quera.com}
\affiliation{QuEra Computing Inc., 1284 Soldiers Field Road, Boston,
MA, 02135, US}

\begin{abstract}
  Atom loss is a major error source in neutral-atom quantum
  computers, accounting for over 40\% of the total physical errors in
  recent experiments.
  We consider the loss model of current alkali-atom architectures, in
  which a loss is revealed only by the destructive readout at the end
  of the atom's life cycle and removes the gates acting on the lost
  atom, as opposed to erasures heralded mid-circuit.
  The nonlinear and correlated nature of this error poses significant
  challenges:
  current syndrome extraction circuits require additional overhead
  or sacrifice loss tolerance, and existing decoders are
  computationally inefficient, suboptimal, or lack provable
  guarantees.
  To address these challenges, we propose the
  \emph{Pauli Envelope} framework, which bounds the
  effect of atom loss with low-weight, efficiently computable
  Pauli approximations, generalizing existing loss-to-Pauli methods
  and enabling rigorous analysis.
  Guided by this framework, we design an improved atom-replenishing
  syndrome extraction circuit, the \midswapSC{}, which achieves the
  optimal loss distance for rotated surface codes at the space-time
  cost of standard syndrome extraction, whereas previous circuits
  reach it only with additional overhead or only in the absence of
  Pauli errors.
  We also propose two decoders:
  an \emph{\envelopeMLEDecoder{}} achieving the optimal loss
  distance $d_{\mathrm{loss}}\sim d$, and an
  \emph{\envelopeMatchingDecoder{}} achieving
  $d_{\mathrm{loss}}\sim 2d/3$ via Minimum-Weight Perfect
  Matching (MWPM), surpassing the previous best
  ($d_{\mathrm{loss}}\sim d/2$) and readily integrating with fast
  correlated decoding techniques for transversal logical circuits.
  Circuit-level simulations demonstrate over
  40\% higher thresholds and 30\% higher effective distances compared with
  existing methods in the loss-dominated regime.
  Moreover, we explore correlated atom loss and show that it is easier to
  correct than independent loss, with thresholds rising from
  $5.15\%$ to $7.82\%$.
  Remarkably, our
  \envelopeMLEDecoder{} improves the error suppression factor of a
  hybrid MLE--machine-learning decoder from $\Lambda = 2.14$ to
  $\Lambda = 2.24$ on recent experimental data.
\end{abstract}

\maketitle

\section{Introduction}
\label{sec:introduction}
Quantum computing promises transformative advances in
cryptography~\cite{shor1994algorithms},
materials science~\cite{daley2022practical}, and
optimization~\cite{cerezo2021variational}, yet
realizing this potential
requires fault-tolerant operation through quantum error
correction~\cite{gottesman2024surviving}.
The most widely studied quantum error-correcting code for achieving
fault tolerance is the surface code~\cite{dennis2002topological,
fowler2012surface}. It is known to possess a high physical error
threshold under various noise models~\cite{tan2024resilience,
chadwick2024averting}
and can be implemented with nearest-neighbor connectivity.
Among the leading platforms for fault-tolerant quantum computing,
neutral-atom quantum computers have demonstrated exceptional
scalability and gate
fidelity~\cite{chiu2025continuous,bluvstein2025architectural,reichardt2024logical,PhysRevLett.130.180601,Pause:24}.
However, atom loss remains one of the key challenges for achieving
fault tolerance in these systems, accounting for more
than $40\%$ of the total physical errors in recent
experiments~\cite{bluvstein2025architectural}.


In this paper we consider the loss model of current alkali-atom
architectures~\cite{bluvstein2025architectural,chiu2025continuous}: a
loss is revealed only by the destructive, loss-resolving readout at
the end of the atom's life cycle.
Between the loss and this readout, every gate scheduled on the lost
atom acts as the identity; in particular, a two-qubit gate leaves its
partner qubit untouched, as if the gate had been removed from the
circuit.
The missing entangling gates cause correlated Clifford errors that
propagate through the
computation~\cite{baranes2025leveraging,yu2025locating,yu2024processing,brown2020critical}.
%
This delayed detection is specific to the platform and is not the
only possible loss model: platforms with heralded mid-circuit leakage
detection, demonstrated with alkaline-earth
atoms~\cite{ma2023high,scholl2023erasure,tao2026universal}, lead
instead to the erasure model~\cite{wu2022erasure}, which we use only
as a point of comparison (\zcref{sec:background}).  Compared with
erasures, gate-removing loss with delayed detection exhibits two
features that make it particularly challenging.  First,
atom loss persists until replenishment, causing errors to accumulate
across multiple gates and even multiple syndrome extraction
rounds. Specialized syndrome extraction circuits are therefore needed
to localize loss effects in both time and space.  Second, unlike Pauli
errors, which propagate and combine linearly, atom loss behaves
nonlinearly. Multiple atom losses may produce detector patterns that
cannot be derived by composing individual loss patterns, posing
significant challenges for decoder design and analysis.  To address
these challenges, we develop the \emph{Pauli Envelope} framework to
linearize the nonlinear loss effects into tractable Pauli errors.
This framework rigorously characterizes how atom loss effects can be
bounded by low-weight Pauli errors, guiding the design of
syndrome extraction circuits and decoding algorithms with
provable performance guarantees.
%
%
Unlike prior heuristic
approximations, the Pauli envelope provides two key guarantees:
(1)~\emph{correctness}: if the Pauli envelope is decoded correctly,
the original atom-loss events are guaranteed to be decoded correctly;
and (2)~\emph{low weight}: the Pauli envelope has low weight for loss
errors, leading to stronger distance guarantees.

Guided by this framework, we make the following contributions:

\textbf{(1) Syndrome extraction circuit:} We introduce the
\textit{\midswapSC{}}, which reduces the propagation of
loss effects with no additional space-time overhead compared to the
conventional rotated surface code syndrome extraction circuit.
To quantify this improvement, we use the \emph{loss distance}
$d_{\mathrm{loss}}$, which characterizes how many simultaneous atom
losses a code-decoder pair can tolerate before a logical error occurs
(see \zcref{def:effective-distance} for the formal definition).
This circuit achieves $d_{\mathrm{loss}}\sim d$, where $d$ is the code
distance (we write $d_{\mathrm{loss}}\sim f(d)$ to
denote $d_{\mathrm{loss}} = f(d) \pm O(1)$), which is
optimal for a distance-$d$ surface code and significantly
outperforms the \swapSC{}~\cite{baranes2025leveraging}, which we will show only
achieves $d_{\mathrm{loss}}\sim d/2$.

\textbf{(2) Optimal-distance decoder:} We propose the
\textit{Envelope-Most-Likely-Error decoder (Envelope-MLE decoder)},
which achieves $d_{\mathrm{loss}}\sim d$
for the \midswapSC{}. This significantly outperforms the
\textit{\averageMLEDecoder{}} of~\cite{baranes2025leveraging},
which we will show only achieves $d_{\mathrm{loss}}\sim d/2$. The key insight is
that each atom loss triggers exactly one detector pattern in our
Pauli envelope formulation, imposing an \textit{exclusivity}
constraint that the \averageMLEDecoder{} fails to enforce.

\textbf{(3) Efficient decoder:} Inspired by the exclusivity
constraint, we design the \textit{\envelopeMatchingDecoder{}} for
efficient decoding. Instead of changing edge weights to a constant,
which becomes negligible compared with Pauli error weights when
$p_{\mathrm{pauli}}\to 0$, as done in the
\textit{\marginalMatchingDecoder{}} by Gu et
al.~\cite{gu2025fault,gu2024optimizing}, we rescale the affected edge
weights by a
constant factor to discourage selecting multiple edges in the same
envelope, approximately enforcing the exclusivity constraint. This
decoder achieves $d_{\mathrm{loss}}\sim 2d/3$ for the
\midswapSC{}, outperforming the $d_{\mathrm{loss}}\sim d/2$ of the
\marginalMatchingDecoder{}. As an MWPM-based decoder, the
\envelopeMatchingDecoder{} naturally integrates with fast
correlated decoding techniques for transversal logical
circuits~\cite{cain2025fast,zhou2024algorithmic,serra2025decoding}.

\begin{table}[h]
  \centering
  \small
  \resizebox{\columnwidth}{!}{%
    \begin{tabular}{l|cccc}
      \toprule
      \diagbox[width=10em,height=3.5em]{\textbf{Decoder}}{\textbf{Loss\ Dist.}}{\textbf{Circuit}}
      & SWAP & \makecell{\textbf{Mid-}\\\textbf{SWAP}} &
      Teleportation & \makecell{Unrotated\\SWAP~\cite{yu2025locating}}\\
      \midrule
      Average-MLE            & $d/2$ &
      $d/2$ & $d/2$ & $d$ \\
      Marginal-Matching & $d/2$ &
      $d/2$ & $d/2$ & $d$ \\
      \textbf{Envelope-MLE}                & $d/2$ &
      $\mathbf{d}$ & $d$ & $d$ \\
      \textbf{Envelope-Matching}  & $d/2$ & $\mathbf{2d/3}$ &
      $2d/3$ & $d$ \\
      \midrule
      \textbf{Space-Time Cost} & $1\times$ &
      $1\times$ & $\geq 2\times$ & $2\times$ \\
      \bottomrule
    \end{tabular}%
  }
  \caption{Theoretical performance comparison of decoders and syndrome
    extraction circuits with atom-loss replenishment. The bold entries are
    introduced
    in this paper. All distances are up to additive
    constants. Space-time cost is the number of qubits multiplied by the
    duration of one syndrome extraction round, relative to the standard
    rotated surface code circuit. Teleportation-based syndrome
    extraction is briefly
    discussed in \zcref{app:teleportation_circuit}. Unrotated surface
    codes will be discussed as a special case of HGP
  codes in \zcref{app:css}. }
  \label{tab:decoder-comparison}
\end{table}

\zcref{tab:decoder-comparison} summarizes the theoretical performance
of our proposed decoders (\envelopeMLEDecoder{} and \envelopeMatchingDecoder{})
compared to existing approaches across several atom-loss correction
settings, including our newly proposed \midswapSC{}. Unless explicitly
stated otherwise, we consider rotated surface codes throughout.

Before this work, distance-$d$ scaling under a similar loss model
had been reported for the toric code~\cite{yu2025locating} and for a
measurement-based architecture~\cite{yu2024processing};
\zcref{app:css} derives the former as a special case of our
framework. Both, however, incur a higher space-time cost than the
rotated surface code on neutral-atom platforms; for instance, the
unrotated toric code uses twice as many data qubits at the same
distance~\cite{bombin2007optimal}.

On rotated surface codes, numerical studies have shown that the loss
distance can approach $\sim d$ in the absence of Pauli errors, using
SWAP replenishment~\cite{baranes2025leveraging} or additional
loss-detection units~\cite{perrin2024quantum}, and
Ref.~\cite{baranes2025leveraging} proposes a heuristic decoder for
the general case. However, pure-loss scaling does not by itself
establish the same tolerance jointly with Pauli errors
(\zcref{sec:background}), and the heuristic decoder comes without a
guarantee: as we show in \zcref{app:harvard_decoder}, it is limited
to $d_{\mathrm{loss}}\sim d/2$ once Pauli errors are present.

In contrast, we give a formal definition of loss distance that admits
any combination of loss and Pauli errors
(\zcref{def:effective-distance}). Using the Pauli envelope, we prove
that the \midswapSC{} with the \envelopeMLEDecoder{} attains
$d_{\mathrm{loss}}\sim d$ at the space-time cost of the standard
circuit; this is the largest loss distance possible, since $d$ losses
along a logical string erase the encoded information. The efficient
\envelopeMatchingDecoder{} attains $\sim 2d/3$ with the same circuit.

We evaluate the performance of
our proposed decoders and syndrome extraction circuit using
a circuit-level noise model. In the loss-dominated regime, our
approach achieves over $40\%$ higher thresholds and $30\%$
higher effective distances than prior approaches,
and we show that correlated atom loss is \emph{easier} to correct
than independent loss. \zcref{fig:decoder_schedule_ablation} previews
the size of the effect at pure loss: changing only the circuit or only
the decoder lowers the logical error rate by a factor between $3.0$
and $8.4$,
and the combination by $13.7\times$ at $d=9$ and $26.3\times$ at $d=11$.
\begin{figure*}[htbp]
  \centering
  \includegraphics[width=0.9\textwidth]{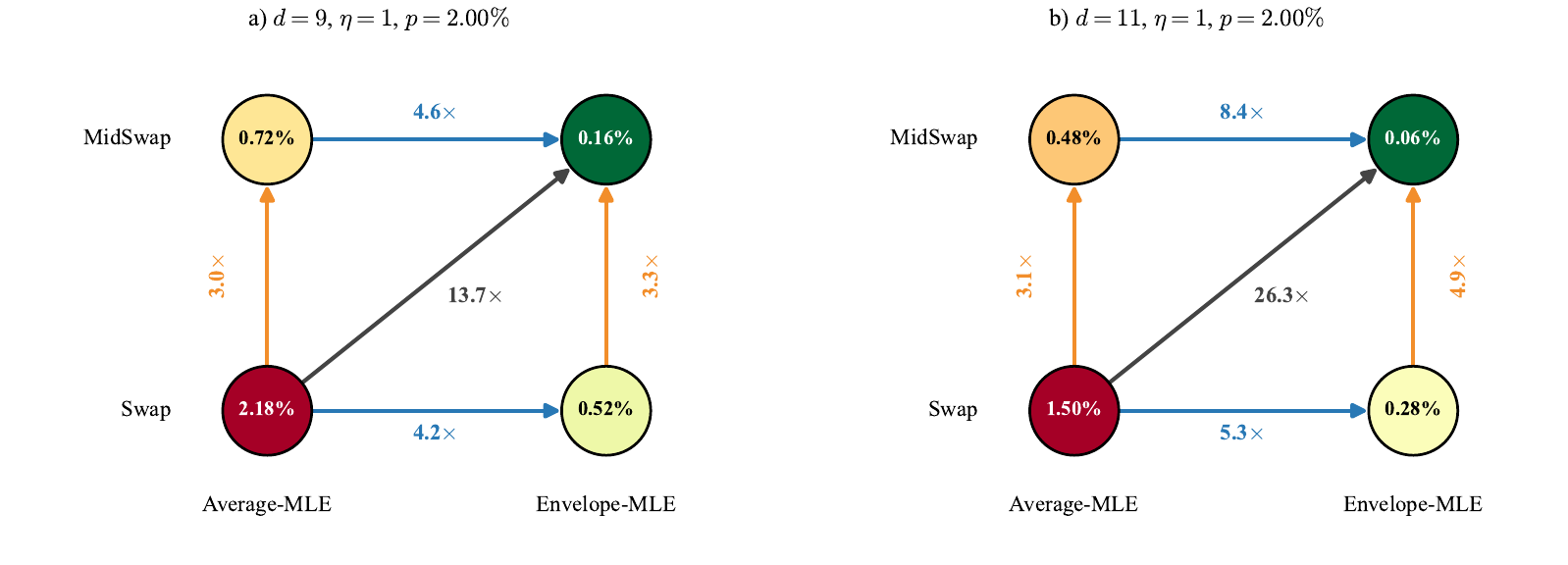}
  \caption{Logical error rate of the four decoder--circuit combinations
    at pure loss ($\eta=1$) and $p=2.0\%$ for (a)~$d=9$ and (b)~$d=11$,
    all with identical noise model and sampling. Edge
    labels give the error-reduction factor for changing only the decoder
    (horizontal), only the circuit (vertical), or both (diagonal).}
  \label{fig:decoder_schedule_ablation}
\end{figure*}


This paper is organized as follows: \zcref{sec:background} reviews
the atom-loss error model and essential background for quantum error correction.
\zcref{sec:pauli-envelope} introduces the Pauli envelope framework.
\zcref{sec:surface-codes} presents the \midswapSC{} circuit.
\zcref{sec:mle-decoder} describes the \envelopeMLEDecoder{}.
\zcref{sec:envelope-matching} presents the \envelopeMatchingDecoder{}.
\zcref{sec:evaluation} provides numerical evaluation.

\section{Atom Loss}
\label{sec:background}

Various terms are used in the literature to describe errors where qubits leave
the computational
subspace~\cite{levine2024demonstrating,mehta2025bias,chou2024superconducting}.
We distinguish them by two criteria, summarized in
\zcref{tab:error_comparison}: when the error is detected and how it
affects subsequent gates. We use \emph{atom loss} for an error that is
detected only at measurement and removes the subsequent gates on the
lost atom, and \emph{erasure} for an error that is detected
immediately and turns the subsequent gates into maximally depolarizing
errors.
\begin{table}[htbp]
  \caption{Comparison of erasure errors and atom loss.}
  \label{tab:error_comparison}
  \centering
  \setlength{\tabcolsep}{0.8pt}
  \begin{tabular}{c|c|c}
    \toprule
    & \makecell{Erasure \\ Error} &
    \makecell{Atom \\ Loss} \\
    \midrule
    \makecell{Detection timing} & Instantaneous & Upon
    measurement \\ \midrule
    \makecell{Effects} & Maximum depolarizing &
    Gate-removing \\ \midrule
    Distance & $d$ & $\sim \mathbf{d}$ \\
    \bottomrule
  \end{tabular}
\end{table}

\paragraph*{Hardware model.}
We model a neutral-atom computer of the alkali-atom architecture
family~\cite{bluvstein2025architectural,chiu2025continuous} as a
collection of atoms with finite life cycles: an atom is created in
$\ket{0}$, acted on by single- and two-qubit gates, and finally
measured. Following the state-selective readout demonstrated on these
platforms~\cite{bluvstein2025architectural,chiu2025continuous}, a
measurement has three outcomes: $\ket{0}$, $\ket{1}$, or lost. The
readout is destructive and ends the atom's life cycle, so it cannot
check for loss mid-circuit; when it reports a loss, we only learn that
the atom was lost at some point between its creation and its
measurement. Measured atoms are discarded, and lost atoms are replaced
only by creating fresh atoms in $\ket{0}$. This is the least demanding
setting capable of fault-tolerant quantum
computation~\cite{chiu2025continuous}, and we assume it throughout this
paper. Fault tolerance then requires syndrome extraction circuits that
regularly measure and replenish every atom, including the data qubits.

\paragraph*{Loss model.}
Following~\cite{baranes2025leveraging}, we model atom loss as a
gate-removing error: after an atom is lost, every subsequent gate
acting on it, single- or two-qubit, is removed from the circuit, i.e.,
acts as the identity (\zcref{fig:loss_example}). Gates that do not
involve the lost atom are unaffected, and the loss itself does not
spread to other atoms. A lost atom's partners only miss their interactions with
it until the lost atom is measured and replenished.
\begin{figure}[htbp]
  \centering
  \includegraphics[width=0.9\columnwidth]{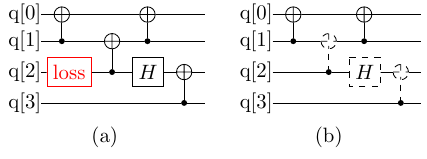}
  \caption{Atom loss as a gate-removing error. (a) An atom loss
    occurs at the beginning of the circuit on qubit $q[2]$. (b) The
    atom loss removes all subsequent gates acting on qubit $q[2]$
    (dashed); gates that do not involve $q[2]$, such as the second
    $CNOT$ between $q[0]$ and $q[1]$, are unaffected, because the loss itself
    does not spread to other atoms.}
  \label{fig:loss_example}
\end{figure}

In contrast to the alkali-atom architecture considered here,
alkaline-earth atoms allow leakage out of the qubit subspace to be
detected mid-circuit without disturbing the remaining
qubits~\cite{covey2019imaging,lis2023midcircuit,ma2023high,scholl2023erasure,tao2026universal};
such heralded leakage is an erasure error in the sense of
\zcref{tab:error_comparison} and is outside the scope of this paper.

\paragraph*{Detectors and nonlinearity.}
Decoding relies on the concept of \textit{detectors}.
Detectors are linear combinations of measurements
that are deterministic under ideal execution of a Clifford circuit. For example,
in a surface code quantum memory, a detector is typically the XOR of
two consecutive syndrome
measurements at the same stabilizer location: under ideal execution, both
measurements yield the same result, so their XOR is always $0$. During noisy
execution, we obtain a detector pattern $\mathcal{D}$, indicating the set of
detectors that are different from the ideal outcome, and a decoder finds an
error configuration that can explain this pattern.

A \textit{detector error model} (DEM) is a collection of elementary error
mechanisms, each
with an associated probability of occurrence and the detector pattern it
triggers. For Pauli noise in Clifford circuits, errors compose linearly: the
combined detector pattern of multiple Pauli errors is the sum
(modulo 2) of their
individual detector patterns, which enables a compact DEM description.

Atom loss breaks this additivity because it introduces Clifford
errors. The interactions between Clifford errors are much more
complicated than those of Pauli errors.
As a result, the detector pattern
resulting from multiple loss events generally cannot be obtained by summing the
detector patterns of individual loss events in isolation.

\zcref{fig:nonlinear_loss} shows a minimal example: without error,
or losing either the first or the second atom alone, the final qubit yields
$\ket{1}$ in panels (a)--(c), while with both losses it yields $\ket{0}$
in panel (d).
\begin{figure}[htbp]
  \centering
  \includegraphics[width=0.9\columnwidth]{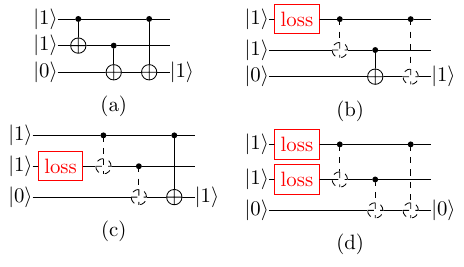}
  \caption{The composition of atom loss is nonlinear.}
  \label{fig:nonlinear_loss}
\end{figure}

This nonlinearity poses a significant challenge for decoding: even if we know
exactly which atoms were lost and when, we do not have a simple closed-form
expression for the distribution of measurement outcomes, as we do for Pauli
errors. This makes rigorous analysis of atom-loss decoding performance
challenging.

\paragraph*{Loss distance.}
To systematically analyze and compare the performance of different
syndrome extraction circuits and decoders, we define the \emph{loss
distance} through the following threshold-style scaling:

\begin{definition}[Loss Distance]
  \label{def:effective-distance}
  A distance-$d$ code and decoder pair has loss distance
  $d_{\mathrm{loss}}$ if there exists a constant $p_{\mathrm{th}}>0$,
  independent of $d$, such that the logical error rate
  $P_{\mathrm{LER}}$ satisfies
  \[
    P_{\mathrm{LER}} =
    O\left(p_{\mathrm{th}}^{-d}(p_{\mathrm{loss}}^{d_{\mathrm{loss}}}+p_{\mathrm{pauli}}^{d/2})\right),
  \]
  where $p_{\mathrm{loss}}$ and $p_{\mathrm{pauli}}$ denote the
  physical loss and Pauli error rates, respectively, and the bound
  holds uniformly for sufficiently small values of both rates.
\end{definition}
The exponent $d/2$ gives a uniform bound including mixed loss--Pauli
terms.
In other words, if a cross-term such as
$p_{\mathrm{loss}}^{d/2}p_{\mathrm{pauli}}$ appears in the logical
error rate, then $d_{\mathrm{loss}}\leq d/2+O(1)$.
In contrast, cross-terms such as
$p_{\mathrm{loss}}^{d/2}p_{\mathrm{pauli}}^{d/4}$ are consistent
with $d_{\mathrm{loss}}=d$, since
$p_{\mathrm{loss}}^{d/2}p_{\mathrm{pauli}}^{d/4} \leq
p_{\mathrm{loss}}^{d}+p_{\mathrm{pauli}}^{d/2}$.
For simplicity, we write the Pauli exponent as $d/2$ instead of the
exact value $(d+1)/2$ of an odd-distance code. This makes the bound
slightly looser and does not affect $d_{\mathrm{loss}}$, which we
state only up to additive constants.

\section{Pauli Envelope}
\label{sec:pauli-envelope}
In this section, we formalize the Pauli envelope, prove its key
properties, and show
how to construct Pauli envelopes for atom loss.
A summary of the notation used throughout the paper is given in
\zcref{app:notation}.

\subsection{Formalization of Pauli Envelope}

The key idea behind the Pauli envelope is simple: although atom loss
behaves nonlinearly, its effect on the decoder can be
\emph{bounded} by a set of Pauli errors. If a decoder can correctly
handle all errors in this bounding set, it will also correctly
decode the original atom loss. This reduces the complex problem of
decoding atom loss to the well-understood problem of decoding Pauli errors.

We now formalize this intuition. Consider a Clifford circuit
$\mathcal{C}$ with Pauli error configuration $p$ and loss
configuration $l$, where $l$ is a set of space-time locations at
which atoms are lost and $p$ assigns a Pauli operator to each faulty
space-time location. If an atom is lost
more than once in a
configuration, we ignore subsequent loss events.
We write $p \oplus p'$ for the composition of two Pauli error
configurations, i.e., the product of their Pauli operators at each
location; global phases never affect detectors or observables and are
ignored. For loss configurations, $l \oplus l'$ denotes the union of
their loss locations.

Let $(\mathcal{D}(p,l), \mathcal{O}(p,l))$ denote the random
detector-observable pair produced by running $\mathcal{C}$ with
error $p$ and loss $l$. We omit the circuit subscript
$\mathcal{C}$ when it is clear from context.
Measuring an atom $m$ yields a three-outcome physical readout
$\tilde b_m \in \{0, 1, \bot\}$, where $\bot$ indicates that the atom
is lost. Detectors and observables are computed from binary bits
$b_m$: $b_m = \tilde b_m$ if the atom is present, and $b_m$ is
assigned a uniformly random bit if $\tilde b_m = \bot$. This random
assignment is a modeling convention for the missing measurement bit
and is independent of the Pauli envelope introduced below.
Define $\mathcal{S}(p,l)$ as the support of this pair, i.e., all
detector-observable pairs with nonzero probability at fixed $p,l$.
For sets $P,L$ of configurations, define
\[
  \mathcal{S}(P,L)=\bigcup_{p\in P,\,l\in L}\mathcal{S}(p,l).
\]
When an argument is a single configuration, it stands for its
singleton set; in particular, $\emptyset$ denotes the no-loss
configuration.
Probabilities below also average over circuit measurements and
assigned bits.

The \emph{loss-resolving readout} $r$ collects the bits
$r_m = 1$ if $\tilde b_m = \bot$ and $r_m = 0$ otherwise over all
measurements. With perfect loss detection, it is a deterministic
function of the loss
configuration, denoted $r = \mathcal{R}(l)$, and readouts compose by
taking the union of the atoms they report as lost, $r \oplus r'$.
Multiple loss
configurations (differing in \emph{when} those atoms were lost)
may produce the same loss-resolving readout, so we define
$\mathcal{L}(r) = \{l : \mathcal{R}(l) = r\}$ as the
preimage.

\begin{definition}[Pauli Envelope]
  \label{def:pauli-bounded}
  A loss configuration $l$ has Pauli envelope $E$, where $E$ is a set
  of Pauli error configurations, if for all Pauli error configurations $p$,
  \[\mathcal{S}(p, l) \subseteq
  \mathcal{S}(p \oplus E, \emptyset),\]
  where $\emptyset$ denotes no atom loss and $p \oplus E$ denotes the
  composition of $p$ with errors
  from $E$, i.e., $p \oplus E =
  \{p \oplus e : e \in E\}$.
\end{definition}

Similarly, we say a readout $r$ has Pauli envelope $E$ if
$\mathcal{S}(p, \mathcal{L}(r)) \subseteq \mathcal{S}(p \oplus
E, \emptyset)$ for all $p$.

In other words, any detector-observable pair that can arise from a
Pauli error $p$ combined with any loss configuration producing
readout $r$ can also arise from the Pauli errors $p \oplus E$
with no atom loss.

The Pauli envelope ensures that if we can correctly decode the
Pauli envelope,
we can decode the atom losses correctly.

\begin{lemma}[Sufficiency of Decoding Pauli Envelope]
  \label{lem:sufficiency}
  Let $\mathcal{C}$ be a Clifford circuit suffering from a Pauli
  error configuration $p$ and a loss configuration with readout $r$.
  Suppose that $r$ has Pauli envelope $E$.
  Let $\textsf{Dec}$ be a deterministic decoder that maps detector
  patterns and loss-resolving readouts to logical
  observable predictions.

  If $\textsf{Dec}(\cdot, r)$ correctly decodes all
  detector-observable pairs in
  $\mathcal{S}(p \oplus E,\emptyset)$, i.e.,
  for all $(\mathcal{D}, \mathcal{O}) \in \mathcal{S}(
  p \oplus E, \emptyset)$,
  we have $\textsf{Dec}(\mathcal{D}, r) = \mathcal{O}$,
  then $\textsf{Dec}(\cdot, r)$ also correctly decodes all
  detector-observable pairs in
  $\mathcal{S}(p, \mathcal{L}(r))$.
\end{lemma}

\begin{proof}
  Since $r$ has Pauli envelope $E$, by
  definition we have
  \[\mathcal{S}(p, \mathcal{L}(r)) \subseteq
  \mathcal{S}(p \oplus E,\emptyset).\]
  Therefore, if $\textsf{Dec}(\cdot, r)$ correctly decodes all pairs
  in $\mathcal{S}(p \oplus E,\emptyset)$, it also
  correctly decodes all pairs in $\mathcal{S}(p,
  \mathcal{L}(r))$.
\end{proof}

Next, we relate the logical error rate with atom loss to the
failure probability with the Pauli envelope.
We first extend the definition of logical error rate to include atom loss.
\begin{definition}[Logical Error Rate with Atom Loss]
  Let $\mathcal{C}$ be a Clifford circuit with the
  following error model:
  \begin{itemize}
    \item Pauli error configurations are sampled from distribution
      $\mathcal{P}_p$.
    \item Loss configurations are sampled from distribution $\mathcal{P}_l$.
  \end{itemize}
  The logical error rate $P_{\mathrm{LER}}(\mathcal{C},
    \mathcal{P}_l, \mathcal{P}_p,
  \textsf{Dec})$ is defined as
  \begin{equation}
    P_{\mathrm{LER}}=
    \Pr_{\substack{l \sim \mathcal{P}_l \\ p
    \sim \mathcal{P}_p}}
    \left[\textsf{Dec}(\mathcal{D}(p,l), \mathcal{R}(l)) \neq
    \mathcal{O}(p,l)\right]
  \end{equation}
\end{definition}

We further define the failure probability with Pauli envelope.
\begin{definition}[Failure Probability with Pauli Envelope]
  Let $\mathcal{C}$ be a Clifford circuit with the
  following error model:
  \begin{itemize}
    \item Pauli error configurations are sampled from distribution
      $\mathcal{P}_p$.
    \item Loss configurations are sampled from distribution
      $\mathcal{P}_l$. However, this loss configuration is not
      applied directly to the
      circuit. Instead, we compute the loss-resolving readout $r=\mathcal{R}(l)$
      and find the corresponding Pauli envelope $E(r)$.
  \end{itemize}
  The failure probability
  $P_{\mathrm{fail}}(\mathcal{C}, \mathcal{P}_l, \mathcal{P}_p,
  \textsf{Dec})$ is defined as
  \begin{equation}
    P_{\mathrm{fail}}=
    \Pr_{\substack{l \sim \mathcal{P}_l
    \\ p \sim \mathcal{P}_p}}\bigl[\exists
      (\mathcal{D}, \mathcal{O}) \in \mathcal{S}(p
      \oplus E(r), \emptyset) :
    \textsf{Dec}(\mathcal{D}, r) \neq \mathcal{O}\bigr].
  \end{equation}
\end{definition}
Notice that the failure probability does
not average over $E(r)$,
but counts the entire $E(r)$ as a failure if any error in $E(r)$ causes
a decoding failure.
\begin{restatable}[Logical Error Rate is Upper-Bounded by Failure
  Probability]{theorem}{thmeffectivedistance}
  \label{thm:effective-distance}
  \[P_{\mathrm{LER}}(\mathcal{C}, \mathcal{P}_l, \mathcal{P}_p,
    \textsf{Dec}) \leq
    P_{\mathrm{fail}}(\mathcal{C}, \mathcal{P}_l, \mathcal{P}_p,
  \textsf{Dec}).\]
\end{restatable}

\begin{proof}
  See \zcref{app:proof-effective-distance}.
\end{proof}

The key insight is that by \zcref{lem:sufficiency}, any decoding failure
with loss can be traced back to a failure in decoding the Pauli
envelope $p \oplus E$ without loss. This establishes that the logical
error rate is bounded by the failure probability of decoding the
envelope errors.

\subsection{Construction of Pauli Envelope for Atom Loss}

The preceding results show that all performance guarantees for the
Pauli envelope
will automatically transfer to atom loss errors.
A natural question arises: can we construct a Pauli envelope
with good decoding properties, specifically one that is linear and
low-weight? We show below that the answer is affirmative.

Without loss of generality, assume $\mathcal{C}$ is generated by $CZ$, $H$,
and $S$ gates. For a loss configuration
containing a single atom loss
location, the Pauli envelope $E$ can be constructed as follows.

\begin{restatable}[Pauli Envelope of Atom Loss]{lemma}{lemdetectorequivalence}
  \label{lem:detector-equivalence}
  Let $\mathcal{C}$ be a Clifford circuit and $l$ be an atom loss
  configuration
  corresponding to a single atom loss location. Then the following
  construction yields a Pauli envelope $E$ for $l$.

  Let $\Pi_i = \{I, X, Y, Z\}$ denote the set of all single-qubit Pauli
  operators at a specific space-time location $i$.
  For two such sets, we define
  $\Pi_i \oplus \Pi_j = \{e_i \oplus e_j : e_i \in \Pi_i, e_j \in \Pi_j\}$
  as the set of all pairwise compositions.
  We identify the relevant locations:
  (1) the loss location itself,
  (2) immediately after each subsequent Hadamard gate acting on the
  lost atom, and
  (3) just before measuring the lost atom.
  The Pauli envelope is then $E = \Pi_1 \oplus \Pi_2 \oplus \cdots
  \oplus \Pi_k$,
  where $\Pi_1, \ldots, \Pi_k$ are the Pauli sets at these relevant locations.
\end{restatable}
\begin{proof}
  See \zcref{app:proof-detector-equivalence}.
\end{proof}

The key insight of the proof is that a lost atom can be modeled as if
reset gates are inserted after each gate following the loss event.
Since $CZ$ and $S$ gates preserve $\ket{0}$, redundant resets between
consecutive $CZ$ and $S$ gates are not required.

This construction aligns with previous analysis of atom
loss~\cite{yu2024processing,yu2025locating} when there is only one
atom loss. However, this construction naturally generalizes to
multiple atom losses through the following linearity properties.

\begin{restatable}[Linearity of Pauli Envelope]{lemma}{lemlossresolvingreadout}
  \label{lem:loss-resolving-readout}
  (a) Consider a loss configuration $l = l_1 \oplus l_2 \oplus \cdots
  \oplus l_k$, where each $l_i$ is a single-atom loss configuration.
  The following composition yields a Pauli envelope for $l$:
  $E(l) = \bigoplus_{i=1}^{k} E(l_i)$.

  (b) For a readout $r$ corresponding to a single atom loss with
  multiple possible loss configurations
  $\mathcal{L}(r) = \{l^{(1)}, l^{(2)}, \ldots\}$,
  the following union yields a Pauli envelope for $r$:
  $E(r) = \bigcup_{l \in \mathcal{L}(r)} E(l)$.

  (c) Let $r^{(m)}$ denote the readout in which only atom $m$ is
  reported as lost. For a readout $r$, the following composition over
  atoms reported as lost yields a Pauli envelope for $r$:
  $E(r) = \bigoplus_{m : r_m = 1} E(r^{(m)})$.
\end{restatable}
\begin{proof}
  See \zcref{app:proof-loss-resolving}.
\end{proof}

This linearity ensures that the Pauli envelope for any combination of atom
losses can be constructed from single-atom-loss envelopes, avoiding the need
to analyze exponentially many multi-loss scenarios directly or to reason about
the nonlinearity of atom loss.

In the following, we mainly use $CNOT$ circuits for simplicity.
To apply the results to $CNOT$ circuits, we can view each $CNOT$
as compiled into $CZ$ and $H$ gates. Thus, we need to add
Pauli errors whenever a qubit changes its role from a target qubit
to a control qubit or vice versa.

We explicitly construct the Pauli envelope for a data atom loss
in a standard syndrome extraction circuit for the surface code.
\begin{example}[Data qubit loss]
  \zcref{fig:data_qubit_loss} illustrates a loss configuration
  with a data-qubit loss. Assume we measure the data qubit after the
  four CNOTs and it is reported as lost. The loss may have occurred
  at any of the five locations $1$--$5$, so $\mathcal{L}(r)$ contains
  five loss configurations. Applying \zcref{lem:detector-equivalence}
  to each of them and taking the union as in
  \zcref{lem:loss-resolving-readout}(b), the readout $r$ has Pauli
  envelope
  \begin{multline*}
    E = (\Pi_1 \oplus \Pi_2 \oplus \Pi_4 \oplus \Pi_5)
    \cup (\Pi_2 \oplus \Pi_4 \oplus \Pi_5) \\
    \cup (\Pi_3 \oplus \Pi_4 \oplus \Pi_5)
    \cup (\Pi_4 \oplus \Pi_5) \cup \Pi_5,
  \end{multline*}
  where each $\Pi_i$ denotes the set $\{I, X, Y, Z\}$ at the numbered
  location $i$ in \zcref{fig:data_qubit_loss}.
\end{example}
\begin{figure}[htbp]
  \centering
  \includegraphics[width=0.9\columnwidth]{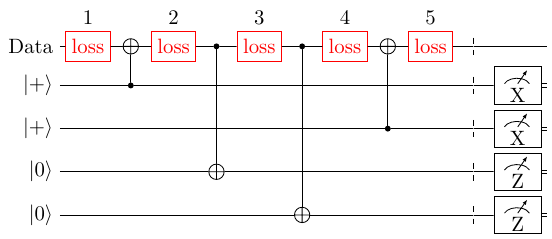}
  \caption{Surface code syndrome extraction circuit with a data-qubit
  loss. Red boxes mark the candidate loss locations $1$--$5$.}
  \label{fig:data_qubit_loss}
\end{figure}




\section{Loss-Replenishing Syndrome Extraction for Surface Codes}
\label{sec:surface-codes}

As discussed in \zcref{sec:background}, specialized syndrome
extraction circuits are needed to localize the effects of atom loss in
both time and space while satisfying hardware restrictions. In
particular, loss-resolving readout can only be applied during
measurement, which inevitably collapses the logical information.

In this section, we review the \swapSC{}, proposed by
\cite{ghosh2013understanding,suchara2015leakage} and adopted by
\cite{chow2024circuit,perrin2024quantum,baranes2025leveraging}, which
handles atom loss under these constraints. We then use the Pauli
envelope framework to analyze its limitations and introduce a new syndrome
extraction circuit, the \midswapSC{}. With the \envelopeMLEDecoder{},
it doubles the effective distance against atom loss compared to the
\swapSC{}, without additional hardware overhead.

\subsection{SWAP Syndrome Extraction}
The \swapSC{} addresses atom loss on data qubits by exchanging the
roles of ancilla and data qubits after each syndrome extraction
round. As shown in
\zcref{fig:swap_surface_code}, the final $CNOT$ interactions (ticks 4 and 8)
are reversed, followed by atom shuttling. After the measurements, a
classically controlled Pauli gate is applied to the
data qubit, which can be implemented in software. This ensures
that each atom is measured and
replenished within $8$ $CNOT$ gates after initialization.

During the data phase (the last $4$ ticks), the $CNOT$ direction alternates
between control and target. According to
\zcref{lem:detector-equivalence}, regardless of whether the atom loss
occurs during or before the data phase, there will always be data
errors on the data qubits.

Furthermore, if
atom loss occurs during the ancilla phase, it will create a
hook error originating from the ancilla qubit because a Pauli error
needs to be inserted at the position of the loss.

\zcref{fig:time_line_swap} summarizes this observation about the life
cycle: an atom loss early in the
ancilla phase of the \swapSC{} leads to \textbf{both a hook error and
a data error} in
the syndrome extraction circuit, making it less robust against atom loss.

In \zcref{app:loss-edge-correspondence-swap}, we show that this
widespread error limits the \swapSC{} to $d_{\mathrm{loss}}\sim d/2$.

\begin{figure}[htbp]
  \centering
  \captionsetup[subfigure]{position=top,justification=raggedright,singlelinecheck=false,skip=2pt}
  \begin{subfigure}{\columnwidth}
    \subcaption{}
    \centering
    \includegraphics[width=0.66\linewidth]{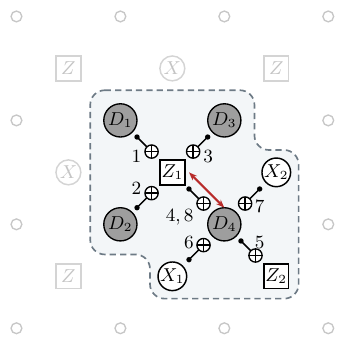}
  \end{subfigure}
  \par\smallskip
  \begin{subfigure}{\columnwidth}
    \subcaption{}
    \centering
    \includegraphics[width=0.95\linewidth]{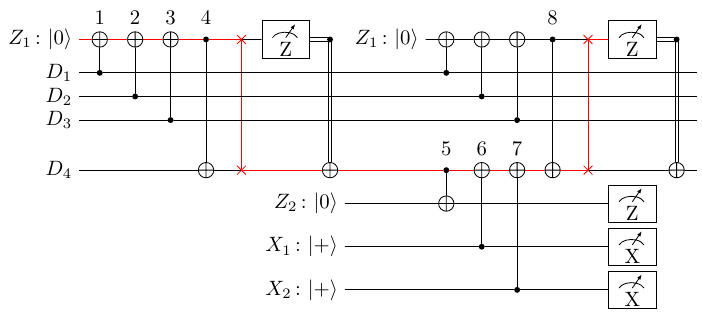}
  \end{subfigure}
  \par\smallskip
  \begin{subfigure}{\columnwidth}
    \subcaption{}
    \centering
    \includegraphics[width=0.95\linewidth]{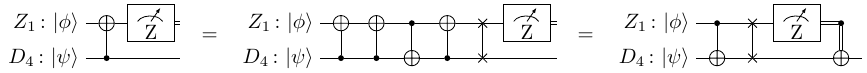}
  \end{subfigure}
  \caption{\swapSC{}. (a) Spatial layout: $D_1,\ldots,D_4$ are data
    qubits, and $X_1,X_2,Z_1,Z_2$ are stabilizer ancillas; names
    identify the same sites and circuit wires in (a) and (b). Numbers
    label the tracked atom's interactions over two rounds, and the red
    double arrow marks atom shuttling after ticks 4 and 8. The dashed
    boundary selects the qubits shown in (b); $CNOT$ gates involving an
    outside qubit are omitted.
    (b) Circuit over two consecutive rounds. The red line tracks one
    physical atom: it serves as the ancilla $Z_1$ for ticks 1--4, is
    shuttled into the data qubit $D_4$ for ticks 5--8, and is shuttled
    back and measured after tick 8. The lower three wires are the other
    stabilizer ancillas adjacent to $D_4$.
  (c) Circuit identity validating the swap in (b).}
  \label{fig:swap_surface_code}
\end{figure}

\begin{figure}[htbp]
  \centering
  \includegraphics[width=0.9\columnwidth]{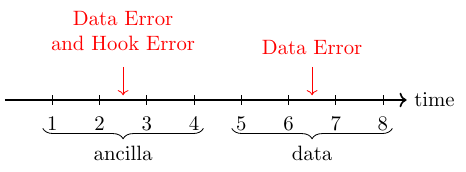}
  \caption{Timeline of an atom in the \swapSC{}.
    The ticks correspond to the gate labels in \zcref{fig:swap_surface_code}.
  The red arrow marks an atom loss between ticks $2$ and $3$.}
  \label{fig:time_line_swap}
\end{figure}

\subsection{Mid-SWAP Syndrome Extraction}

To address the possibility of creating both a hook error and a data
error from a single atom loss event in the \swapSC{}, we introduce
the \midswapSC{}, which performs atom shuttling during each syndrome
extraction round rather than at the end of each round. This allows
each newly replenished atom to serve as a data
qubit first and then transition to an ancilla qubit.

As illustrated in \zcref{fig:mid_swap_surface_code}, the \midswapSC{}
modifies a standard surface code syndrome extraction circuit by inserting an
atom-shuttling step immediately after the first $CNOT$ gate in each round,
while all other $CNOT$ gates remain unchanged.

\zcref{fig:time_line_mid_swap} shows the life cycle of a qubit in
the \midswapSC{}.

\begin{figure}[htbp]
  \centering
  \captionsetup[subfigure]{position=top,justification=raggedright,singlelinecheck=false,skip=2pt}
  \begin{subfigure}{\columnwidth}
    \subcaption{}
    \centering
    \includegraphics[width=0.66\linewidth]{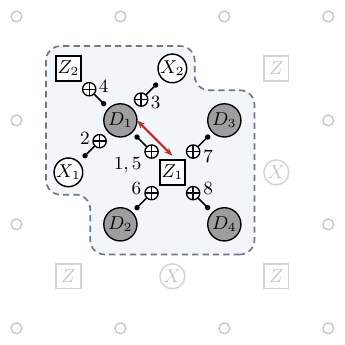}
  \end{subfigure}
  \par\smallskip
  \begin{subfigure}{\columnwidth}
    \subcaption{}
    \centering
    \includegraphics[width=0.95\linewidth]{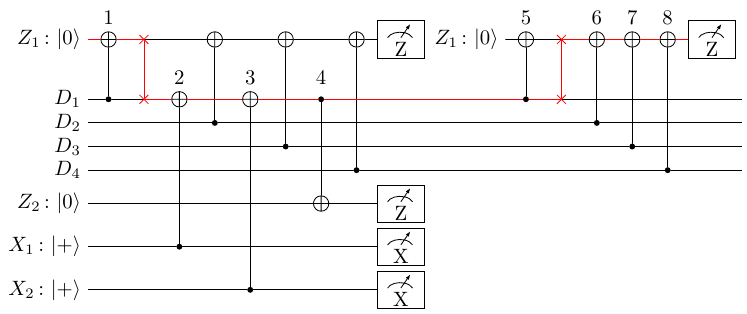}
  \end{subfigure}
  \par\smallskip
  \begin{subfigure}{\columnwidth}
    \subcaption{}
    \centering
    \includegraphics[width=0.4\linewidth]{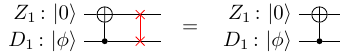}
  \end{subfigure}
  \caption{\midswapSC{}. Notation and conventions follow
    \zcref{fig:swap_surface_code}.
    (a) Spatial layout. Atom shuttling (red double arrow) now occurs
    after ticks 1 and 5.
    (b) Circuit over two consecutive rounds. The tracked atom (red
    line) serves as the ancilla $Z_1$ at tick 1, is shuttled into the
    data qubit $D_1$ for ticks 2--4, and is shuttled back after tick 5
    to serve as the ancilla for ticks 6--8.
  (c) Circuit identity validating the swap in (b).}
  \label{fig:mid_swap_surface_code}
\end{figure}

\begin{figure}[htbp]
  \centering
  \includegraphics[width=0.9\columnwidth]{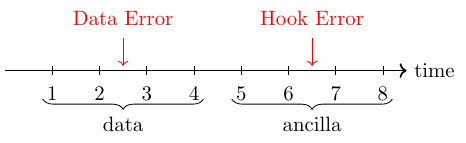}
  \caption{Timeline of an atom in the \midswapSC{}. The ticks
  correspond to the gate labels in \zcref{fig:mid_swap_surface_code}.}
  \label{fig:time_line_mid_swap}
\end{figure}

If an atom loss occurs during the data phase, it will cause a data
error on that data qubit but will not create hook errors during the
ancilla phase because the entire ancilla phase is skipped. Moreover,
because an ancilla qubit is always a target or always a control of the $CNOT$
gates, \zcref{lem:detector-equivalence} implies that there is no need to
insert any Pauli error on the ancilla if it is lost before the beginning of
the ancilla phase, except for a possible measurement error.

If an atom loss occurs during the ancilla phase, it will cause
a hook error but not a data error because only gates after the loss
are affected.

This ensures that in \midswapSC{}, an atom loss can only cause
\textbf{either a hook error or a data error}, but not both.
\zcref{lem:z-shaped-loss} formalizes this observation.
Focusing on physical $Z$ errors detected by $X$-stabilizer measurements
and projecting each detector onto the $X$-$Y$ plane,
a $Z$-ancilla loss in the \midswapSC{} triggers at most one edge
forming a $Z$-shaped path as illustrated by $S_1$, $S_2$, and $S_3$ in
\zcref{fig:z-shaped-loss}, while an $X$-ancilla loss triggers at most
one edge at a fixed position.
This
structure is crucial for the
decoder analysis in the following sections. After introducing the
\envelopeMLEDecoder{} in the next section, we prove the optimality of
\midswapSC{}.

\begin{figure}[htbp]
  \centering
  \includegraphics[width=0.9\columnwidth]{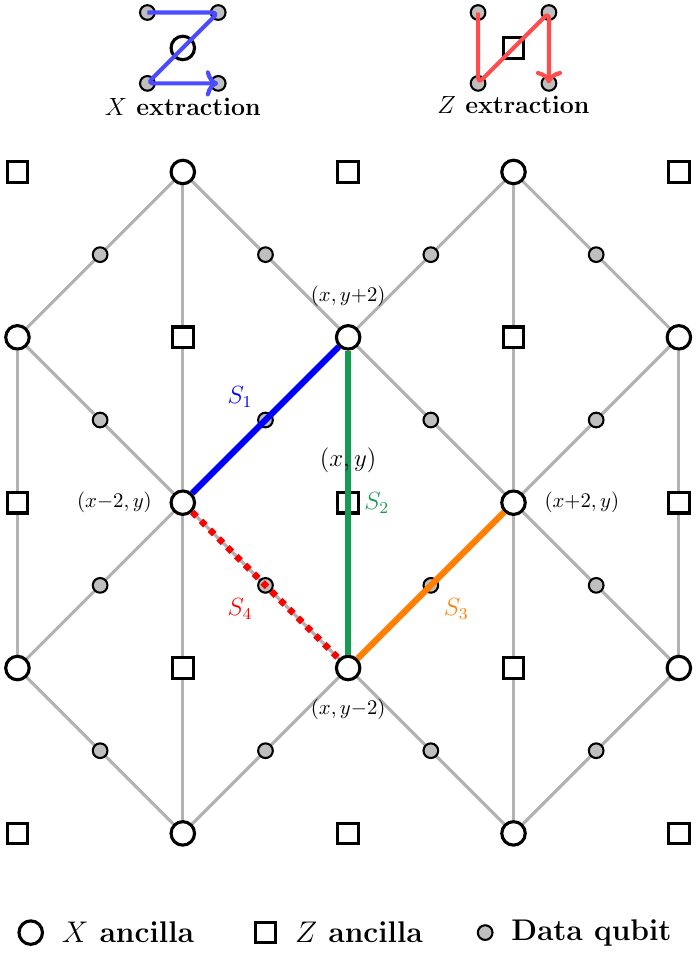}
  \caption{Detector patterns induced by a single atom loss in the \midswapSC{}.
    Each edge represents an error mechanism that can trigger the two
    endpoint detectors. $X$-detectors correspond to $X$ stabilizers and are
    triggered by $Z$ errors; similarly, $Z$-detectors correspond to $Z$
    stabilizers and are triggered by $X$ errors.
    An atom loss detected at
    position $(x,y)$ (during $Z$-ancilla measurements) triggers at
    most one of three
    detector pairs:
    $S_1$, $S_2$, and $S_3$, which together form a $Z$-shaped path in the
    matching graph. An atom loss detected at position $(x,y-2)$ (during
  $X$-ancilla measurements) can only trigger $S_4$.}
  \label{fig:z-shaped-loss}
\end{figure}

\begin{lemma}[Exclusive Loss-induced Patterns in \midswapSC{}]
  \label{lem:z-shaped-loss}

  Consider a standard surface code with data qubits and ancilla
  qubits mapped to the grid as shown in \zcref{fig:z-shaped-loss}.
  Assume the logical $Z$ observable is defined horizontally, and
  stabilizers are extracted in the order shown above the figure.
  We focus on logical $Z$ errors; logical $X$ errors can be
  handled symmetrically.

  In the matching graph for logical $Z$ decoding, only $X$-stabilizer
  measurements (which detect $Z$-type data errors) are relevant.
  There are two types of ancilla loss events to consider:
  loss at a $Z$-stabilizer ancilla position and loss at an
  $X$-stabilizer ancilla position.

  For a $Z$-stabilizer ancilla loss detected at position $(x,y)$,
  the loss triggers at most one of the following $Z$-error detector
  pairs (projected onto the $(x,y)$ plane), which together form a
  $Z$-shaped path:
  \begin{enumerate}
    \item $S_1 = ((x-2,y), (x,y+2))$: a diagonal edge from left to top.
    \item $S_2 = ((x,y+2), (x,y-2))$: a vertical edge from top to bottom.
    \item $S_3 = ((x,y-2), (x+2,y))$: a diagonal edge from bottom to
      right, parallel to $S_1$.
  \end{enumerate}

  For an $X$-stabilizer ancilla loss detected at position $(x,y-2)$,
  the loss triggers $S_4=((x-2,y), (x,y-2))$ or none of these detector pairs.

\end{lemma}
\begin{proof}
  This follows from enumerating all possible errors; details are
  provided in \zcref{app:loss-edge-correspondence-midswap}.
\end{proof}

\section{\envelopeMLEDecoder{}}
\label{sec:mle-decoder}

We introduce the \envelopeMLEDecoder{}, formulated as a Mixed-Integer Linear
Program (MILP), which jointly optimizes over Pauli envelope selections
and Pauli errors to find a minimum-weight error configuration
consistent with the observed detector outcomes. Throughout this section, we
assume loss-resolving readouts are perfect; extensions to noisy
readouts are discussed in \zcref{app:proof-measurement-errors}.

\begin{algorithm}[\envelopeMLEDecoder{}]
  \label{alg:mle-decoder}
  Before receiving the detector patterns and loss-resolving readouts,
  the decoder first performs the
  following steps as preprocessing:
  \begin{enumerate}
    \item \textbf{Find the Pauli envelope for each loss-resolving readout:}
      For each measured atom $m$, find the Pauli envelope $E_m$ of the
      single-loss readout $r^{(m)}$ (in which only atom $m$ is reported as lost)
      according to \zcref{lem:detector-equivalence,lem:loss-resolving-readout}.

    \item \textbf{Assign variables for loss patterns:}
      For each pair $(\mathcal{D}_{m,j},
      \mathcal{O}_{m,j})\in \mathcal{S}(E_m,\emptyset)$, assign a
      binary variable $x_{m,j} \in \{0,
      1\}$, indicating whether the loss pattern is selected ($x_{m,j} = 1$)
      or not ($x_{m,j} = 0$). Here, the index $m$ represents the
      measured atom, while $j$ indexes the distinct detector-observable
      pairs in $\mathcal{S}(E_m,\emptyset)$.

    \item \textbf{Assign variables for Pauli errors:}
      For each Pauli error mechanism $i$ in the circuit, assign a binary
      variable $y_i \in \{0, 1\}$, indicating whether the Pauli error
      mechanism is selected ($y_i = 1$) or not ($y_i = 0$).
  \end{enumerate}

  After receiving the detector pattern and loss-resolving readout,
  the decoder performs the
  following steps:
  \begin{enumerate}
    \item \textbf{Enforce constraints:}

      (i)~\emph{Detector satisfaction:}
      For each detector $u$, enforce
      \begin{equation}
        \sum_i D_{i,u} y_i + \sum_{m,j} D_{m,j,u} x_{m,j} \equiv
        s_u \pmod{2},
      \end{equation}
      where $D_{i,u} \in \{0,1\}$ indicates whether Pauli error
      $i$ flips detector $u$,
      $D_{m,j,u} \in \{0,1\}$ indicates whether loss pattern
      $(m,j)$ flips detector $u$,
      and $s_u \in \{0,1\}$ is the observed value of detector $u$.

      (ii)~\emph{Loss exclusivity:} For each readout bit
      $r_m$, where $r_m=1$ indicates that atom $m$ is reported as lost, enforce
      \begin{equation}
        \sum_j x_{m,j} = r_m,
      \end{equation}
      ensuring exactly one detector pattern is selected per loss event.

    \item \textbf{Minimize the total weight:}
      Subject to the above constraints, solve the following MILP:
      \begin{equation}
        \min_{x,y} \sum_{i} w_i y_i + \sum_{m,j} \lambda_{m,j} x_{m,j},
      \end{equation}
      where $w_i=-\log \frac{p_i}{1-p_i}$ is the log-likelihood ratio
      of Pauli error $i$, and
      $\lambda_{m,j}=-\log \Pr[(\mathcal{D}_{m,j},\mathcal{O}_{m,j})
      \mid r_m=1]$ is the negative log-probability of pattern $j$
      among the patterns of atom $m$, computed from the loss
      probabilities of the locations in $\mathcal{L}(r^{(m)})$ with
      the Pauli operators of the envelope weighted uniformly.

    \item \textbf{Predict logical observable:}
      Let $x_{m,j}^*,y_i^*$ be the optimal solutions to the MILP.
      Calculate the predicted logical flip:
      \begin{equation}
        \mathcal{O}_{pred} = \bigoplus_i y_i^* \mathcal{O}_i \oplus
        \bigoplus_{m,j} x_{m,j}^* \mathcal{O}_{m,j}
      \end{equation}
      where $\mathcal{O}_i$ is the observable flip caused by Pauli error $i$.
  \end{enumerate}
\end{algorithm}

The decoder returns the most likely joint configuration of Pauli
errors and envelope patterns in the independent Pauli-envelope model
specified by these weights. The envelope guarantees support inclusion,
not equality with the physical loss distribution. As is standard for
distance proofs~\cite{dennis2002topological,fowler2012surface}, the
analysis below takes the limit $p\to0$. If all
Pauli error probabilities are proportional to $p$, then
$w_i=\log(1/p)+O(1)$, so for a fixed circuit one additional Pauli
error outweighs the bounded envelope-selection weights
$\lambda_{m,j}$ and the $O(1)$ differences among the $w_i$. The
decoder then minimizes the number of Pauli errors subject to the
envelope constraints, and the proofs set
$\lambda_{m,j}=0$ and $w_i=w$. At practical error rates the bounded
terms still affect decoding, so the implementation used in
\zcref{sec:evaluation} keeps the full objective.

\begin{restatable}[Optimality of
  \envelopeMLEDecoder{}]{lemma}{lemoptimalityofmledecoder}
  \label{lem:optimality-of-mle-decoder}
  In the limit $p\to0$, the \envelopeMLEDecoder{}
  finds a solution with Pauli weight no greater
  than that of the actual error configuration.
\end{restatable}
\begin{proof}
  See \zcref{app:proof-optimality-mle}.
\end{proof}

The optimality follows from observing that by
\zcref{lem:detector-equivalence}, the actual error configuration is
always feasible for the MILP. We now show the decoder performance on
the \midswapSC{}, where the localized Pauli envelope structure enables
stronger guarantees.

\begin{restatable}[MLE Decoder Achieves Optimal
  Distance]{theorem}{thmmleoptimaldistance}
  \label{thm:mle-optimal-distance}
  In the limit $p\to0$, the \envelopeMLEDecoder{}
  achieves optimal distance for the \midswapSC{}.
  Specifically, the decoder correctly decodes whenever
  \begin{equation}
    n_l + 2n_p < d,
  \end{equation}
  where $n_l$ is the number of atom losses, $n_p$ is the number of
  Pauli errors, and $d$ is the code distance.
\end{restatable}
\begin{proof}
  See \zcref{app:proof-mle-optimal-distance}.
\end{proof}

The proof analyzes the symmetric difference between the actual error and
decoder output. By the optimality of the MLE decoder, the Pauli component
has weight at most $2n_p w$. Forming an undetectable logical error with
$n_l$ loss pattern mismatches requires Pauli weight at least
$(d-n_l)w$. Combining these bounds yields the condition
$2n_p + n_l \geq d$ for decoder failure.

This bound is tight. When $n_l + 2n_p = d$, consider a length-$d$ logical
$X$ string. Suppose the first $n_l$ data qubits on the string are lost, so
the remaining segment has length $d-n_l=2n_p$. Now compare two fault
patterns: (i) apply the logical $X$ and add $X$ errors on any $n_p$ of the
remaining qubits; (ii) apply no logical operator and instead add $X$ errors
on the complementary $n_p$ remaining qubits. After tracing out the lost
qubits, the two cases yield the same physical state and identical
loss-resolving readouts, so a decoder cannot distinguish them and must fail
on at least one, inducing a logical error.

\zcref{thm:mle-optimal-distance} by itself does not guarantee fault
tolerance. We next convert this combinatorial condition into a
loss-distance scaling statement for the logical error rate,
showing that the logical error rate can be made arbitrarily small by
increasing the code distance.



\begin{restatable}[Loss Distance of
  \envelopeMLEDecoder{}]{theorem}{thmfailureprobability}
  \label{thm:failure-probability}
  For the \midswapSC{} with code distance $d$, subject to independent
  atom loss and Pauli errors and decoded by the
  \envelopeMLEDecoder{}, the loss distance defined in
  \zcref{def:effective-distance} is
  $d_{\mathrm{loss}}\sim d$.
\end{restatable}
\begin{proof}
  See \zcref{app:proof-failure-probability}.
\end{proof}

The proof follows from standard
lattice counting arguments. By defining $\bar{p} := \max(p_{\mathrm{loss}},
\sqrt{p_{\mathrm{pauli}}})$, we unify the error rates and show the sum is
dominated by minimum-weight configurations, yielding the stated scaling.

In contrast, the \averageMLEDecoder{} proposed in
\cite{baranes2025leveraging} constructs a detector error model by
enumerating all possible detector patterns induced by individual atom
losses and combining them with the Pauli detector error model as if
all errors are independent. We attribute its suboptimal performance to
two limitations that our \envelopeMLEDecoder{} avoids:

\begin{itemize}
  \item \textbf{No exclusivity constraint:} The \averageMLEDecoder{} does not
    enforce that each atom can be lost at most once. As a result,
    certain detector patterns caused by Pauli errors may be
    incorrectly explained as loss errors, even when no actual loss
    configuration can produce them.
    This miscorrection limits the decoder to
    $d_{\mathrm{loss}}\sim d/2$, as we demonstrate in
    \zcref{app:harvard_decoder}.

  \item \textbf{Ignoring nonlinear effects:} The
    \averageMLEDecoder{} enumerates
    detector patterns from individual losses and combines them
    linearly. However, multiple simultaneous losses can produce
    detector patterns that cannot be decomposed into the sum of
    individual loss
    patterns. When such nonlinear patterns occur, the decoder must
    explain them using additional Pauli errors, degrading performance.
    This effect is particularly significant in the \swapSC{} as shown in
    \zcref{app:non_linearity_of_atom_loss}.
\end{itemize}

\section{Envelope-Matching Decoder}
\label{sec:envelope-matching}

Matching-based decoders offer efficiency and natural compatibility
with transversal logical circuits. Inspired by the exclusivity
constraint of the \envelopeMLEDecoder{}, we
propose the \envelopeMatchingDecoder{}, which achieves
$d_{\mathrm{loss}} \sim 2d/3$ for the \midswapSC{}. This exceeds the
$d_{\mathrm{loss}}\sim d/2$ achieved by previous
matching-based decoders and is higher than
the corresponding Pauli-error exponent $\sim d/2$.
Consequently, for comparable small physical error rates, atom-loss errors
are asymptotically negligible compared with Pauli errors.

\begin{algorithm}[Envelope-Matching Decoder]
  \label{alg:envelope-matching-decoder}
  \leavevmode
  \begin{enumerate}
    \item \textbf{Graph Initialization (for the $Z$-error sector):}
      Construct a standard matching graph $G=(V, \mathcal{E})$ where
      vertices $V$
      correspond to $X$-detectors and edges $\mathcal{E}$ correspond to
      Pauli errors, connecting the detectors that each error flips.
      Initialize the weight of every edge
      $e \in \mathcal{E}$ to $w_e = -\log \frac{p_e}{1-p_e}$ where $p_e$ is the
      probability of this edge being triggered. Let $w$ be the
      average weight of the edges in the matching graph.

    \item \textbf{Loss Processing:}
      For each atom $m$ with readout bit $r_m = 1$, find the Pauli
      envelope $E_m$ of $r^{(m)}$ as in \zcref{alg:mle-decoder}.
      Then, for each edge $e = (u, v)$, where $u$
      and $v$ can be triggered simultaneously by some Pauli errors in
      $E_m$, update the weight as follows:

      \textbf{Space-like edges} (detectors $u$ and $v$
      involving measurements on different ancilla qubits): $w_e
      \leftarrow 0.5 w$.

      \textbf{Time-like edges} (detectors $u$ and $v$
      only involve measurements on the same ancilla qubit): $w_e
      \leftarrow 0.25 w$.

    \item \textbf{Matching:}
      Run the MWPM decoder on the matching graph $G$ with the updated weights.
  \end{enumerate}
\end{algorithm}

By reweighting the affected edges to weights comparable to regular edges,
instead of reducing them to a constant weight (which is asymptotically
negligible compared with Pauli error weights), the
\envelopeMatchingDecoder{} approximately enforces the exclusivity
constraint: selecting multiple edges within the same envelope is discouraged.

\begin{restatable}[Failure Condition of
  \envelopeMatchingDecoder{}]{lemma}{lemfailureconditionmatching}
  \label{lem:failure-condition}
  For the \envelopeMatchingDecoder{} on the \midswapSC{}, decoding
  failure in the presence of $n_p$ Pauli errors and $n_l$ atom losses
  requires
  \begin{equation}
    2n_p + 1.5 n_l \geq d.
  \end{equation}
\end{restatable}
\begin{proof}
  The proof relies on the fact that after the reweighting, under the
  equal-edge-weight assumption $w_e=w$ used in
  \zcref{app:proof-failure-condition}, the reweighted matching graph has
  logical distance at least $(d-0.5n_l)w$ for $n_l$ losses, and each
  loss-induced detector pattern can be matched with total weight at most
  $0.5w$. See
  \zcref{app:proof-failure-condition} for details.
\end{proof}

With this lemma, we can show that the decoder achieves
$d_{\mathrm{loss}} \sim 2d/3$ for the \midswapSC{}.
\begin{restatable}[Loss Distance of
  \envelopeMatchingDecoder{}]{theorem}{thmreweightmatching}
  \label{thm:reweight-matching}
  For the \midswapSC{} with code distance $d$, subject to independent
  atom loss and Pauli errors, the loss distance is
  $d_{\mathrm{loss}}\sim 2d/3$.
\end{restatable}
\begin{proof}
  This follows from standard MWPM threshold analysis. See
  \zcref{app:proof-reweight-matching} for details.
\end{proof}

In contrast, the \marginalMatchingDecoder{}~\cite{gu2025fault,gu2024optimizing}
computes the marginal probability that each atom is lost given the
loss-resolving readout, adds depolarizing errors with strength
proportional to these marginal probabilities, and runs a standard MWPM
decoder. Like the \averageMLEDecoder{}, it sets the edges that could
be triggered to a constant weight, meaning no exclusivity constraint
is enforced,
limiting its loss distance to
$d_{\mathrm{loss}}\sim d/2$.

\section{Evaluation}
\label{sec:evaluation}
\begin{figure*}[htbp]
  \centering
  \includegraphics[width=0.9\textwidth]{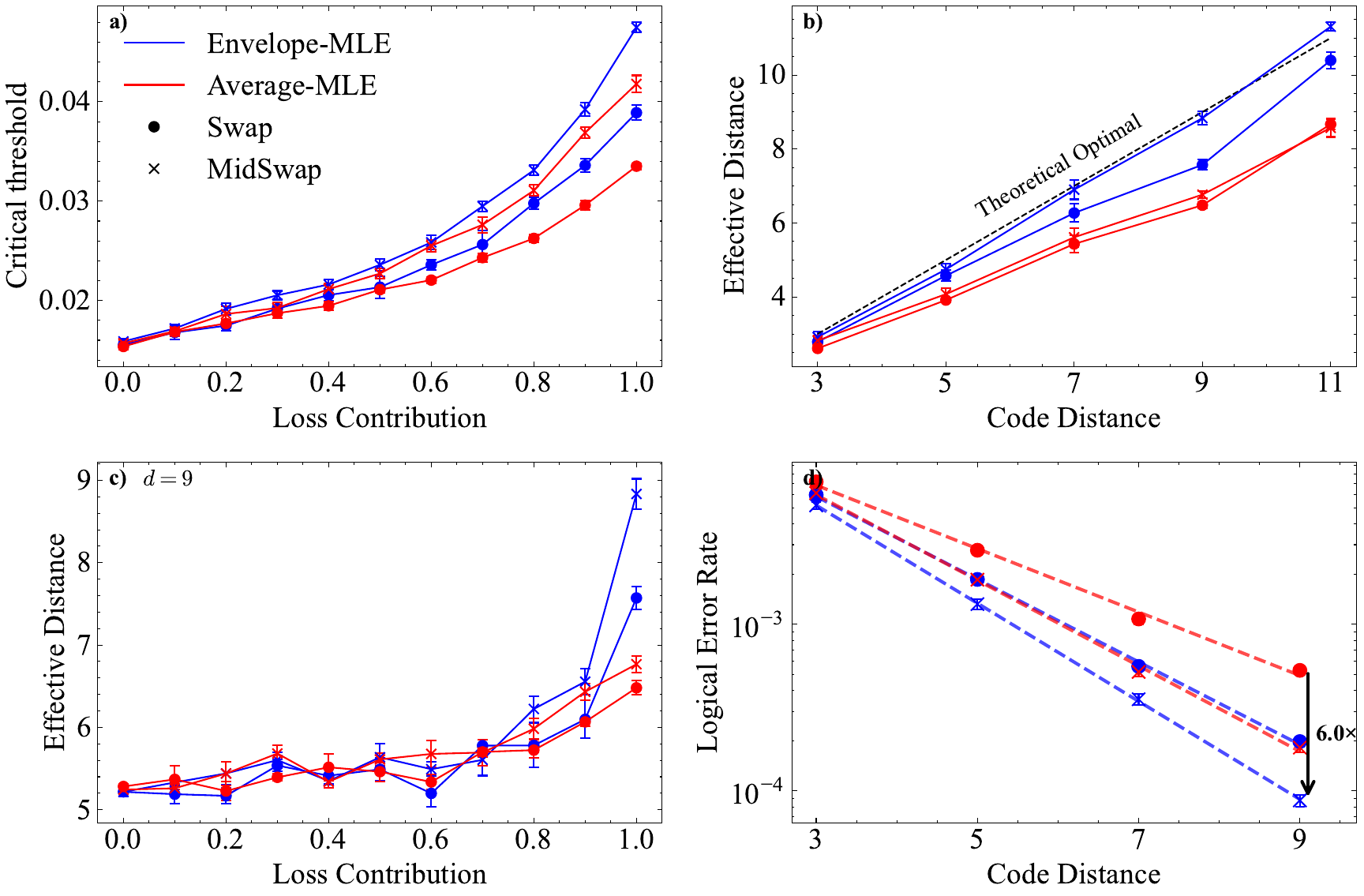}
  \caption{Comparison of the \envelopeMLEDecoder{} and the \averageMLEDecoder{}
    on \midswapSC{} and \swapSC{}: (a) threshold versus loss contribution $\eta$
    (ratio of atom loss to total error rate $p$); (b) effective distance versus
    code distance with $\eta=1$, up to $d=11$; (c) effective distance at $d=9$
    versus loss contribution; and (d) distance scaling at $\eta=90\%$ and
  total error rate $p=1.5\%$.}
  \label{fig:mle_vs_harvard_threshold}
\end{figure*}

\subsection{Experimental Setup}
\textbf{Error Model.}
We simulate a noise model combining circuit-level
depolarizing noise (probability $p_{\mathrm{pauli}}$) with atom loss
(probability $p_{\mathrm{loss}}$).
First, we apply standard circuit-level Pauli errors: single-qubit
gates and idling locations experience
single-qubit depolarizing noise, while two-qubit gates ($CNOT$)
experience two-qubit depolarizing noise, all with parameter
$p_{\mathrm{pauli}}$. Resets and measurements have bit-flip probability
$p_{\mathrm{pauli}}$.
Second, atom loss occurs after each operation. For two-qubit gates, loss occurs
independently on each qubit with probability $p_{\mathrm{loss}}/2$, so
a two-qubit gate causes at least one loss with probability approximately
$p_{\mathrm{loss}}$, consistent with previous work~\cite{baranes2025leveraging}.
Finally, loss-resolving readouts are also subject to
errors: with probability $p_{\mathrm{readout}}/2$, an existing atom is
misidentified as lost, and with probability $p_{\mathrm{readout}}/2$, a lost
atom is misidentified as existing and assigned a random 0/1 outcome.
The distance statements of
\zcref{sec:mle-decoder,sec:envelope-matching} assume perfect
loss-resolving readouts; here the decoders receive the noisy
readouts without modification, which tests their robustness to this
error (\zcref{app:proof-measurement-errors,thm:reweight-measurement-errors}
  give readout-error-aware extensions of the two decoders that we do
not use).

We parameterize the noise model using a total error rate $p$ and a
loss contribution $\eta \in [0,1]$, where
$p_{\mathrm{readout}}=p_{\mathrm{loss}} = p \cdot
\eta$ and $p_{\mathrm{pauli}} = p \cdot (1 - \eta)$.

Since most decoders calculate Pauli weights using $-\log
\frac{p_{\mathrm{pauli}}}{1-p_{\mathrm{pauli}}}$, which is undefined when
$p_{\mathrm{pauli}}=0$, we set
$p_{\mathrm{pauli}}=10^{-9}$ for $\eta=1$ to
avoid numerical issues.

\noindent\textbf{Software.}
We use Stim~\cite{gidney2021stim} for circuit simulation and sampling;
\zcref{app:sampling} describes how loss configurations are sampled and
converted into Stim circuits. For
decoding, we employ PyMatching~\cite{higgott2022pymatching} for minimum-weight
perfect matching
implementations (\envelopeMatchingDecoder{} and \marginalMatchingDecoder{}) and
Gurobi v13 for MILP-based decoders
(\averageMLEDecoder{} and \envelopeMLEDecoder{}), with a 60-second time
limit per shot; \zcref{app:solver_timeouts} reports more details.

\noindent\textbf{Simulation parameters.}
We perform $3d$ rounds of syndrome extraction for code distance $d$ to
minimize boundary effects. For \swapSC{} and \midswapSC{}, we
alternate the stabilizer order between rounds: odd rounds use the
CNOT order $1,2,3,4$ of \zcref{fig:z-shaped-loss} and even rounds
the reversed order $4,3,2,1$, because under a fixed order the data
qubits along one boundary row and one boundary column would never be
replenished; \zcref{app:ablation} explains this and quantifies the
effect of the alternation.
Hyperparameters for \envelopeMatchingDecoder{} are empirically
optimized as described in \zcref{app:best_hyperparams}.

\noindent\textbf{Data fitting.}
We extract effective distances using power-law fitting at low
physical error rates $P_{\mathrm{LER}}=\alpha
(p/p_{\mathrm{th}})^{d_{\mathrm{eff}}}$.
Following~\cite{baranes2025leveraging},
$d_{\mathrm{eff}}$ denotes the fitted exponent: for pure Pauli
noise in a fault-tolerant distance-$d$ circuit with odd $d$, the
expected asymptotic exponent is $(d+1)/2$; for pure loss, it depends
on the circuit and decoder.
We extract thresholds via critical-point fitting
following~\cite{wang2003confinement} near
the threshold regime
$P_{\mathrm{LER}}=a + bx + cx^2$ where $x=(p -
p_{\mathrm{th}}) \times d^{1/\nu}$. The thresholds in
\zcref{fig:mle_vs_harvard_threshold} and those of the matching-based
decoders are estimated
from $d \in \{5,7,9\}$; \zcref{app:solver_timeouts} shows the
underlying crossing for the \envelopeMLEDecoder{} on \midswapSC{}.

\subsection{Envelope-MLE Decoder vs Average-MLE Decoder}
We compare the \envelopeMLEDecoder{} against the \averageMLEDecoder{}
on the \swapSC{} and
\midswapSC{}. The teleportation circuit is excluded due to
substantially inferior performance (\zcref{app:teleportation_circuit}).

\zcref{fig:mle_vs_harvard_threshold} shows
that \envelopeMLEDecoder{} achieves higher thresholds and effective distances
than \averageMLEDecoder{} on both circuits, with \midswapSC{} outperforming
\swapSC{} in all cases. Combining \envelopeMLEDecoder{} with \midswapSC{}
improves the threshold from $3.35\%$ to $4.75\%$, a $42\%$ increase
over prior work. Changing only the circuit (\averageMLEDecoder{} on
\midswapSC{}) gives $4.18\%$ and changing only the decoder
(\envelopeMLEDecoder{} on \swapSC{}) gives $3.89\%$, so for the
threshold the two contributions are roughly additive. Specifically, at loss
contribution $\eta=1$, \envelopeMLEDecoder{} on \midswapSC{} achieves
$d_{\mathrm{eff}}=6.9\pm0.3$, $8.8\pm0.2$, and $11.3\pm0.1$ at
$d=7$, $9$, and $11$, matching the theoretical optimum
$d_{\mathrm{eff}}=d$. When the
total error rate is $p=1.5\%$ and loss contribution is $\eta=90\%$,
we achieve $6\times$ error suppression over previous work.
When extended to correlated atom loss (\zcref{app:correlated}),
the \envelopeMLEDecoder{} reveals that correlated loss is
\emph{easier} to correct than independent loss: the threshold
increases from $5.15\%$ to $7.82\%$ as the correlated fraction grows,
because correlated events supply more information about the
locations of the lost atoms.

\subsection{Envelope-MLE Decoder on Experimental Data}
\begin{figure}[htbp]
  \centering
  \includegraphics[width=0.9\columnwidth]{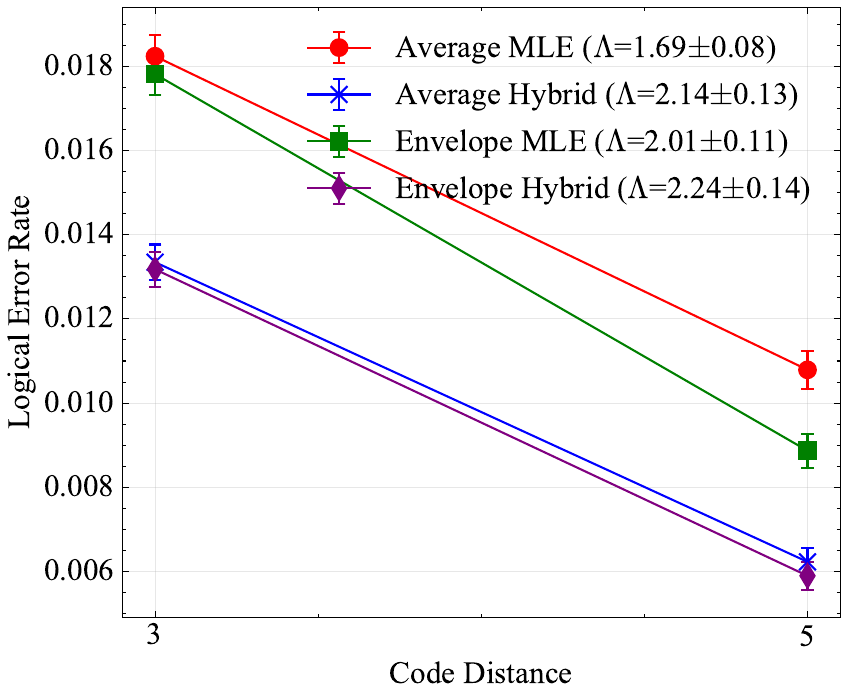}
  \caption{Per-round logical error rate and error suppression factor
    $\Lambda$ ($P_{\mathrm{LER}}^{d=3} / P_{\mathrm{LER}}^{d=5}$)
    evaluated on experimental data from Ref.~\cite{bluvstein2025architectural}.
    Hybrid refers to combining an MLE decoder with a machine-learning decoder.
    Replacing the \averageMLEDecoder{} with the \envelopeMLEDecoder{}
    improves $\Lambda$ both as a standalone decoder (from $1.69$ to $2.01$) and
  when combined with a machine-learning decoder (from $2.14$ to $2.24$).}
  \label{fig:harvard_data}
\end{figure}

We evaluate the \envelopeMLEDecoder{} on recent experimental data
reported by Ref.~\cite{bluvstein2025architectural}.
As shown in \zcref{fig:harvard_data}, replacing the \averageMLEDecoder{}
with the \envelopeMLEDecoder{} improves the error suppression factor
from $\Lambda = 1.69\pm 0.08$ to $\Lambda = 2.01\pm 0.11$.
When combined with a machine-learning decoder, the
\envelopeMLEDecoder{} further increases $\Lambda$ from $2.14\pm 0.13$
to $2.24\pm 0.14$, demonstrating that our framework provides
complementary gains on top of learned decoders.
\subsection{Envelope-Matching Decoder vs Marginal-Matching Decoder}
\begin{figure*}[htbp]
  \centering
  \includegraphics[width=\textwidth]{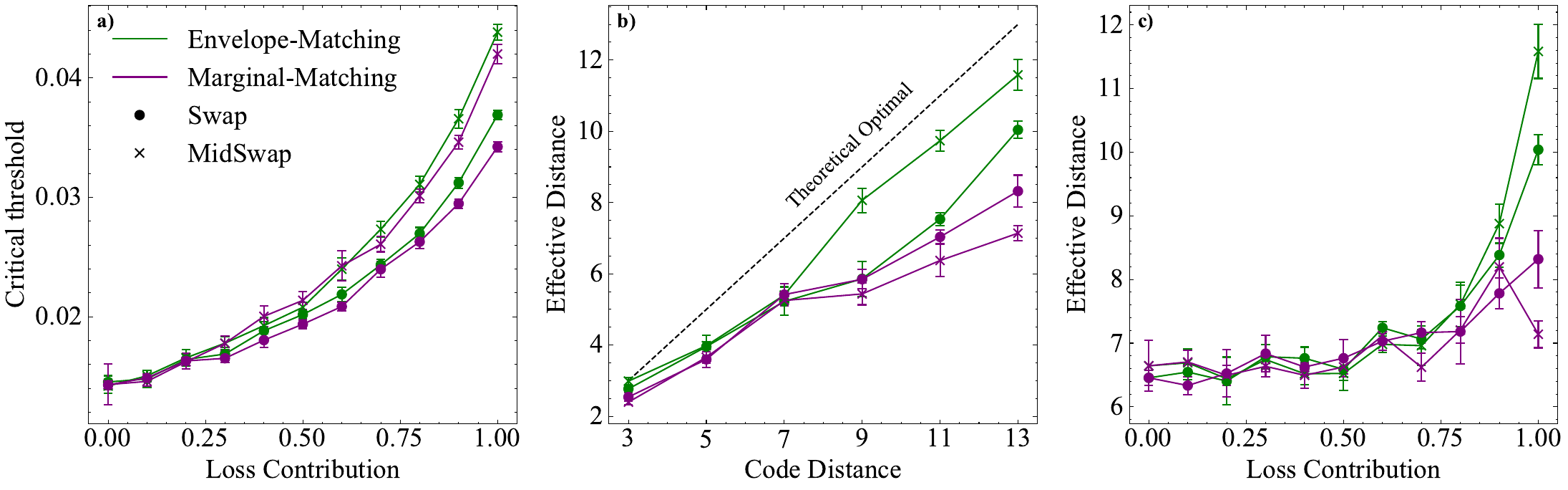}
  \caption{Comparison of the \envelopeMatchingDecoder{} and the
    \marginalMatchingDecoder{} on \midswapSC{} and \swapSC{}:
    (a) threshold versus loss contribution;
    (b) effective distance versus code distance with $\eta=1$;
  (c) effective distance at $d=13$ versus loss contribution.}
  \label{fig:matching_threshold}
\end{figure*}

We evaluate our scalable \envelopeMatchingDecoder{} against the
\marginalMatchingDecoder{} on \swapSC{} and \midswapSC{}.
The \marginalMatchingDecoder{} was originally designed for erasure errors
with delayed checks and is therefore not fully aware of the error model of
atom loss; nevertheless, we adopt it as our baseline since it represents
the best available matching-based decoder from prior work.
Concretely, we supply \marginalMatchingDecoder{} with an error model in
which each atom loss is replaced by an erasure error with full depolarization,
as this is the closest approximation to atom loss within its supported error
model.
\zcref{fig:matching_threshold} shows that \envelopeMatchingDecoder{}
achieves substantially higher
thresholds and effective distances on both circuits, with \midswapSC{}
consistently outperforming \swapSC{}. Combining
\envelopeMatchingDecoder{} with \midswapSC{} improves the threshold
from $3.42\%$ to $4.38\%$, a $30\%$ increase over prior work.

\marginalMatchingDecoder{} exhibits severe performance degradation as $\eta \to
1$: as $p_{\mathrm{pauli}} \to 0$, the Pauli weights $-\log
\frac{p_{\mathrm{pauli}}}{1-p_{\mathrm{pauli}}}\to\infty$, while
loss-induced edges retain constant weight, causing the decoder to
over-select loss-induced edges.
Theoretical analysis confirms that \marginalMatchingDecoder{} is limited to
$d_{\mathrm{loss}}\sim d/2$ (\zcref{app:harvard_decoder}), whereas our
approach achieves $d_{\mathrm{loss}}\sim 2d/3$.

\subsection{Compatibility with Transversal Logical Circuits}
Beyond memory experiments, we evaluate \envelopeMatchingDecoder{} on
transversal logical circuits. We consider a two-logical-qubit circuit
that repeatedly performs transversal $CNOT$ gates (with the first qubit
as control and the second as target) after each syndrome extraction cycle,
followed by a logical
$Z$ measurement on the target qubit.
\zcref{fig:transversal_threshold} shows that
\envelopeMatchingDecoder{} with \midswapSC{} achieves higher thresholds
and effective distances than \marginalMatchingDecoder{}, demonstrating that our
approach naturally extends to logical circuit decoding.

\begin{figure*}[htbp]
  \centering
  \includegraphics[width=0.9\textwidth]{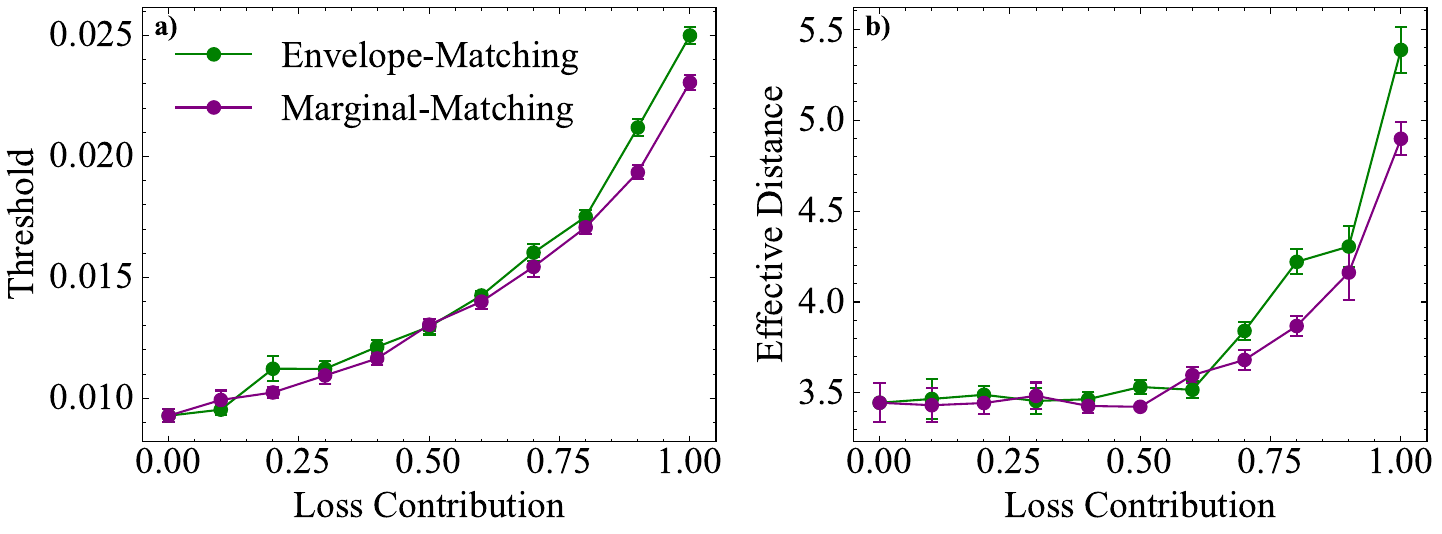}
  \caption{Comparison of the \envelopeMatchingDecoder{} and the
    \marginalMatchingDecoder{} for transversal logical circuits with
    \midswapSC{}: (a) threshold versus loss contribution;
  (b) effective distance at $d=7$ versus loss contribution.}
  \label{fig:transversal_threshold}
\end{figure*}

\section{Summary and Outlook}

We have presented the Pauli envelope, a theoretical framework for
analyzing atom loss in neutral-atom quantum computers. Using this
framework, we designed the \midswapSC{}, a syndrome extraction circuit
that achieves an optimal loss distance
$d_{\mathrm{loss}} \sim d$
with the \envelopeMLEDecoder{},
and $d_{\mathrm{loss}} \sim 2d/3$ with the
efficient \envelopeMatchingDecoder{}.

Our results demonstrate that atom loss, despite requiring destructive
detection, can achieve the same loss distance as erasure errors
and a much higher threshold than Pauli errors.
This shows that atom loss need not limit the achievable distance in
the delayed-detection, gate-removal model studied here
and provides new insights for future hardware and decoder co-design.
As shown in \zcref{app:correlated}, correlated atom loss is
\emph{easier} to correct than independent loss, with the threshold increasing
from $5.15\%$ to $7.82\%$ as the correlated fraction grows when only
considering atom loss on two-qubit gates. This suggests
that correlated atom-loss events should not pose an additional
challenge in this noise model, though a more detailed
analysis under realistic hardware noise models is needed.

Several directions remain open for future work.
First, extending our methods beyond surface codes remains open;
we include a brief discussion on general CSS codes in \zcref{app:css}.
Second, the performance gap between \envelopeMLEDecoder{} and
\envelopeMatchingDecoder{} suggests room for intermediate decoders
that balance optimality and efficiency.
Third, dual-rail superconducting qubits exhibit a similar error model
to atom loss. Applying our framework to this platform is a promising
direction.
Finally, the Pauli envelope can serve as an input feature for
machine learning-based decoders like those
in~\cite{bonilla2025neural}, potentially improving their
interpretability and generalization.

\acknowledgments
We thank
G. Baranes, J. P. Bonilla Ataides, J. Claes,
H. Dehghani, C. Duckering, S.Z. Gu,
A. Kubica,
M. D. Lukin,
Y. Wu, and C.C. Yu for helpful discussions and comments on the manuscript.
P.L.\ and U.A.\ are supported by NSF grants CCF-1901381, CCF-2115104,
CCF-2119352, and CCF-2107241. We are grateful to Chameleon Cloud for
providing the computing cycles needed for the experiments.
S.J.S.T.\ acknowledges support from the National University of
Singapore (NUS) Development Grant.
E.H.\ is supported by the Fulbright Future Scholarship.
H.Z.\ and C.Z.\ are supported by IARPA and the Army Research Office
under the Entangled Logical Qubits program (Cooperative Agreement
Number W911NF-23-2-0219) and the DARPA MeasQuIT program
(HR0011-24-9-0359).


\paragraph*{Data availability.}
The data supporting the findings of this study are available on Zenodo
\cite{pengyu_2026_19339056}.

\clearpage
\bibliography{main}
\newpage
\appendix
\setcounter{secnumdepth}{3}
\section{Proofs from Pauli Envelope Section}
\label{app:methods-proofs}

This section contains the detailed proofs of lemmas and theorems
presented in the Pauli envelope section.




\subsection{Proof of Effective Distance Theorem}
\label{app:proof-effective-distance}

\thmeffectivedistance*

\begin{proof}
  Consider any loss configuration $l$ with readout $r = \mathcal{R}(l)$
  and Pauli error $p$. Let $E$ denote the Pauli envelope of $r$.
  By \zcref{lem:sufficiency}, if $\textsf{Dec}(\cdot, r)$ correctly
  decodes all detector-observable pairs in
  $\mathcal{S}(p \oplus E, \emptyset)$,
  then it also correctly decodes all pairs in
  $\mathcal{S}(p, \mathcal{L}(r))$.

  Taking the contrapositive: for any realized detector-observable
  pair at fixed $(p,l)$ where the decoder fails, i.e.,
  $\textsf{Dec}(\mathcal{D}(p,l), \mathcal{R}(l)) \neq \mathcal{O}(p,l)$,
  there must exist some
  $(\mathcal{D}, \mathcal{O}) \in \mathcal{S}(p \oplus E, \emptyset)$
  such that $\textsf{Dec}(\mathcal{D}, \mathcal{R}(l)) \neq \mathcal{O}$.

  Averaging over circuit measurements, assigned bits, and the sampled
  loss and Pauli configurations, we obtain
  \begin{align*}
    &P_{\mathrm{LER}}(\mathcal{C}, \mathcal{P}_l, \, \mathcal{P}_p,
    \textsf{Dec}) \\
    &= \Pr_{\substack{l \sim \mathcal{P}_l \\ p
    \sim \mathcal{P}_p}}
    \left[\textsf{Dec}(\mathcal{D}(p,l), \mathcal{R}(l)) \neq
    \mathcal{O}(p,l)\right] \\
    &\leq \Pr_{\substack{l \sim \mathcal{P}_l
    \\ p \sim \mathcal{P}_p}}\left[\exists
      (\mathcal{D}, \mathcal{O}) \in \mathcal{S}(p
      \oplus E, \emptyset) :
    \textsf{Dec}(\mathcal{D}, \mathcal{R}(l)) \neq \mathcal{O}\right] \\
    &= P_{\mathrm{fail}}(\mathcal{C}, \mathcal{P}_l, \mathcal{P}_p,
    \textsf{Dec}).
  \end{align*}
\end{proof}

\subsection{Proof of Pauli Envelope of Atom Loss}
\label{app:proof-detector-equivalence}

\begin{figure}[htbp]
  \centering
  \includegraphics[width=\columnwidth]{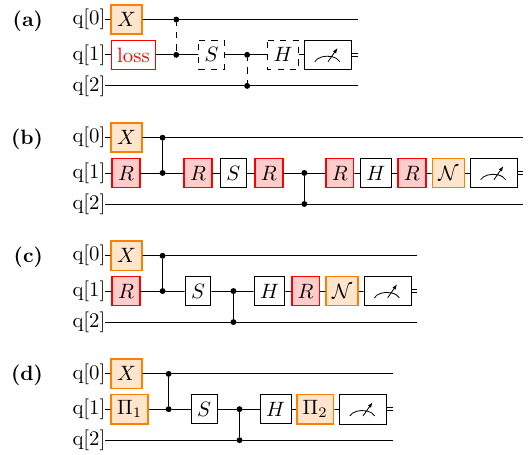}
  \caption{Construction of the Pauli envelope for atom loss. The
    orange $X$ gate represents a Pauli error $p$.
    (a)~Original circuit $\mathcal{C}$ with a loss event (red box) on qubit $q_1$.
    (b)~Circuit $\mathcal{C}_1$ where reset gates $R$ are inserted after every
    gate following the loss, and a completely depolarizing channel
    $\mathcal{N}$ is inserted before the measurement.
    (c)~Circuit $\mathcal{C}_2$ where resets are only at key locations:
    the loss location and after each Hadamard gate; $\mathcal{N}$ is kept.
    $CZ$ and $S$ gates preserve $\ket{0}$, so intermediate resets are redundant.
    (d)~Circuit $\mathcal{C}_3$, which is equivalent to $\mathcal{C}$
    with Pauli envelope $\Pi_1\oplus \Pi_2$, where
  $\Pi_i$ is the Pauli set at labeled location $i$.}
  \label{fig:pauli_envelope_proof}
\end{figure}

\lemdetectorequivalence*

\begin{proof}
  We construct a sequence of equivalent circuits, illustrated in
  \zcref{fig:pauli_envelope_proof}.

  \textbf{Step 1} ($\mathcal{C} \to \mathcal{C}_1$):
  Insert a reset gate after every gate following the loss event.
  Since the lost atom is in the $\ket{0}$ state, all the $CZ$ gates
  have no effect.

  Furthermore, a completely depolarizing channel $\mathcal{N}$ is added
  before the measurement to model the random bit $b_m$ assigned to the
  lost atom.

  We use $\mathcal{S}(\mathcal{C})$ to denote the set of possible
  outcomes of $\mathcal{C}$. By our modeling of atom loss,
  $\mathcal{S}(\mathcal{C}) = \mathcal{S}(\mathcal{C}_1)$.

  \textbf{Step 2} ($\mathcal{C}_1 \to \mathcal{C}_2$):
  Retain resets only at two key locations:
  (1) the loss location, and
  (2) immediately after each subsequent Hadamard gate on the lost atom;
  the depolarizing channel $\mathcal{N}$ before the measurement is kept.
  Since $CZ$ and $S$ gates preserve the $\ket{0}$ state,
  the intermediate resets are redundant, giving
  $\mathcal{S}(\mathcal{C}_1) =
  \mathcal{S}(\mathcal{C}_2)$.

  \textbf{Step 3} ($\mathcal{C}_2 \to$ Pauli envelope):
  Let $E$ denote all possible Pauli errors at these three types of
  locations (the two reset locations and the location of $\mathcal{N}$),
  and let $\mathcal{C}_3$ be $\mathcal{C}_2$ with each reset and
  $\mathcal{N}$ replaced by the Pauli set $\Pi$ at that location.
  We claim $\mathcal{S}(\mathcal{C}_2) \subseteq
  \mathcal{S}(\mathcal{C}_3)$. This is a statement about supports
  only; no probabilities of the two circuits are compared.
  A reset is not a mixture of Pauli operators, but the completely
  depolarizing channel is, and it contains the reset: for any state
  $\rho$ of the circuit, a reset of qubit $q$ gives
  $R(\rho)=\ket{0}\bra{0}_q\otimes\operatorname{Tr}_q\rho$, whereas
  \[
    \mathcal{N}(\rho)=\tfrac{1}{2}\ident_q\otimes\operatorname{Tr}_q\rho
    =\tfrac{1}{2}R(\rho)+\tfrac{1}{2}X_q R(\rho) X_q .
  \]
  Hence every outcome that has nonzero probability after $R$ also has
  nonzero probability after $\mathcal{N}$. Since
  $\mathcal{N}(\rho)=\tfrac{1}{4}\sum_{P\in\Pi}P\rho P$, every outcome
  that has nonzero probability after $\mathcal{N}$ has nonzero
  probability for at least one Pauli operator $P\in\Pi$ inserted at
  that location. Applying this at each of the three types of locations
  gives $\mathcal{S}(\mathcal{C}_2) \subseteq
  \mathcal{S}(\mathcal{C}_3)$.

  Notice $\mathcal{C}_3$ is just $\mathcal{C}$ with Pauli envelope
  $E$, and none of the above derivations depends on whether there
  are additional Pauli errors. Combining these steps, we have
  $\forall p, \mathcal{S}(p, l) \subseteq \mathcal{S}(p \oplus E, \emptyset)$.
\end{proof}

\subsection{Proof of Pauli Envelope for Multiple Losses}
\label{app:proof-loss-resolving}

\lemlossresolvingreadout*

\begin{proof}
  \textbf{Part (a):}
  Consider a loss configuration $l = l_1 \oplus l_2 \oplus \cdots \oplus l_k$
  with $k$ loss locations.
  We construct a sequence of circuits
  $\mathcal{C}_0, \mathcal{C}_1, \cdots, \mathcal{C}_k$,
  where $\mathcal{C}_i$ is obtained from $\mathcal{C}_{i-1}$ by
  removing the gates affected by loss $l_i$, thus
  $\mathcal{S}_{\mathcal{C}_i}(p,
  \emptyset)=\mathcal{S}_{\mathcal{C}_{i-1}}(p, l_i)$.
  Let $E(l_i)$ denote the Pauli envelope for loss $l_i$.
  Applying \zcref{lem:detector-equivalence} to the transition from
  $\mathcal{C}_{i-1}$ to $\mathcal{C}_i$, we have
  $\mathcal{S}_{\mathcal{C}_i}(p, \emptyset) =
  \mathcal{S}_{\mathcal{C}_{i-1}}(p, l_i) \subseteq
  \mathcal{S}_{\mathcal{C}_{i-1}}(p \oplus E(l_i), \emptyset)$.
  Iterating this relation backwards from $k$ to $0$, we obtain
  \begin{align*}
    \mathcal{S}_{\mathcal{C}}(p, l) &= \mathcal{S}_{\mathcal{C}_k}(p,
    \emptyset) \\
    &\subseteq \mathcal{S}_{\mathcal{C}_{k-1}}(p \oplus E(l_k), \emptyset) \\
    &\subseteq \mathcal{S}_{\mathcal{C}_{k-2}}(p \oplus E(l_k) \oplus
    E(l_{k-1}), \emptyset) \\
    &\vdots \\
    &\subseteq \mathcal{S}_{\mathcal{C}_0}\left(p \oplus
    \bigoplus_{i=1}^k E(l_i), \emptyset\right).
  \end{align*}
  Thus $E(l) = \bigoplus_{i=1}^{k} E(l_i)$ is a valid Pauli envelope.

  \textbf{Part (b):}
  For a readout $r$ with possible loss configurations
  $\mathcal{L}(r) = \{l^{(1)}, l^{(2)}, \ldots\}$,
  we have $\mathcal{S}(p, \mathcal{L}(r)) = \bigcup_{l \in
  \mathcal{L}(r)}
  \mathcal{S}(p, l)$.
  By definition, $\mathcal{S}(p, l) \subseteq \mathcal{S}(p \oplus
  E(l), \emptyset)$,
  so $\mathcal{S}(p, \mathcal{L}(r)) \subseteq \bigcup_l
  \mathcal{S}(p \oplus
  E(l), \emptyset)
  = \mathcal{S}(p \oplus \bigcup_l E(l), \emptyset)$.
  Thus $E(r) = \bigcup_{l \in \mathcal{L}(r)} E(l)$.

  \textbf{Part (c):}
  Let $M_r = \{m : r_m = 1\}$. The possible configurations are
  $\mathcal{L}(r) = \{\bigoplus_{m \in M_r} l_m :
  l_m \in \mathcal{L}(r^{(m)})\}$.
  By part (a), $E(\bigoplus_{m} l_m) = \bigoplus_{m} E(l_m)$.
  By part (b), $E(r) = \bigcup_{(l_m)_m} \bigoplus_{m} E(l_m)
  = \bigoplus_{m} \bigcup_{l_m} E(l_m) = \bigoplus_{m} E(r^{(m)})$,
  with all indices $m \in M_r$,
  where the second equality uses the distributive property
  $(A \cup B) \oplus (C \cup D) = (A \oplus C) \cup (A \oplus D)
  \cup (B \oplus C) \cup (B \oplus D)$.
\end{proof}

\section{Detector Patterns from a Single Atom Loss}
\label{app:loss-edge-correspondence}

In this section, we illustrate how each atom loss event triggers a
set of detector flips. For simplicity, we consider the $Z$-error sector,
detected by $X$-stabilizer measurements, as these errors are already
sufficient for our proofs. The detector patterns for both \midswapSC{}
and \swapSC{} are shown in \zcref{fig:loss_edge_correspondence_combined}.

\begin{figure*}[htbp]
  \centering
  \begin{subfigure}[b]{0.48\textwidth}
    \centering
    \includegraphics[width=\textwidth]{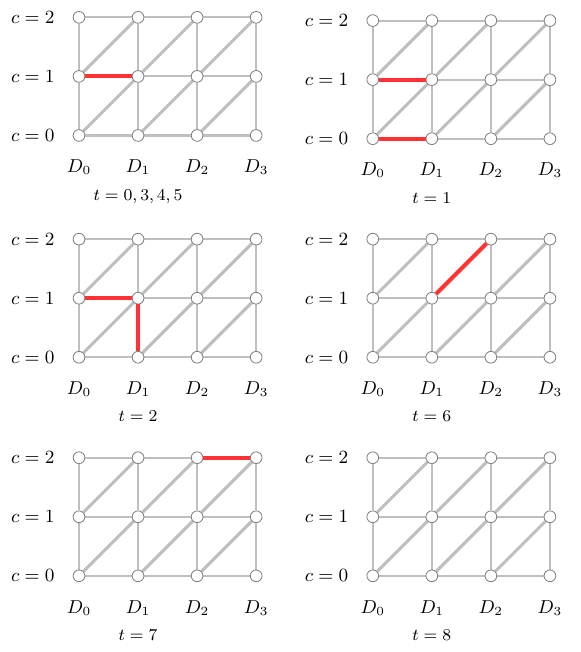}
    \caption{\midswapSC{} (Z-type ancilla loss)}
    \label{fig:loss_edge_correspondence_midswap_x}
  \end{subfigure}
  \hfill
  \begin{subfigure}[b]{0.48\textwidth}
    \centering
    \includegraphics[width=\textwidth]{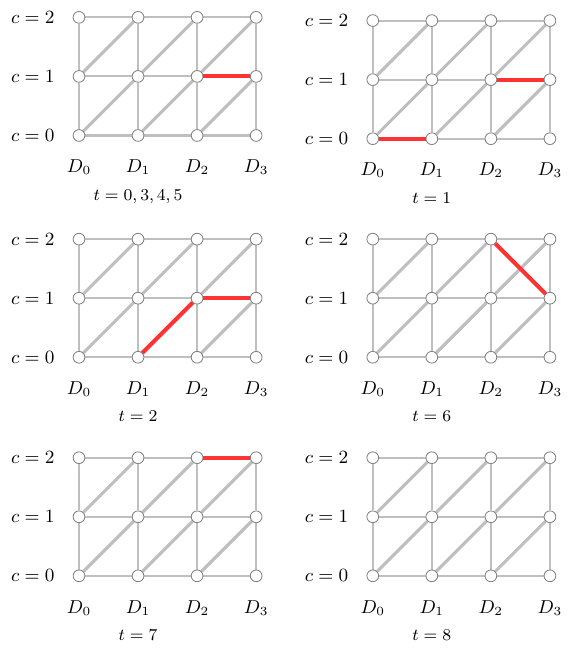}
    \caption{SWAP syndrome extraction (Z-type ancilla loss)}
    \label{fig:loss_edge_correspondence_swap_x}
  \end{subfigure}

  \vspace{1em}

  \begin{subfigure}[b]{0.48\textwidth}
    \centering
    \includegraphics[width=\textwidth]{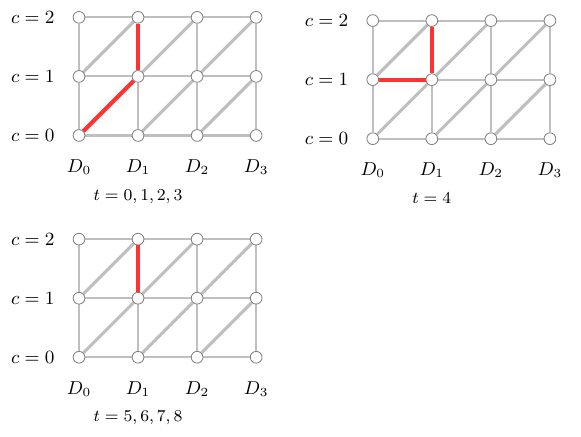}
    \caption{\midswapSC{} (X-type ancilla loss)}
    \label{fig:loss_edge_correspondence_midswap_z}
  \end{subfigure}
  \hfill
  \begin{subfigure}[b]{0.48\textwidth}
    \centering
    \includegraphics[width=\textwidth]{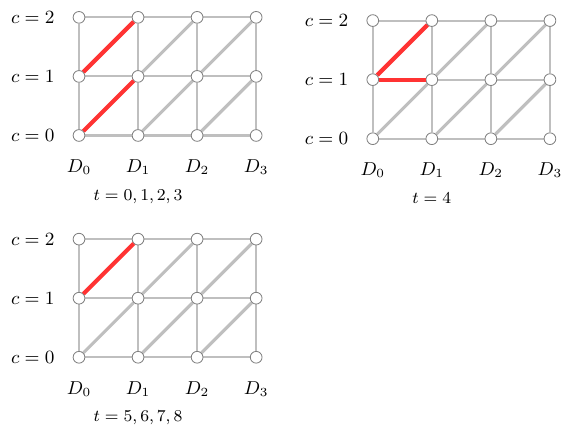}
    \caption{SWAP syndrome extraction (X-type ancilla loss)}
    \label{fig:loss_edge_correspondence_swap_z}
  \end{subfigure}

  \caption{Loss-induced detector patterns for different surface code
    variants. Each red edge can be triggered independently
    and will flip its two endpoints. This is a flattened view of the
    $S_1$--$S_2$--$S_3$ path as shown in \zcref{fig:z-shaped-loss}.
    $t$ represents
    the time step (the number of $CNOT$ gates since the last reset) when
    the atom loss event occurs. $c$ represents the index of syndrome extraction
  rounds.}
  \label{fig:loss_edge_correspondence_combined}
\end{figure*}

\subsection{Mid-SWAP Syndrome Extraction}
\label{app:loss-edge-correspondence-midswap}
The detector patterns for \midswapSC{} are shown in
\zcref{fig:loss_edge_correspondence_midswap_x} and
\zcref{fig:loss_edge_correspondence_midswap_z} when a loss-resolving
readout is triggered on $Z$- and $X$-type ancilla qubits,
respectively. To generate these figures, we first convert
each atom loss event to its Pauli envelope, then identify all
affected edges in the matching graph.
These results will be sufficient for the proof of
\zcref{lem:z-shaped-loss}.
\subsection{SWAP Syndrome Extraction}
\label{app:loss-edge-correspondence-swap}
For \swapSC{}, the detector patterns are shown in
\zcref{fig:loss_edge_correspondence_swap_x} and
\zcref{fig:loss_edge_correspondence_swap_z} when a loss-resolving
readout is triggered on $Z$- and $X$-type ancilla qubits,
respectively. In this case, we do not convert the atom
loss event to its Pauli envelope, but instead directly identify all
affected edges by removing the gates that act on the lost atom. This
approach is used because these figures will be employed to prove the
upper bound on the loss distance in the next section, where we
need to ensure that the failure events can actually occur due to atom
loss.

\section{Proofs from \envelopeMLEDecoder{} Section}
\label{app:mle-decoder-proofs}

This section contains the detailed proofs of lemmas and theorems
from the \envelopeMLEDecoder{} section, in the limit $p\to0$ of
\zcref{sec:mle-decoder}, i.e., with $\lambda_{m,j}=0$ and $w_i=w$.

\subsection{Proof of Optimality of \envelopeMLEDecoder{}}
\label{app:proof-optimality-mle}

\lemoptimalityofmledecoder*

\begin{proof}
  According to \zcref{lem:detector-equivalence}, the Pauli envelope $E$
  constructed for each loss configuration satisfies the property that
  any detector-observable pair produced by the circuit with loss can
  also be produced by the circuit with errors from $E$ and no loss.

  Therefore, the actual error configuration (consisting of actual Pauli
    errors and loss events with their corresponding Pauli
  envelopes) is a feasible solution to the MILP.

  Since the MILP minimizes the total Pauli error weight over all
  feasible solutions, and the actual error is feasible, the optimal
  solution found by the decoder must have weight no greater than the
  weight of the actual error.
\end{proof}

\subsection{Proof of Loss Pattern Weight Requirement}
\label{app:proof-loss-pattern-reduction}

We work in one error sector of the matching graph. For each
loss-induced pair $(\mathcal{D}_{m,j},\mathcal{O}_{m,j})$, choose a
representative error chain $F_{m,j}$ with detector boundary
$\partial F_{m,j}=\mathcal{D}_{m,j}$ and logical observable flip
$\mathcal{O}_{m,j}$. Here $\oplus$ on chains denotes symmetric difference.

For $n_l$ atoms $m_i$ reported as lost, let $j_i$ and $k_i$ index the pairs
in a Pauli-envelope representation of the actual error and in the
decoder's solution, respectively. Their loss component is the chain
\[
  C_L=\bigoplus_{i=1}^{n_l}(F_{m_i,j_i}\oplus F_{m_i,k_i}),
\]
whose boundary is the syndrome difference
$\partial C_L=\bigoplus_{i=1}^{n_l}
(\mathcal{D}_{m_i,j_i}\oplus\mathcal{D}_{m_i,k_i})$.
Let $C_P$ be the symmetric difference of the actual and decoded
Pauli-error chains. Then $C=C_P\oplus C_L$ has empty detector
boundary. If the decoder fails, $C$ is an \emph{undetectable logical
error}: a chain with empty detector boundary and a nontrivial
logical observable flip.

\begin{figure*}[t]
  \centering
  \includegraphics[width=\textwidth]{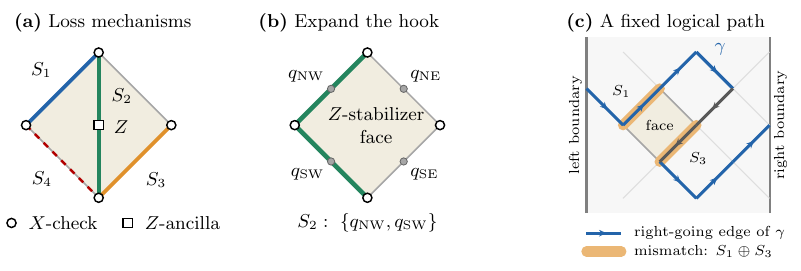}
  \caption{Spatial supports and path counting for
    \zcref{lem:loss-pattern-reduction} in the $Z$-error sector.
    (a)~The local loss mechanisms from \zcref{fig:z-shaped-loss};
    circles denote $X$-checks and the square denotes the $Z$-ancilla.
    (b)~In the code-capacity matching graph, each side of the diamond is
    one data qubit. The hook $S_2$ expands into the two green edges;
    adding the face stabilizer exchanges them for the two right-side
    edges. NW, SW, NE, and SE denote the four compass positions.
    (c)~A simple path $\gamma$ joins the left and right boundaries.
    Blue arrows indicate its right-going edges and gray arrows its
    left-going edges. The orange bands mark the mismatch
    $S_1\oplus S_3$. Although the path uses both marked edges, it
    traverses them in opposite horizontal directions, so the mismatch
  covers only one right-going edge.}
  \label{fig:loss-weight-diamond}
\end{figure*}

\begin{lemma}
  \label{lem:loss-pattern-reduction}
  For the \midswapSC{} with distance $d$, if $C_P\oplus C_L$
  forms an undetectable logical error, then $|C_P| \geq d - n_l$,
  where $n_l$ is the number of loss events, and $|\cdot|$ denotes the
  number of errors, or the weight divided by $w$.
\end{lemma}
\begin{proof}
  Work in the $Z$-error sector; the $X$ sector is analogous. We fix
  one logical path with at least $d$ right-going edges and show that
  each Pauli mechanism and each loss mismatch covers at most one of
  the right-going edges.

  \textbf{Projection onto space.}
  Use the code-capacity matching graph: its vertices are the
  $X$-checks, and each data qubit is an edge joining the checks that
  detect a $Z$ error on it; a data qubit on the left or right boundary
  has only one such check, and its edge ends on that boundary. The four
  data qubits of a bulk $Z$-stabilizer bound a diamond-shaped face.
  This graph has no hook or time-like edges, and a hook such as $S_2$
  becomes two data-qubit edges (\zcref{fig:loss-weight-diamond}(a,b)).
  Propagating a fault through the ideal circuit to the data qubits at
  the end of the circuit gives its
  support, a set of edges defined up to $Z$-stabilizers. Each
  Pauli mechanism in $C_P$ and each loss mismatch
  $F_{m_i,j_i}\oplus F_{m_i,k_i}$ in $C_L$ has a support, and their
  symmetric difference is the projection of $C$.

  \textbf{Local support shapes.}
  Label the data qubits of a bulk $Z$-stabilizer face
  $q_{\rm NW},q_{\rm SW},q_{\rm NE},q_{\rm SE}$
  (\zcref{fig:loss-weight-diamond}(b)); the $Z$-extraction order of
  \zcref{fig:z-shaped-loss} is
  $q_{\rm NW},q_{\rm SW},q_{\rm NE},q_{\rm SE}$. A Pauli fault on the
  $Z$-ancilla or one of its CNOTs propagates to a suffix of this order,
  which up to the face stabilizer is empty, a single edge, or the
  vertical hook $\{q_{\rm NW},q_{\rm SW}\}$ (equivalently
  $\{q_{\rm NE},q_{\rm SE}\}$). For a loss at this ancilla,
  \zcref{lem:z-shaped-loss} gives the patterns
  $\emptyset,S_1,S_2,S_3$ with supports
  \[
    S_1:\{q_{\rm NW}\},\qquad
    S_2:\{q_{\rm NW},q_{\rm SW}\},\qquad
    S_3:\{q_{\rm SE}\}.
  \]
  Equal patterns cancel, and the other mismatches are
  \begin{itemize}
    \item a single edge: $S_1$, $S_3$,
      $S_1\oplus S_2=\{q_{\rm SW}\}$, and
      $S_2\oplus S_3=\{q_{\rm NW},q_{\rm SW},q_{\rm SE}\}$, which is
      equivalent to $\{q_{\rm NE}\}$ up to stabilizers;
    \item the vertical hook: $S_2$;
    \item two opposite edges: $S_1\oplus S_3=\{q_{\rm NW},q_{\rm SE}\}$.
  \end{itemize}
  A loss at an $X$-ancilla gives only the fixed edge $S_4$ or nothing,
  and at a weight-two boundary $Z$-stabilizer every support reduces to
  at most one edge. Every support is therefore empty, a single edge, a
  vertical hook, or two opposite edges of one face.

  \textbf{A left-to-right path.}
  Since $C$ has empty detector boundary, its projection has even
  degree at every $X$-check, so it decomposes into cycles and paths
  ending on the boundary. Cycles and paths returning to the boundary
  they started from are products of $Z$-stabilizers; since $C$ flips
  the logical observable, an odd number of the paths join the left
  boundary to the right one. Choose one, remove any loops so that it
  visits no vertex twice, and orient the resulting path $\gamma$ from
  left to right. Each edge changes the data-qubit column by one and the
  boundaries are $d$ columns apart, so $\gamma$ has at least $d$
  right-going edges.

  Every support contains at most one right-going edge of $\gamma$;
  only the two-edge shapes need an argument. The two edges of a
  vertical hook meet at the west (or east) corner of the face, so if
  both were right-going, $\gamma$ would leave (or enter) that corner
  twice. For two opposite edges, close $\gamma$ by an arc outside the
  code patch. By the Jordan curve
  theorem~\cite{thomassen1992jordan}, the resulting simple closed
  polygon separates the plane into two regions, one on the left and
  one on the right of $\gamma$ all along it. The interior of the face
  is connected and disjoint from the polygon, so the face lies on the
  same side of every edge of $\gamma$ on its boundary. Traversed
  rightward, however, $q_{\rm NW}$ has the face on its right and
  $q_{\rm SE}$ on its left, so at most one of them is right-going
  (\zcref{fig:loss-weight-diamond}(c)).

  \textbf{Counting.}
  Each loss event contributes at most one right-going edge of
  $\gamma$, and so does each Pauli mechanism in $C_P$. Since every
  right-going edge of $\gamma$ lies in the projection of $C$, it is
  contributed by at least one of them. Hence $d\leq|C_P|+n_l$.
\end{proof}

\subsection{Proof of \envelopeMLEDecoder{} Achieves Optimal Distance}
\label{app:proof-mle-optimal-distance}

\thmmleoptimaldistance*

\begin{proof}
  Suppose the actual error pattern consists of $n_l$ atom losses and
  $n_p$ Pauli errors.

  Use the decomposition $C=C_P\oplus C_L$ from
  \zcref{app:proof-loss-pattern-reduction}, where $C_P$ is the Pauli
  chain difference and $C_L$ is the difference of the selected loss
  representative chains.

  By the optimality of the \envelopeMLEDecoder{}
  (\zcref{lem:optimality-of-mle-decoder}), the decoder
  outputs Pauli weight at most $n_p w$. Thus, the Pauli component satisfies:
  \begin{equation}
    |C_P| \leq 2n_p.
  \end{equation}

  For decoder failure, $C$ must be a non-trivial undetectable logical error.
  Applying \zcref{lem:loss-pattern-reduction} with the $n_l$ loss pattern
  mismatches in $C_L$, we obtain $|C_P| \geq d-n_l$, and therefore:
  \begin{equation}
    |C_P| \geq d - n_l.
  \end{equation}

  Combining these bounds, decoder failure requires:
  \begin{equation}
    2n_p + n_l \geq d.
  \end{equation}

  Therefore, the decoder succeeds when $2n_p + n_l < d$. This bound
  is tight: there exists a configuration of $d$ losses (e.g., forming
  a chain across the lattice) that completely erases the logical
  information, making recovery impossible for any decoder. This
  proves the optimality of both the \envelopeMLEDecoder{} and the
  \midswapSC{}.
\end{proof}

\subsection{Proof of Loss-Distance Scaling for \envelopeMLEDecoder{}}
\label{app:proof-failure-probability}

\thmfailureprobability*

\begin{proof}
  By \zcref{thm:mle-optimal-distance}, decoder failure requires an
  undetectable logical error with $n_l$ losses and $n_p$ Pauli errors
  satisfying $\omega := 2n_p + n_l \geq d$.

  We bound the failure probability by summing over all connected error
  configurations with effective weight $\omega \geq d$. Let $N(s)$
  denote the number of
  connected subgraphs of size $s$ in the lattice. Since each error
  triggers a constant number of detectors, let $\mu$ be the maximum
  number of errors that can trigger overlapping detectors with a
  fixed error, which is also a constant.
  We have
  $N(s) \leq C(d) \mu^s $ for a polynomial function $C(d)$ and a
  constant $\mu$~\cite{kovalev2013fault}.

  For a configuration with $s$ errors containing $k$ losses and $s-k$
  Pauli errors, the effective weight is $\omega = k + 2(s-k) = 2s -
  k$, occurring with
  probability $p_{\mathrm{loss}}^k p_{\mathrm{pauli}}^{s-k}$.

  Define $\bar{p} := \max(p_{\mathrm{loss}}, \sqrt{p_{\mathrm{pauli}}})$. Then:
  \[ p_{\mathrm{loss}}^k p_{\mathrm{pauli}}^{s-k} \leq \bar{p}^k
    (\bar{p}^2)^{s-k} = \bar{p}^{2s
  - k} = \bar{p}^\omega. \]

  Summing over configurations:
  \begin{align}
    P_{\mathrm{fail}} &\leq \sum_{\omega=d}^{\infty} \sum_{s=\lceil \omega/2
    \rceil}^{\omega} C(d) \mu^s \binom{s}{2s-\omega} \bar{p}^\omega \\
    &\leq \sum_{\omega=d}^{\infty} \tilde{C}(d) \tilde{\mu}^\omega
    \bar{p}^\omega,
  \end{align}
  where $\tilde{C}(d), \tilde{\mu}$ absorb the inner sum and account for
  the combinatorial factors.

  For $\bar{p} < p_{\mathrm{th}} := 1/\tilde{\mu}$, this geometric series
  converges to the following:
  \begin{align}
    P_{\mathrm{fail}} &\leq  \frac{\tilde{C}(d)(\tilde{\mu}\bar{p})^d}{1 -
    \tilde{\mu}\bar{p}} \\
    &= \frac{\tilde{C}(d)}{1 -
    \tilde{\mu}\bar{p}}\left(\frac{\bar{p}}{p_{\mathrm{th}}}\right)^d \\
    &\leq \frac{\tilde{C}(d)}{1 -
    \tilde{\mu}\bar{p}}\left(\frac{p_{\mathrm{loss}}^{d}+p_{\mathrm{pauli}}^{d/2}}{p_{\mathrm{th}}^d}\right)
    \\
    &=
    O\left(p_{\mathrm{th}}^{-d}(p_{\mathrm{loss}}^{d}+p_{\mathrm{pauli}}^{d/2})\right).
  \end{align}



  The third step follows from $\bar{p} = \max(p_{\mathrm{loss}},
  \sqrt{p_{\mathrm{pauli}}})$,
  which gives $\bar{p}^d = \max(p_{\mathrm{loss}}^d, p_{\mathrm{pauli}}^{d/2})$.
\end{proof}

\subsection{Extension to Loss-Resolving Errors}
\label{app:proof-measurement-errors}

\begin{theorem}[Extension to Loss-Resolving Errors]
  \label{thm:measurement-errors}
  The \envelopeMLEDecoder{} can be extended to handle loss-resolving
  errors. Suppose with probability
  $p_{\mathrm{readout}}=p_{\mathrm{loss}}$ an atom is
  detected as lost but is actually not lost, and with probability
  $p_{\mathrm{loss}}$ a lost atom is randomly assigned to $0$ or $1$
  with equal probability.
  We modify the \envelopeMLEDecoder{} as follows: instead of enforcing
  $\sum_j x_{m,j} = r_m$ when $r_m = 0$, we relax this constraint to
  $\sum_j x_{m,j}\leq 1$ and add a penalty term $w\sum_j x_{m,j}$ to
  the objective function.
  This modification ensures the same failure probability scaling as
  without loss-resolving errors:
  \[P_{\mathrm{fail}} =
  O\left(p_{\mathrm{th}}^{-d}(p_{\mathrm{loss}}^{d}+p_{\mathrm{pauli}}^{d/2})\right).\]
\end{theorem}

\begin{proof}
  Loss-resolving errors fall into two categories:

  \textbf{False positives} (loss detected but atom not actually
  lost): These do not require special treatment. When the decoder is told
  an atom was
  lost (but it wasn't), we can treat the atom as if it were lost
  just before measurement, and this case is already handled by the
  original decoder.

  \textbf{False negatives} (atom lost but not detected): For false
  negatives, the key observation is that an undetected atom
  loss requires two independent events: (1) the atom is lost, and (2)
  the loss-resolving measurement fails to detect
  it. Therefore, the effective probability is $p_{\mathrm{loss}}^2$.

  Notice that after adding the penalty term, an undetected atom loss
  has the same form as a Pauli error; by repeating the proof of
  \zcref{lem:loss-pattern-reduction}, we can show that the failure
  condition generalizes to $2n_p + n_l + 2n_f \geq d$,
  where $n_f$ is the number of false negatives.

  Following the proof structure of \zcref{thm:failure-probability},
  we can show that the failure probability scaling is unchanged:
  \[P_{\mathrm{fail}} =
  O\left(p_{\mathrm{th}}^{-d}(p_{\mathrm{loss}}^{d}+p_{\mathrm{pauli}}^{d/2})\right).\]

\end{proof}
We did not use this extension in our numerical
evaluation as it significantly increases MILP problem sizes.

\section{Proofs from \envelopeMatchingDecoder{} Section}
\label{app:reweight-matching-proofs}

This section contains the detailed proofs of lemmas and theorems
from the \envelopeMatchingDecoder{} section. For simplicity, we assume
all baseline Pauli-edge weights satisfy $w_e=w$.

\subsection{Proof of Failure Condition}
\label{app:proof-failure-condition}

We first
establish the following property regarding the matching graph weights.
Recall that in \zcref{alg:envelope-matching-decoder}, we set each
edge weight to $w$, and then for each edge affected by atom loss, we
reduce its weight to $0.5w$ for space-like edges and $0.25w$ for
time-like edges.
\begin{lemma}
  \label{lem:loss-weight-properties}
  For the \midswapSC{} and any configuration of $n_l$ atom losses,
  the reweighted matching graph in \zcref{alg:envelope-matching-decoder}
  has the following properties:
  \begin{enumerate}
    \item The detector pattern induced by each loss can be matched with total
      weight at most $0.5w$.
    \item Every undetectable logical error chain $C$ satisfies
      $\mathrm{Weight}(C)\geq(d-0.5n_l)w$.
  \end{enumerate}
\end{lemma}

\begin{proof}
  The first property follows from
  \zcref{fig:loss_edge_correspondence_midswap_x,fig:loss_edge_correspondence_midswap_z};
  the only case worth noting is $t=1$ in
  \zcref{fig:loss_edge_correspondence_midswap_x}, where four
  detectors can be triggered and the pattern is matched with two
  weight-$0.25w$ time-like edges.

  For the second property, fix the full loss configuration and an
  undetectable logical error chain $C$. Work in the $Z$-error sector;
  the $X$ sector is analogous. We group the space-like edges of $C$
  and show that at least $d$ groups are needed. Each group has weight
  $w$, except for at most one group of weight $0.5w$ per loss.

  \textbf{Spatial supports.}
  As in the proof of \zcref{lem:loss-pattern-reduction}, propagate
  errors to the data qubits at the end of the circuit and represent
  their supports in the code-capacity matching graph. Time-like edges
  are measurement-bit flips and have empty support. A Pauli mechanism,
  a single loss pattern, or the symmetric difference of two patterns
  from the same loss has, up to stabilizers, one of four support
  shapes: the empty set, a single edge, a vertical hook, or two
  opposite edges of one stabilizer face.

  Each such support contains at most one right-going edge of any
  simple left-to-right logical path. For a vertical hook, traversing
  both edges rightward would require the path to leave or enter the
  same vertex twice. For two opposite edges, traversing both
  rightward would put the face interior on opposite sides of the
  path. This is impossible as discussed in the proof of
  \zcref{lem:loss-pattern-reduction}.

  \textbf{Grouping the edges.}
  Assign each weight-$0.5w$ space-like edge of $C$ to one atom
  reported as lost whose readout reweights it. If several atoms
  qualify, choose any one, so that each edge is assigned exactly
  once. For each atom, pair its assigned edges arbitrarily, leaving
  one unpaired edge if their number is odd. This partitions all
  space-like edges of $C$ into three types of groups:
  \begin{enumerate}
    \item $a$ groups consisting of one edge of weight $w$ each;
    \item $b$ groups consisting of two weight-$0.5w$ edges assigned
      to the same atom, with total weight $w$ per group;
    \item $c$ groups consisting of one unpaired weight-$0.5w$ edge
      each.
  \end{enumerate}
  Each atom leaves at most one unpaired edge, so
  \[
    c\leq n_l.
  \]
  A group of the second type has the support of the symmetric
  difference of two patterns from the same loss; a group of the
  third type has the support of a single pattern. Thus every group
  admits one of the support shapes described above.

  \textbf{Counting groups.}
  Choose such a support representative for each group. Their
  symmetric difference agrees with the spatial projection of $C$
  up to stabilizers, so it still represents a nontrivial logical
  error. It therefore contains a simple path $\gamma$ joining the
  left and right boundaries. Since these boundaries are $d$ columns
  apart and each spatial edge changes the column by one, $\gamma$
  has at least $d$ right-going edges. Each group contains at most one
  of them, while every edge of $\gamma$ belongs to at least one
  group. Consequently,
  \[
    a+b+c\geq d.
  \]
  Overlaps between group supports do not affect this count: an edge
  surviving their symmetric difference must occur in at least one
  support.

  \textbf{Bounding the weight.}
  The space-like edges have total weight $(a+b+0.5c)w$. Dropping
  the nonnegative weight of the time-like edges gives
  \[
    \begin{aligned}
      \mathrm{Weight}(C)
      &\geq (a+b+0.5c)w\\
      &= (a+b+c-0.5c)w\\
      &\geq (d-0.5c)w\\
      &\geq (d-0.5n_l)w.
    \end{aligned}
  \]
  This proves the second property.
\end{proof}

\lemfailureconditionmatching*

\begin{proof}
  Let $F$ denote the actual error as a set of edges in the matching graph.
  The MWPM decoder predicts an error correction $\widehat{F}$ such that
  $\mathrm{Weight}(\widehat{F}) \leq \mathrm{Weight}(F)$. The decoder fails
  when $F \oplus \widehat{F}$
  forms an undetectable logical error.

  Using \zcref{lem:loss-weight-properties}, we derive the bounds:
  \begin{enumerate}
    \item By Property 1, each loss event contributes weight at most
      $0.5w$ to the error matching. Thus, for $n_p$ Pauli errors and
      $n_l$ losses:
      \begin{equation}
        \mathrm{Weight}(F) \leq (n_p + 0.5 n_l)w.
      \end{equation}

    \item Apply Property 2 directly to the undetectable logical chain
      $F\oplus\widehat F$. Since
      $\mathrm{Weight}(F\oplus\widehat F)\leq
      \mathrm{Weight}(F)+\mathrm{Weight}(\widehat F)$, we have:
      \begin{equation}
        \mathrm{Weight}(F) + \mathrm{Weight}(\widehat{F}) \geq (d - 0.5 n_l)w.
      \end{equation}
  \end{enumerate}

  Combining these bounds with the failure condition
  $2\,\mathrm{Weight}(F) \geq \mathrm{Weight}(F) +
  \mathrm{Weight}(\widehat{F})$:
  \begin{equation}
    2(n_p + 0.5 n_l)w \geq (d - 0.5 n_l)w.
  \end{equation}
  Simplifying yields the final condition:
  \begin{equation}
    2n_p + 1.5 n_l \geq d.
  \end{equation}
\end{proof}

\subsection{Proof of Loss-Distance Scaling for
\envelopeMatchingDecoder{}}
\label{app:proof-reweight-matching}

\thmreweightmatching*

\begin{proof}
  Similar to the proof of \zcref{thm:failure-probability}, we define
  the effective weight $\omega = 2n_p + 1.5 n_l$. For an error
  configuration with $n_p$ Pauli errors and $n_l$ losses, the
  probability is $p_{\mathrm{pauli}}^{n_p} p_{\mathrm{loss}}^{n_l}$.

  Define $\bar{p} := \max(p_{\mathrm{loss}}^{2/3},
  \sqrt{p_{\mathrm{pauli}}})$. Then:
  \begin{equation}
    p_{\mathrm{pauli}}^{n_p} p_{\mathrm{loss}}^{n_l} \leq
    (\bar{p}^2)^{n_p} (\bar{p}^{1.5})^{n_l} = \bar{p}^{2n_p + 1.5n_l}
    = \bar{p}^\omega.
  \end{equation}

  Following the same counting argument as in
  \zcref{thm:failure-probability}, we sum over all error
  configurations with weight $\omega \geq d$:
  \begin{equation}
    P_{\mathrm{fail}} \leq \sum_{\omega=d}^{\infty} \tilde{C}(d)
    \tilde{\mu}^\omega \bar{p}^\omega.
  \end{equation}

  For $\bar{p} < p_{\mathrm{th}} := 1/\tilde{\mu}$, the series converges and is
  dominated by the $\omega = d$ term:
  \begin{equation}
    P_{\mathrm{fail}} = O(p_{\mathrm{th}}^{-d} \bar{p}^d) =
    O\left(p_{\mathrm{th}}^{-d}(p_{\mathrm{loss}}^{2d/3}+p_{\mathrm{pauli}}^{d/2})\right).
  \end{equation}
\end{proof}

\subsection{Extension to Loss-Resolving Errors}
\begin{theorem}[Extension to Loss-Resolving Errors]
  \label{thm:reweight-measurement-errors}
  The decoder can be extended to handle loss-resolving errors
  by initializing time-like edges with
  weight $w_T = 0.5w$ instead of $w$. The decoder maintains:
  \[ P_{\mathrm{fail}} =
    O\left(p_{\mathrm{th}}^{-d}(p_{\mathrm{loss}}^{2d/3}+p_{\mathrm{pauli}}^{d/2})\right).
  \]
\end{theorem}

\begin{proof}
  The extension handles two types of loss-resolving errors:

  \textbf{False positives}: Similar to the proof of
  \zcref{thm:measurement-errors}, this case is already handled by
  the original decoder.

  \textbf{False negatives} (atom lost but not detected):
  With the modified initialization ($w_T = 0.5w$), an undetected loss
  induces a detector pattern. The key observation is:
  \begin{itemize}
    \item The total weight for matching this pattern is at most $w$,
      equivalent to a single Pauli error.
    \item With no atom lost, the weight of a logical error is still $dw$.
    \item With $n_l$ detected losses, every undetectable logical error
      has weight at least $(d-0.5n_l)w$.
  \end{itemize}

  Therefore, in the failure condition analysis, an undetected loss
  ($n_f$ false negatives) contributes similarly to Pauli errors:
  \begin{equation}
    2(n_p + n_f) + 1.5 n_l \geq d.
  \end{equation}

  Since false negatives happen with probability
  $p_{\mathrm{loss}}^2$, and it requires $d/2$ false negatives to
  form a logical error, repeating the proof of
  \zcref{thm:failure-probability}, we can show that this contributes a
  term that scales like $(p_{\mathrm{loss}}^2)^{d/2}$, which
  is dominated by $p_{\mathrm{loss}}^{2d/3}$.

  Therefore, the failure probability remains:
  \[P_{\mathrm{fail}} =
  O\left(p_{\mathrm{th}}^{-d}(p_{\mathrm{loss}}^{2d/3}+p_{\mathrm{pauli}}^{d/2})\right).\]
\end{proof}
We did not use this extension in our numerical
evaluation as it significantly reduces the threshold for Pauli errors.

\begin{figure*}[htbp]
  \centering
  \begin{subfigure}[t]{0.30\textwidth}
    \centering
    \includegraphics[width=\linewidth]{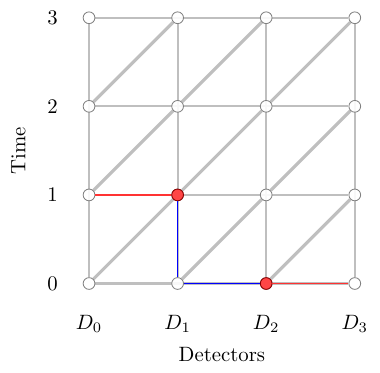}
    \caption{A space-time slice of the matching graph showing a scenario where
    the conventional matching decoder fails.}
    \label{fig:matching_failure}
  \end{subfigure}
  \hfill
  \begin{subfigure}[t]{0.30\textwidth}
    \centering
    \includegraphics[width=\linewidth]{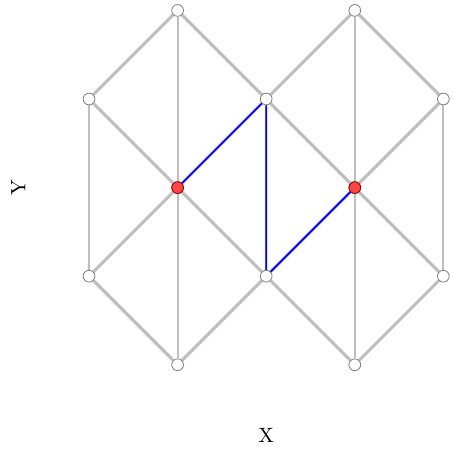}
    \caption{A space-space slice of the matching graph showing a
      scenario where the \averageMLEDecoder{} and
    \marginalMatchingDecoder{} fail.}
    \label{fig:harvard_failure}
  \end{subfigure}
  \hfill
  \begin{subfigure}[t]{0.30\textwidth}
    \centering
    \includegraphics[width=\linewidth]{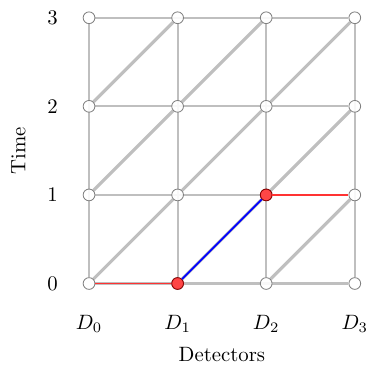}
    \caption{A space-time slice of the matching graph showing a scenario where
    the \swapSC{} fails.}
    \label{fig:swap_failure}
  \end{subfigure}
  \hfill
  \caption{Failure scenarios for (a) the conventional matching
    decoder, (b) the \averageMLEDecoder{} and
    \marginalMatchingDecoder{}, and (c) the \swapSC{}.
    Red dots represent
    detector flips. Red and blue edges represent two sets of edges with
    (asymptotically) equal weight that can produce the same detector
  flips, which cannot be distinguished by the decoder.}
  \label{fig:combined_failure}
\end{figure*}

\section{Upper Bound on the Loss Distance}

We analyze the upper bound on the loss distance for (1)
the conventional matching decoder, (2) \averageMLEDecoder{} and
\marginalMatchingDecoder{}, and (3)
\swapSC{} by directly constructing undecodable error patterns.

\subsection{Conventional matching decoder}
\label{app:bailey_issues}
A conventional matching decoder decodes as if only Pauli errors were
present. However, this approach only achieves
$d_{\mathrm{loss}}\sim d/3$.

As shown in
\zcref{fig:matching_failure},
if errors similar to those in
\zcref{fig:loss_edge_correspondence_midswap_x} ($t=2$) occur, the
conventional decoder incorrectly interprets this
pattern as a weight-2 error, leading to miscorrection.

If $d/3$ such errors occur, the decoder fails.

\subsection{\averageMLEDecoder{} and \marginalMatchingDecoder{}}
\label{app:harvard_decoder}
The \averageMLEDecoder{}~\cite{baranes2025leveraging} constructs a
detector error model by enumerating all possible single-loss
patterns and combining them with the Pauli detector error
model. However, this approach only achieves $d_{\mathrm{loss}}\sim d/2$, as
shown in \zcref{fig:harvard_failure}.

When an ancilla qubit is lost, the \averageMLEDecoder{} will reweight
all the blue edges to a constant weight, while other edges scale like
$-\log \frac{p_{\mathrm{pauli}}}{1-p_{\mathrm{pauli}}}$.

It is possible that the ancilla loss triggers no detector at all
(this can happen with a constant probability). However, since the
blue path has a weight that becomes negligible compared with other edges,
the decoder will connect the two red detectors using this path,
even when the detectors are actually caused by boundary Pauli errors.

If $d/2$ such errors occur, the decoder fails.

\subsection{\swapSC{}}
\label{app:swap_sc}

As shown in \zcref{fig:swap_failure}, when a qubit is lost, either
the red paths or the blue path can produce the same detector flips,
and these cases have the same probability. Therefore, there is no way
to distinguish them.

If $d/2$ such errors occur, any decoder fails. These
cases can actually happen as illustrated in
\zcref{fig:loss_edge_correspondence_swap_x} ($t=1$ and $t=2$).

\section{Extension to Other CSS Codes}
\label{app:css}
In this section, we discuss how to extend our analysis beyond surface
codes to more general CSS codes.

To make the \midswap{} technique applicable, we assume the following
structure for the CSS code. We believe these assumptions can be
relaxed by using different alternating syndrome extraction orders in time.
\begin{definition}[Structure of a CSS Code]
  Let $C$ be a CSS code
  with parameters $[[n, k, d]]$.
  Assume there are $n_x$ $X$-stabilizers and $n_z$ $Z$-stabilizers
  such that $n_x + n_z = n$. (In general, these
    stabilizers need not be independent; we assume this so that the
  \midswap{} technique can be applied.)

  Let each $X$-type stabilizer $S_i^X$ act on data qubits $q_{i,1}^X,
  q_{i,2}^X, \ldots, q_{i,m_i}^X$, and let each $Z$-type stabilizer $S_j^Z$ act
  on data qubits $q_{j,1}^Z,
  q_{j,2}^Z, \ldots, q_{j,m_j}^Z$.

  We assume that all $q_{i,1}^X$ for all $i \in \{1, \ldots, n_x\}$
  and $q_{j,1}^Z$ for all $j \in \{1, \ldots, n_z\}$ are distinct. In other
  words, there is a one-to-one correspondence between ancilla qubits
  and data qubits such that each data qubit appears as the first qubit
  in exactly one stabilizer.
\end{definition}
When the context is clear, we drop the superscripts $X$ and $Z$ for simplicity.
\begin{definition}[Hook error]
  For each $X$-stabilizer $S_i$ acting on data qubits $q_{i,1},
  q_{i,2}, \ldots, q_{i,m}$, we define a hook error $H_{i,t}$ as follows:
  \begin{equation}
    H_{i,t} = \prod_{j=1}^{t} X_{q_{i,j}}.
  \end{equation}
  This describes how an $X$ error on an $X$-ancilla $a_i$ at time step $t$
  propagates to the data qubits. Note that the hook error $H_{i,t}$
  has weight at most $m/2$ since we can always apply $S_i$ to reduce
  its weight to an equivalent Pauli error.

  We define the hook-error distance $d_h$ of $\mathcal{C}$ as the
  minimum number of hook errors and regular Pauli errors combined
  needed to implement a nontrivial
  logical operator.

  Formally, we have
  \begin{equation}
    (\prod_{i=1}^{n_h} H_{i,t_i})\times(\prod_{i=1}^{n_p} X_{q_i})
  \end{equation}
  where any such product with $n_h+n_p<d_h$
  does not implement a nontrivial logical operator. The hook-error distance for
  logical $Z$ operators is defined similarly.
  \label{def:hook-error}
\end{definition}

For any CSS code, we can construct a simple syndrome
extraction circuit that uses only $X$-type and
$Z$-type ancilla qubits to extract the syndrome with a circuit-level
distance of $d_h$ for Pauli errors.

\begin{example}
  \label{ex:simple-syndrome-extraction-circuit}
  A syndrome extraction circuit can be constructed using the following steps:
  \begin{enumerate}
    \item Assign $n_x$ $X$-type ancilla qubits to the $X$-type
      stabilizers, and $n_z$ $Z$-type ancilla qubits to the $Z$-type
      stabilizers.
    \item For each syndrome extraction cycle, we initialize $X$-type
      ancilla qubits in the $|+\rangle$ state and $Z$-type ancilla
      qubits in the $|0\rangle$ state.
    \item For each $X$-type ancilla $a_i$ associated with qubits
      $q_{i,1}, q_{i,2}, \ldots, q_{i,m}$, apply $CNOT(a_i, q_{i,1})$,
      $CNOT(a_i, q_{i,2})$, $\ldots$,
      $CNOT(a_i, q_{i,m})$.
    \item For each $Z$-type ancilla $a_i$ associated with qubits
      $q_{i,1}, q_{i,2}, \ldots, q_{i,m}$, apply $CNOT(q_{i,1}, a_i)$,
      $CNOT(q_{i,2}, a_i)$, $\ldots$, $CNOT(q_{i,m}, a_i)$.
    \item Measure all the $X$-type ancilla qubits in the $X$-basis
      and all the $Z$-type ancilla qubits in the $Z$-basis.
  \end{enumerate}
\end{example}

Some examples of codes satisfying the condition $d_h=d$ include
surface codes, hypergraph product (HGP)
codes~\cite{tillich2014quantum}, and some
bivariate bicycle (BB) codes~\cite{bravyi2024high}. For such codes, the syndrome
extraction circuits are also optimal in terms of
circuit-level distance~\cite{manes2025distance,tan2025effective}.

We can apply the \midswap{} technique to this syndrome extraction
circuit to obtain a new syndrome extraction circuit with atom replenishment.

\begin{example}[\midswap{} syndrome extraction circuit]
  \label{ex:midswap-syndrome-extraction-circuit}
  Formally, for the third and fourth steps of
  \zcref{ex:simple-syndrome-extraction-circuit}, we insert a $SWAP$
  gate between the ancilla-data qubit pair after the first $CNOT$ gate.
  \begin{enumerate}
    \item For each $X$-type ancilla $a_i$ associated with qubits
      $q_{i,1}, q_{i,2}, \ldots, q_{i,m}$, apply
      $CNOT(a_i, q_{i,1})$, $SWAP(a_i, q_{i,1})$,
      $CNOT(a_i, q_{i,2})$, $\ldots$,
      $CNOT(a_i, q_{i,m})$.
    \item For each $Z$-type ancilla $a_i$ associated with qubits
      $q_{i,1}, q_{i,2}, \ldots, q_{i,m}$, apply
      $CNOT(q_{i,1}, a_i)$, $SWAP(q_{i,1}, a_i)$,
      $CNOT(q_{i,2}, a_i)$, $\ldots$, $CNOT(q_{i,m}, a_i)$.
  \end{enumerate}
\end{example}
To analyze the tolerance to atom loss, we need to generalize the loss
pattern of \zcref{lem:z-shaped-loss} to this CSS setting.
\begin{lemma}[Loss-induced hook error]
  \label{lem:generalized-z-shaped-loss}
  Let $\mathcal{C}$ be the \midswap{} syndrome extraction circuit from
  \zcref{ex:midswap-syndrome-extraction-circuit}.

  Let $a_i$ be an $X$-type ancilla associated with qubits
  $q_{i,1}, q_{i,2}, \ldots, q_{i,m}$.
  If an atom
  loss of $a_i$ is detected, then the
  Pauli envelope has the following form when propagated to the end of
  this syndrome extraction cycle:
  \begin{equation}
    E=\{I_{q_{i,1}}, Z_{q_{i,1}}\}\oplus\{
    H_{i,t}\}_{t=0}^{m}\oplus\{\text{Measurement errors}\}.
  \end{equation}
  where $H_{i,t}$ is the hook error
  \begin{equation}
    H_{i,t} = \prod_{j=1}^{t} X_{q_{i,j}}.
  \end{equation}
  Measurement errors, which flip the $0/1$ outcome of the ancilla
  readout, project to the empty set on the data qubits and are ignored
  below. They are distinct from loss-resolving readout errors, which
  misreport whether an atom was lost (\zcref{sec:evaluation}).

  Furthermore, if $a_i$ is a $Z$-type ancilla, it will only trigger a
  data error on $q_{i,1}$.
\end{lemma}
\begin{proof}
  The construction of the \midswap{} syndrome extraction circuit ensures
  that each atom
  loss affects at most two consecutive syndrome extraction cycles. Thus, we
  only need to analyze two consecutive cycles and will
  focus on the remaining errors on the data qubits after these two
  cycles. Suppose the atom loss occurs at time step $t$ (after the
  $t$-th $CNOT$ gate) and is
  detected on the ancilla qubit $a_i$ at the end of the second cycle.
  \begin{enumerate}
    \item $0\le t\le m$: the atom is lost during the data phase.
      The Pauli envelope can be written as
      \begin{equation}
        E_i = \bigoplus_{\tau=t}^{m} \Pi_{q_{i,1}}^\tau \oplus \Pi_{a_i}^{m+1},
      \end{equation}
      where $\Pi_q^\tau = \{I_q^\tau, X_q^\tau, Y_q^\tau, Z_q^\tau\}$
      denotes the Pauli error set on qubit $q$
      at time step $\tau$.

      For any element in the first term of the Pauli envelope
      $\bigoplus_{\tau=t}^{m} \Pi_{q_{i,1}}^\tau$, the residual error
      on the data qubit is an arbitrary Pauli error on $q_{i,1}$.

      For the second term of the Pauli envelope $\Pi_{a_i}^{m+1}$,
      this corresponds to an error on ancilla qubit $a_i$ after the
      first $CNOT$ gate in the second cycle; when propagated to the
      data qubits, it becomes an $X$ error or identity.

      Therefore, the joint effect of these two terms is an arbitrary
      Pauli error on data qubit $q_{i,1}$, which can be written in the form
      \begin{equation}
        E=\{I_{q_{i,1}}, Z_{q_{i,1}}\}\oplus\{
        H_{i,t}\}_{t=0}^{1}.
      \end{equation}

    \item $m < t \le 2m$: The Pauli envelope contains only a single Pauli
      operator:
      \begin{equation}
        E_i = \Pi_{a_i}^t.
      \end{equation}
      When this error propagates to the data qubits, it becomes an
      $X$-type hook error of the form
      \begin{equation}
        H_{i,t-m} = \prod_{j=1}^{t-m} X_{q_{i,j}}.
      \end{equation}

  \end{enumerate}
  Therefore, $\{I_{q_{i,1}}, Z_{q_{i,1}}\}\oplus\{
  H_{i,t}\}_{t=0}^{m}$ is a valid Pauli envelope for the atom loss of $a_i$.
\end{proof}

Without loss-resolving readouts, the effect of atom loss is similar
to hook errors or single data-qubit errors. Thus, if a code can
tolerate $\lfloor (d_h-1)/2\rfloor$ hook errors or single data-qubit
errors, then
it can also tolerate $\lfloor (d_h-1)/2\rfloor$ atom loss errors.

\begin{theorem}[Tolerance to atom loss without loss-resolving readouts]
  Even without loss-resolving readouts, the \midswap{} syndrome
  extraction circuit can correct
  $\lfloor (d_h-1)/2\rfloor$ atom loss errors.
\end{theorem}

With loss-resolving readouts, the decoder obtains more information about
the atom loss. The decoder (\zcref{alg:mle-decoder}) will know the index
$i$ of the hook error $H_{i,t}$ through the loss-resolving readouts
but will be uncertain about the exact time step $t$. If the decoder
guesses the time step as $t^*$, then the residual error is
$H_{i,t}\oplus H_{i,t^*}$.

We define a segment-hook error as an error on a contiguous segment of
data qubits in the index order
$q_{i,1}, q_{i,2}, \ldots, q_{i,m}$ associated with stabilizer $S_i^X$.
Formally, it is an error of the form
\begin{equation}
  H_{i,j_1,j_2} = \prod_{j=j_1}^{j_2} X_{q_{i,j}}.
\end{equation}
where $1 \le j_1 \le j_2 \le m$, and contiguity means consecutive
indices $j_1, j_1+1, \ldots, j_2$ in this order.
Let $d_s$ be the minimum number of segment-hook errors whose product
implements a nontrivial logical operator.
\begin{theorem}[Tolerance to atom loss with loss-resolving readouts]
  A code with segment-hook distance $d_s$ can correct any
  $n_l<d_s$ atom loss errors with loss-resolving readouts.
\end{theorem}

Rotated surface codes (as proved in this paper)
satisfy the stronger condition $d_s=d_h=d$,
thus achieving optimal tolerance to atom loss with the \midswap{}
technique. The larger family of HGP codes, including unrotated
surface codes, which are known to support
efficient error correction~\cite{campbell2019theory,
quintavalle2021single, berthusen2025adaptive}, state
preparation~\cite{hong2025single}, and logical
gates~\cite{quintavalle2023partitioning,
  zhu2023non,breuckmann2024fold,breuckmann2024cups,lin2024transversal,xu2025fast,berthusen2025automorphism,golowich2025constant,golowich2025quantum,zhu2025topological,
zhu2025transversal,tan2025single,li2025transversal}, have also been
shown to achieve the condition
$d_s=d_h=d$~\cite{manes2025distance,tan2025effective}. This result is
consistent with previous work on toric
codes~\cite{yu2025locating}. Thus, they are
strong candidates for quantum error-correcting codes for neutral-atom
architectures.

\section{Extension to Correlated Atom Loss}
\label{app:correlated}
We now discuss the tolerance to correlated atom loss. A correlated
atom loss is an error in which a physical $CZ$ gate causes
both atoms to be lost simultaneously.

Intuitively, correlated atom loss is more destructive than
independent atom loss as it introduces high-weight errors easily.
However, as we will see, correlated atom loss has the same effective
distance as Pauli errors even without
loss-resolving readouts.
\begin{lemma}[Correlated loss-induced hook error]
  \label{lem:generalized-z-shaped-loss-correlated}
  Let $\mathcal{C}$ be the \midswap{} syndrome extraction circuit from
  \zcref{ex:midswap-syndrome-extraction-circuit}.

  Let $a_i$ be an $X$-type ancilla associated with qubits
  $q_{i,1}, q_{i,2}, \ldots, q_{i,m}$.
  If a correlated atom
  loss of $a_i$ with some $q_{i,\tau}$ occurs, then the
  Pauli envelope of this particular loss configuration has the following form
  when propagated to the end of
  this syndrome extraction cycle:
  \begin{equation}
    E=\{I_{q_{i,\tau}}, Z_{q_{i,\tau}}\}\oplus\{
    H_{i,\tau}\}\oplus\{\text{Measurement errors}\}.
  \end{equation}
  where $H_{i,\tau}$ is the hook error
  \begin{equation}
    H_{i,\tau} = \prod_{j=1}^{\tau} X_{q_{i,j}}.
  \end{equation}
\end{lemma}
\begin{proof}
  The proof follows from the same argument as
  \zcref{lem:generalized-z-shaped-loss}.
\end{proof}

\begin{theorem}[Tolerance to correlated atom loss without
  loss-resolving readouts]
  Even without loss-resolving readouts, the circuit can tolerate
  $\lfloor (d_h+1)/2\rfloor$ correlated atom loss errors.
\end{theorem}

With loss-resolving readouts, correlated atom loss provides more information
about the locations of the lost atoms. For example, if
only one pair
of qubits is lost, in the \midswap{} circuit, one
pair of qubits interacts only once, so we can perfectly identify the
locations of the lost atoms. If perfectly identifying the location of
the lost atoms were possible, the Pauli envelope
contains only a single hook error, and the decoder can correctly
decode $n_h$ correlated atom loss errors, meaning $2n_h$ atoms are
lost simultaneously.

Although perfectly identifying the location of the lost atoms is not
always possible, we can still modify the \envelopeMLEDecoder{} to
handle correlated atom loss.

The high-level idea is to jointly minimize the number of correlated
losses and independent loss events.

\begin{algorithm}[\envelopeMLEDecoder{} with correlated atom loss]
  \label{alg:mle-correlated-decoder}
  Before receiving the detector patterns, the decoder first performs the
  following steps as preprocessing:
  \begin{enumerate}
    \item \textbf{Find the Pauli envelope for each loss-resolving readout:}
      For each measured atom $m$, find the Pauli envelope $E_m$ of the
      single-loss readout $r^{(m)}$ as in \zcref{alg:mle-decoder}.
    \item \textbf{Find the Pauli envelope for each correlated
      loss-resolving readout:}
      For each pair of atoms $(m,n)$, $m<n$, that can be lost together, find
      the Pauli envelope $E_{m,n}$ of the readout $r^{(m)}\oplus r^{(n)}$
      corresponding to a single correlated
      atom loss, according to \zcref{lem:detector-equivalence}.
    \item \textbf{Assign variables for loss patterns:}
      For each pair $(\mathcal{D}_{m,j},
      \mathcal{O}_{m,j})\in \mathcal{S}(E_m,\emptyset)$, assign a
      binary variable $x_{m,j} \in \{0,
      1\}$, indicating whether the loss pattern is selected ($x_{m,j} = 1$)
      or not ($x_{m,j} = 0$).
    \item \textbf{Assign variables for correlated loss patterns:}
      For each pair $(\mathcal{D}_{(m,n),j},
      \mathcal{O}_{(m,n),j})\in \mathcal{S}(E_{m,n},\emptyset)$, assign a
      binary variable $x_{(m,n),j} \in \{0,
      1\}$, indicating whether the correlated loss pattern is
      selected ($x_{(m,n),j} = 1$)
      or not ($x_{(m,n),j} = 0$).

    \item \textbf{Assign variables for Pauli errors:}
      For each Pauli error mechanism $i$ in the circuit, assign a binary
      variable $y_i \in \{0, 1\}$, indicating whether the Pauli error
      mechanism is selected ($y_i = 1$) or not ($y_i = 0$).

  \end{enumerate}

  After receiving the detector pattern and loss-resolving readout,
  the decoder performs the
  following steps:
  \begin{enumerate}
    \item \textbf{Enforce constraints:}

      (i)~\emph{Detector satisfaction:}
      For each detector $u$, enforce
      \begin{align}
        &\sum_i D_{i,u} y_i + \sum_{m,j} D_{m,j,u} x_{m,j} \nonumber\\
        &+ \sum_{(m,n),j} D_{(m,n),j,u} x_{(m,n),j} \equiv
        s_u \pmod{2},
      \end{align}
      where $D_{i,u} \in \{0,1\}$ indicates whether Pauli error
      $i$ flips detector $u$,
      $D_{m,j,u} \in \{0,1\}$ indicates whether loss pattern
      $(m,j)$ flips detector $u$,
      and $s_u \in \{0,1\}$ is the observed value of detector $u$.

      (ii)~\emph{Loss exclusivity:} For each readout bit
      $r_m$, where $r_m=1$ indicates that atom $m$ is reported as lost, enforce
      \begin{equation}
        \sum_j x_{m,j} + \sum_{\substack{(a,b),j \\ m \in \{a,b\}}}
        x_{(a,b),j} = r_m,
      \end{equation}
      ensuring exactly one detector pattern is selected per loss event.

    \item \textbf{Minimize loss error count and then Pauli error count:}
      Solve the MILP:
      \begin{equation}
        \min_{x,y} \left(
          \sum_{m,j} x_{m,j} +
          \sum_{(m,n),j} x_{(m,n),j} +
          \epsilon \sum_{i} w_i y_i
        \right)
      \end{equation}
      subject to the above constraints,
      where $\epsilon$ is a small positive constant ensuring the loss
      error count is minimized first.

    \item \textbf{Predict logical observable:}
      Let $x_{m,j}^*,y_i^*$ be the optimal solutions to the MILP.
      Calculate the predicted logical flip:
      \begin{equation}
        \mathcal{O}_{pred} = \bigoplus_i y_i^* \mathcal{O}_i \oplus
        \bigoplus_{m,j} x_{m,j}^* \mathcal{O}_{m,j} \oplus
        \bigoplus_{(m,n),j} x_{(m,n),j}^* \mathcal{O}_{(m,n),j}
      \end{equation}
      where $\mathcal{O}_i$ is the observable flip caused by Pauli error $i$.
  \end{enumerate}
\end{algorithm}

For evaluation, we consider only atom loss after each $CNOT$ gate and
do not consider Pauli errors. For an error rate $p$ and correlated
loss contribution $\eta$, each atom has probability $p(1-\eta)/2$ of
being lost independently and probability $p\eta/2$ of being lost
together with its partner. Thus, the marginal probability of each
atom being lost
remains $p/2$.
\begin{figure}[htbp]
  \centering
  \includegraphics[width=\columnwidth]{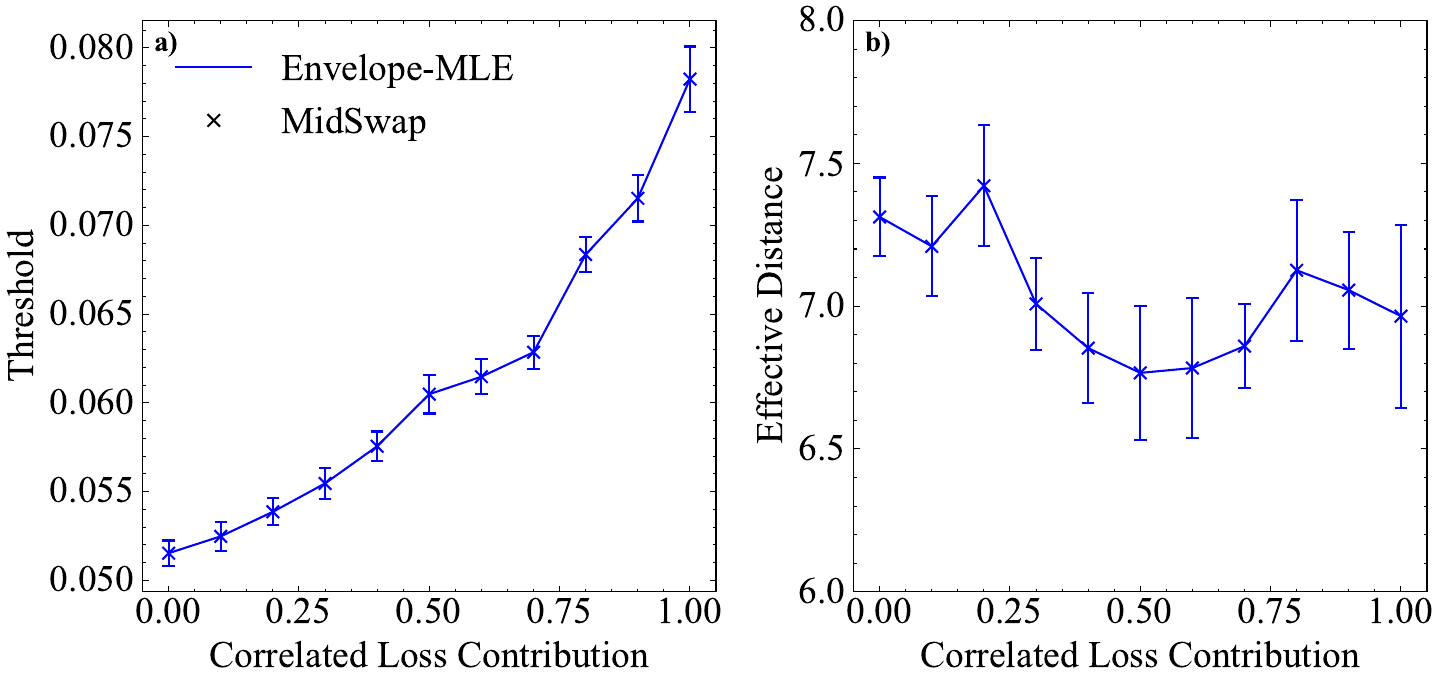}
  \caption{(a) Threshold versus correlated-loss contribution;
    (b) effective distance at code distance $7$ versus correlated-loss
  contribution.}
  \label{fig:correlated_combined}
\end{figure}

As shown in \zcref{fig:correlated_combined}, the threshold
increases from $0.05$ to $0.08$ as the correlated loss
contribution $\eta$ increases from $0$ to $1$. Furthermore, the
effective distance remains approximately $7$.

In summary, without loss-resolving readouts, both independent atom
loss and correlated atom loss behave similarly to hook errors. With
loss-resolving readouts, independent atom loss behaves similarly to a
segment-hook error with known locations, while correlated atom loss
performs at least as well as independent atom loss and provides a higher
threshold by supplying more information about the locations of the
lost atoms. This finding corroborates a recent independent
study~\cite{perrin2026correlatedatomlossresource}, which also
observes that correlated atom loss can be beneficial for error correction.

We leave for future work a careful analysis of the interplay among
regular Pauli errors, independent atom loss, and correlated atom loss
errors, as well as a
threshold theorem and
relaxation of the assumptions that each ancilla qubit is only
associated with one data qubit under alternating syndrome extraction orders.

\section{Teleportation-based Syndrome Extraction Circuit}
\label{app:teleportation_circuit}
\begin{figure}[htbp]
  \centering
  \includegraphics[width=0.6\columnwidth]{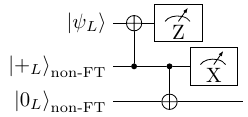}
  \caption{Teleportation-based syndrome extraction circuit proposed in
  Ref.~\cite{baranes2025leveraging}.}
  \label{fig:teleportation_circuit}
\end{figure}

\begin{figure}[htbp]
  \centering
  \includegraphics[width=\columnwidth]{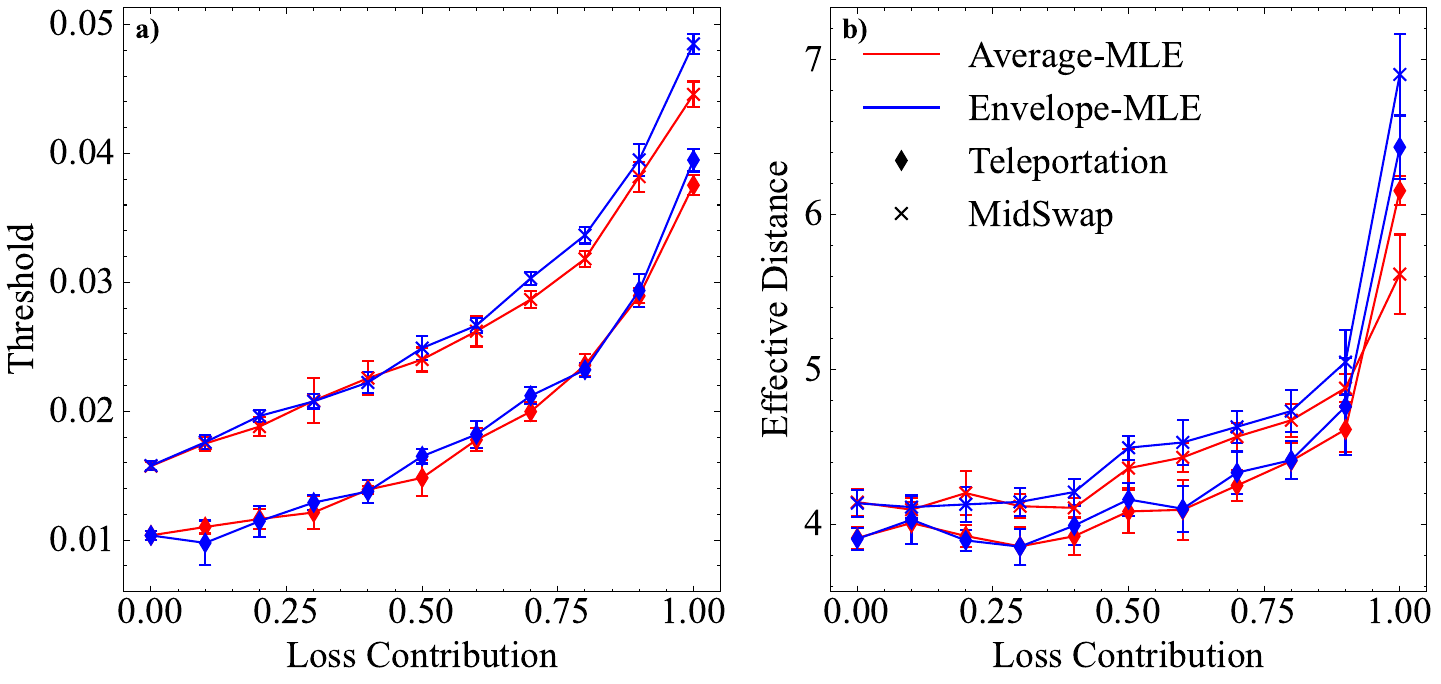}
  \caption{Performance comparison between the teleportation circuit
    and the \midswapSC{}: (a) threshold;
  (b) effective distance at code distance $7$.}
  \label{fig:combined_teleportation}
\end{figure}

Ref.~\cite{baranes2025leveraging} also proposes a teleportation-based
syndrome extraction circuit, which can be viewed as a modification of
the Steane error correction circuit. The circuit is shown in
\zcref{fig:teleportation_circuit}. The two ancilla logical qubits
are prepared non-fault-tolerantly, meaning that only one syndrome
extraction round is performed.
Although Ref.~\cite{baranes2025leveraging} shows that preparing the
ancilla qubits fault-tolerantly can achieve a much higher
threshold, this approach becomes a bottleneck when combined with
transversal logical circuits. Therefore, we only consider the
non-fault-tolerant preparation of ancilla qubits here, which takes
constant time and provides full circuit-level distance.

We exclude the teleportation-based syndrome extraction circuit from the
main text due to its significantly inferior performance compared to
both \midswapSC{} and \swapSC{}. The teleportation circuit exhibits
two major disadvantages: (i) substantially lower threshold and
effective distance, and
(ii) considerably higher space-time overhead.
\zcref{fig:combined_teleportation} quantifies these performance gaps.

To ensure a fair comparison, we exclude idling errors from the
teleportation circuit because its non-fault-tolerant state preparation is
significantly affected by idling errors, while retaining the original
error model for
\midswapSC{}.

The teleportation circuit achieves a threshold that is $0.5\%$ to $1\%$ lower
than that of \midswapSC{}, along with a slightly lower effective
distance.

\section{Measurement and reset error sensitivity}
\label{app:measurement_error_and_reset_error}
Both \swapSC{} and \midswapSC{} exhibit different error-propagation
characteristics compared with the standard surface-code syndrome extraction
circuit under pure Pauli noise. The standard surface-code circuit confines
both measurement and reset errors to time-like edges, providing better
tolerance to these errors. In
contrast, \swapSC{} propagates measurement errors to space-like edges,
while \midswapSC{} propagates reset errors to space-like edges.

\zcref{fig:varying_reset_measurement} examines how logical error rates
depend on the relative magnitudes of reset and measurement error rates
for both circuits. When the sum of reset and measurement error
rates is fixed, the results reveal that the logical error rate of \midswapSC{}
increases
as the reset error rate increases,
whereas \swapSC{}'s logical error rate increases as the measurement
error rate increases.

\begin{figure}[htbp]
  \centering
  \includegraphics[width=0.7\columnwidth]{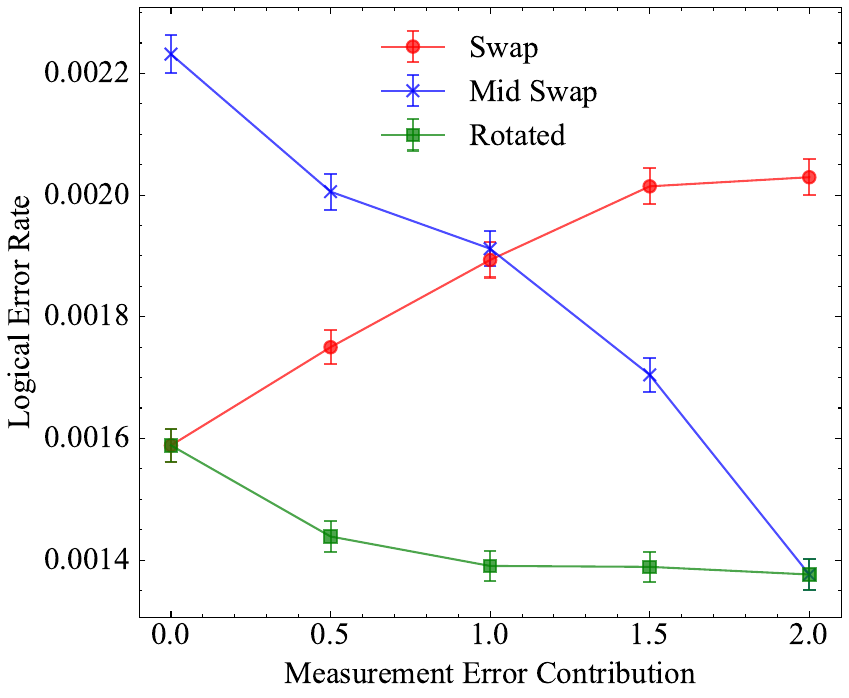}
  \caption{Logical error rate behavior of \midswapSC{}, \swapSC{},
    and standard surface code syndrome extraction with varying reset
    and measurement
    error rates. Simulations are performed at distance $d=5$ with
    error rate $p_{\mathrm{pauli}}=p_{\mathrm{loss}}=0.008$. We scale the
    reset error rate to $0.008x$ and measurement error rate to
    $0.008(2-x)$, where $x$ ranges from $0$ to $2$. The
    \midswapSC{} is more sensitive to reset errors, whereas the \swapSC{} is
  more sensitive to measurement errors.}
  \label{fig:varying_reset_measurement}
\end{figure}

\section{Search for optimal hyperparameters}
\label{app:best_hyperparams}
For \envelopeMatchingDecoder{}, our theoretical analysis
derives reweighting
coefficients of $0.5$ for space-like edges and $0.25$ for time-like
edges.

With finite distances and error rates, these values might not give the
best practical
performance. Furthermore, any value in the range of $(0,0.25)$ for
time-like edges gives the same asymptotic performance.

Therefore, we view the reweighting coefficients as hyperparameters
and empirically optimize them for the
\envelopeMatchingDecoder{} in \swapSC{} and \midswapSC{} and show the
result at code distance $d=13$ as a function of loss contribution in
\zcref{fig:best_hyperparams}.

The results show that the empirically optimized
hyperparameters closely match the theoretical predictions for
\midswapSC{}, confirming the correctness of our analytical approach.
\begin{figure}[htbp]
  \centering
  \includegraphics[width=1.0\columnwidth]{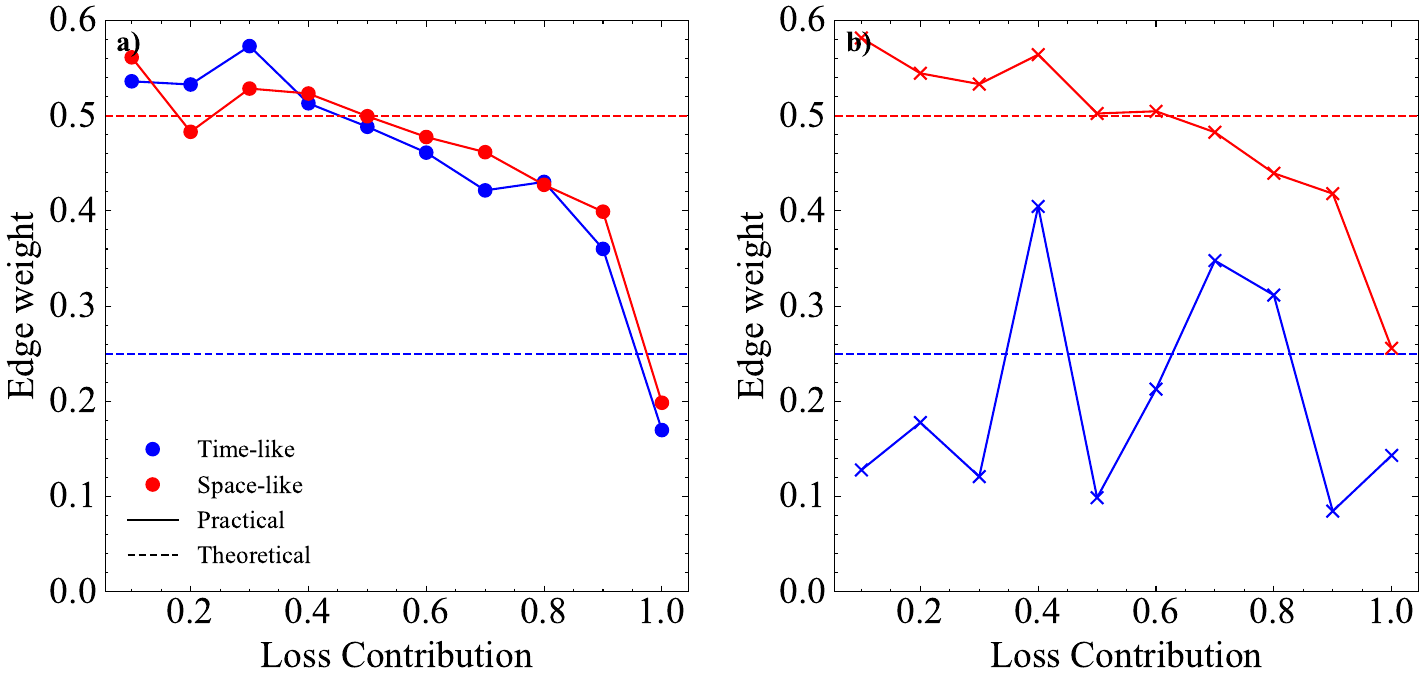}
  \caption{Optimal hyperparameters for the \envelopeMatchingDecoder{}
    in (a) \swapSC{} and (b) \midswapSC{} at code distance $d=13$ as a function
    of the loss contribution.
    For \midswapSC{}, the empirically
    optimized values closely match the theoretically predicted values
    of $0.5$ for space-like edges and $0.25$ for time-like edges,
  validating our analytical approach across different loss rates.}
  \label{fig:best_hyperparams}
\end{figure}

\section{Nonlinearity of atom loss events}
\label{app:non_linearity_of_atom_loss}
Our \envelopeMLEDecoder{} (\zcref{alg:mle-decoder}) converts atom loss events
to Pauli envelope errors rather than directly enumerating loss-induced
detector-observable pairs as in Ref.~\cite{baranes2025leveraging}. This
conversion is not only essential for our proof but also for the
practical performance of the decoder.

Due to the nonlinearity of atom loss, when multiple losses occur
simultaneously but are not modeled properly, the \envelopeMLEDecoder{} has to
use multiple Pauli errors to model the loss-induced errors,
significantly degrading the performance.

\zcref{fig:bounded_pauli_vs_remove_gates} empirically validates this
argument by comparing the \envelopeMLEDecoder{} with the correct
implementation and a modified version where the first step is
removed, and the detector-observable pairs are enumerated by removing
the Clifford gates affected by the loss. For \swapSC{}, the naive
gate-removal approach yields significantly higher
logical error rates, confirming that loss nonlinearity cannot be
ignored in this circuit. However, for \midswapSC{}, both approaches
yield nearly identical performance, indicating that atom losses behave
almost linearly in this architecture. We conjecture that this
near-linearity also contributes to the superior performance of \midswapSC{}
over \swapSC{} when using \envelopeMLEDecoder{}.
\begin{figure}[htbp]
  \centering
  \includegraphics[width=1.0\columnwidth]{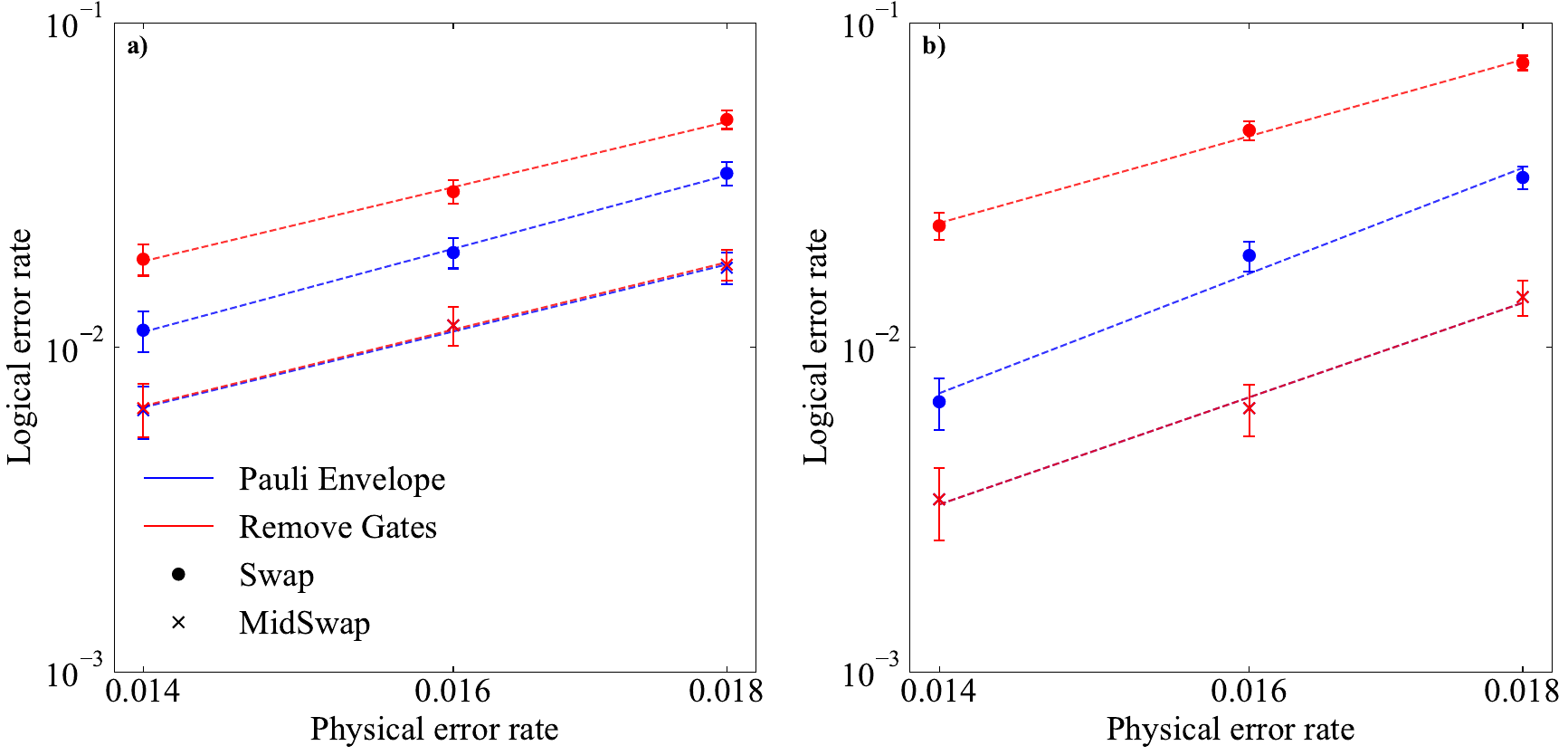}
  \caption{Impact of the Pauli envelope on decoding performance for \swapSC{}
    and \midswapSC{} at code distances (a)~$5$ and (b)~$7$. To emphasize the
  effect, we set the loss contribution to $1$ in all cases.}
  \label{fig:bounded_pauli_vs_remove_gates}
\end{figure}

\section{Attribution of the gain and the fitted exponents of the
baselines}
\label{app:ablation}

\begin{figure*}[htbp]
  \centering
  \includegraphics[width=\textwidth]{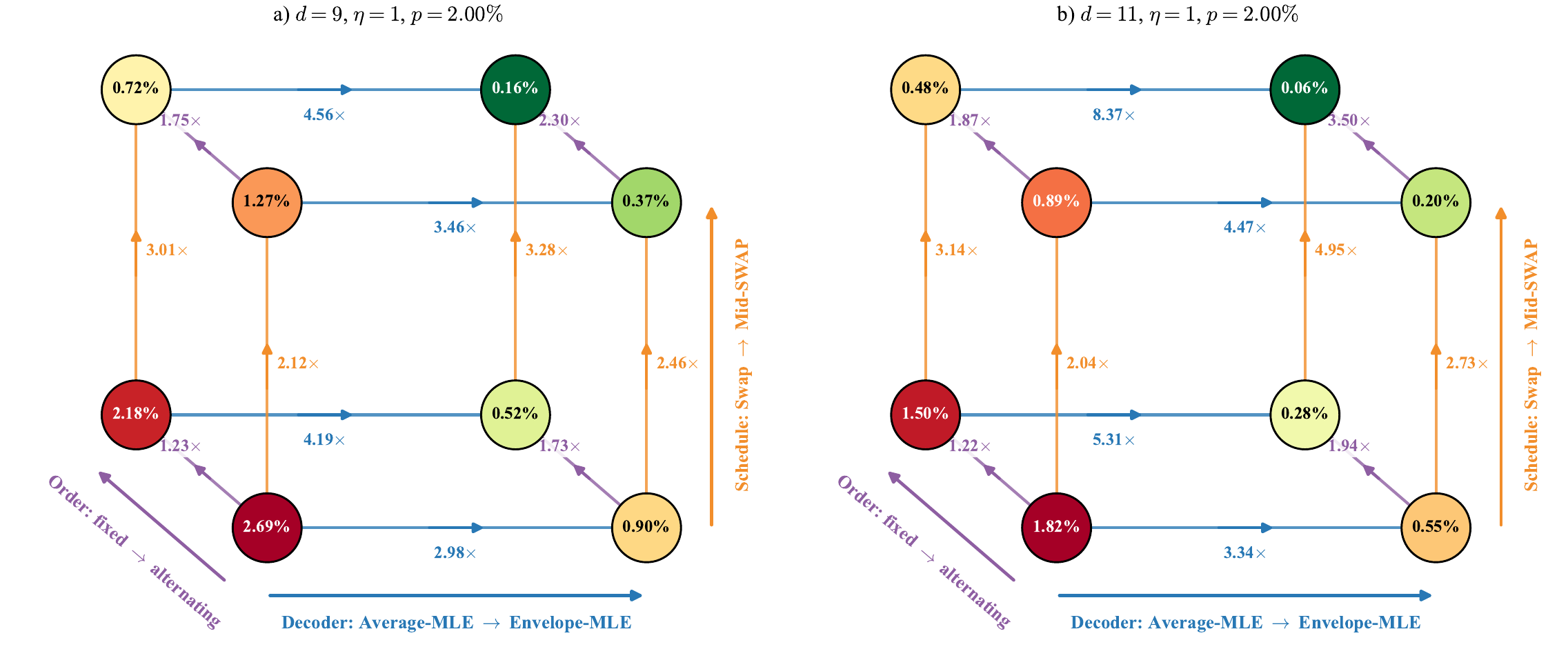}
  \caption{Logical error rate of eight configurations at pure loss
    ($\eta=1$) and $p=2.0\%$ for (a)~$d=9$ and (b)~$d=11$,
    along three axes: decoder (\averageMLEDecoder{} or
    \envelopeMLEDecoder{}), schedule (\swapSC{} or \midswapSC{}), and
    stabilizer order. The front face uses the alternating order of
    this paper (the CNOT order reversed in every other round) and
    repeats \zcref{fig:decoder_schedule_ablation}; the back face uses
    a fixed order made possible by one extra replenishment ancilla per
    data qubit on the bottom and right boundaries. Arrows point from
    the baseline to the improved setting and are labeled with the
    matched error-reduction factor.}
  \label{fig:configuration_cube}
\end{figure*}

\noindent\textbf{Circuit versus decoder.}
\zcref{fig:mle_vs_harvard_threshold} contains all four combinations
of decoder (\averageMLEDecoder{}, \envelopeMLEDecoder{}) and circuit
(\swapSC{}, \midswapSC{}) under the same noise model, distances, and
sampling. \zcref{tab:ablation} collects the numbers that separate the
two contributions, and \zcref{fig:decoder_schedule_ablation} shows
the logical error rates of the four combinations at one pure-loss
operating point.

\begin{table}[htbp]
  \caption{Threshold and effective distance of the four
    decoder--circuit combinations. Thresholds are fitted to
    $d\in\{5,7,9\}$; effective distances are at $\eta=1$.
  Parentheses give the uncertainty in the last digits.}
  \label{tab:ablation}
  \centering
  \resizebox{\columnwidth}{!}{%
    \begin{tabular}{llcccc}
      \hline\hline
      Decoder & Circuit & $p_{\mathrm{th}}^{\eta=1}$ (\%) &
      $p_{\mathrm{th}}^{\eta=0.5}$ (\%) & $d_{\mathrm{eff}}^{d=9}$ &
      $d_{\mathrm{eff}}^{d=11}$ \\
      \hline
      Average-MLE & SWAP & $3.35(3)$ & $2.11(3)$ & $6.48(9)$ & $8.67(14)$ \\
      Average-MLE & Mid-SWAP & $4.18(9)$ & $2.27(5)$ & $6.77(10)$ &
      $8.57(25)$ \\
      Envelope-MLE & SWAP & $3.89(8)$ & $2.13(11)$ & $7.57(14)$ & $10.40(22)$ \\
      Envelope-MLE & Mid-SWAP & $4.75(5)$ & $2.36(6)$ & $8.83(18)$ &
      $11.31(12)$ \\
      \hline\hline
    \end{tabular}%
  }
\end{table}

\noindent\textbf{Stabilizer order.}
The circuits alternate the CNOT order between rounds so that every
boundary data qubit is replenished (\zcref{sec:evaluation}); the
alternative is a fixed order with one extra replenishment ancilla for
each boundary data qubit that the fixed order misses, $O(d)$ in
total. \zcref{fig:configuration_cube} compares the two remedies as
the third axis of a configuration cube whose front face is
\zcref{fig:decoder_schedule_ablation}: the alternating configuration
has the lower logical error rate for every decoder and circuit, by
factors of $1.2$ to $2.3$ at $d=9$ and $1.2$ to $3.5$ at $d=11$.
Since this comparison changes both the order and the boundary
construction, we also ran a control in which only the order changes,
at the fixed geometry of the back face and with the
\averageMLEDecoder{}. Alternating the order lowers the logical error
rate by a factor that grows with the distance: at pure loss and
$p=2.38\%$ it is $1.02$, $1.30$, $1.37$, and $1.50$ for $d=3,5,7,9$;
at $\eta=0.5$ it grows from $1.0$
at $d\leq 5$ to $1.30$ at $d=9$.

This effect is not captured by the average atom lifetime, the figure
of merit used in Ref.~\cite{baranes2025leveraging} to compare
loss-detection circuits, where a shorter lifetime is expected to give
a lower logical error rate. In units of CNOT layers, the average
lifetime is $8+o(1)$ under the alternating order and $8-o(1)$ under
the fixed order with boundary ancillas, the $o(1)$ being boundary
corrections, so this metric slightly favors the fixed order, opposite
to the observed ordering. The alternation is therefore a genuine,
distance-dependent advantage on the planar code, consistent with the
distance-preserving role of alternating orders in
Refs.~\cite{gidney2025alternating,lin2025dynamic}; a proof is left
to future work.

\noindent\textbf{Fitted exponents of the baselines.}
At the error rates and distances studied, the fitted $d_{\mathrm{eff}}$ of the
\averageMLEDecoder{} (\zcref{tab:ablation}) and of the
\marginalMatchingDecoder{} (\zcref{fig:matching_threshold}) lie above
their proven $d_{\mathrm{loss}}\sim d/2$, and that of the
\envelopeMatchingDecoder{} above its proven $2d/3$.
This is expected and does not contradict our theoretical bounds,
which hold for sufficiently small error rates. At finite error rates
and code sizes, the logical error rate can fall faster than the
asymptotic distance scaling, a behavior known as the waterfall
effect~\cite{gu2026scalable}. Several finite-size effects contribute
to the steeper fitted slopes. (i)~Failures heavier than the minimum can
contribute at the finite error rates used in the fit, increasing the
slope above its low-error limit.
(ii)~$d_{\mathrm{loss}}\sim d/2$ means $d/2\pm O(1)$, and at $d\leq 13$
an offset of one or two is a sizable fraction of $d/2$. (iii)~The
minimum-weight failures of \zcref{app:harvard_decoder} require each
loss at a specific circuit location and, for the \averageMLEDecoder{},
that it triggers no detector, so each event carries a constant factor
$c<1$. Their total probability can therefore be small relative to that
of more numerous heavier failures. At these sizes, the practical
separation between the baselines and the \envelopeMatchingDecoder{} is
smaller than the proven bounds suggest; its practical gains
are the higher threshold and the stability as $\eta\to 1$
(\zcref{sec:evaluation}).

\section{Numerical details}
\subsection{Sampling procedure}
\label{app:sampling}

All sampling is done with
Stim~\cite{gidney2021stim}. Stim does not support atom loss natively,
so we extend it in the following way.

We first scan through the circuit, and collect the locations where an
atom can be lost, and the gates that each atom loss location will
affect. For each sample, we draw a loss configuration by triggering
each atom loss independently with the given error rate, and remove
the affected gates.

After that, we get a Stim circuit with some gates removed. We then
insert Pauli errors into the circuit. An atom whose partner is lost
undergoes an idling error. One technical point remains: after the
gates are removed, the observable is no longer deterministic even
without Pauli errors, and Stim rejects such a circuit by default. We
therefore convert the observable into a detector and sample with
Stim's \texttt{allow\_gauge\_detectors} option. After that, we
sample the raw measurement outcomes, mask out those of the lost atoms
so that they leak no information to the decoder, and convert them
into detectors.

Since each new loss configuration requires its own compiled circuit,
which can become a bottleneck when using matching-based decoders, we
use a two-stage sampling procedure. For each loss configuration, we
sample two Pauli-noise shots from the circuit, which roughly speeds
up the sampling by a factor of two but introduces a small correlation
between the two shots.

\subsection{Details on solver timeouts}
\label{app:solver_timeouts}

\begin{figure}[htbp]
  \centering
  \includegraphics[width=\columnwidth]{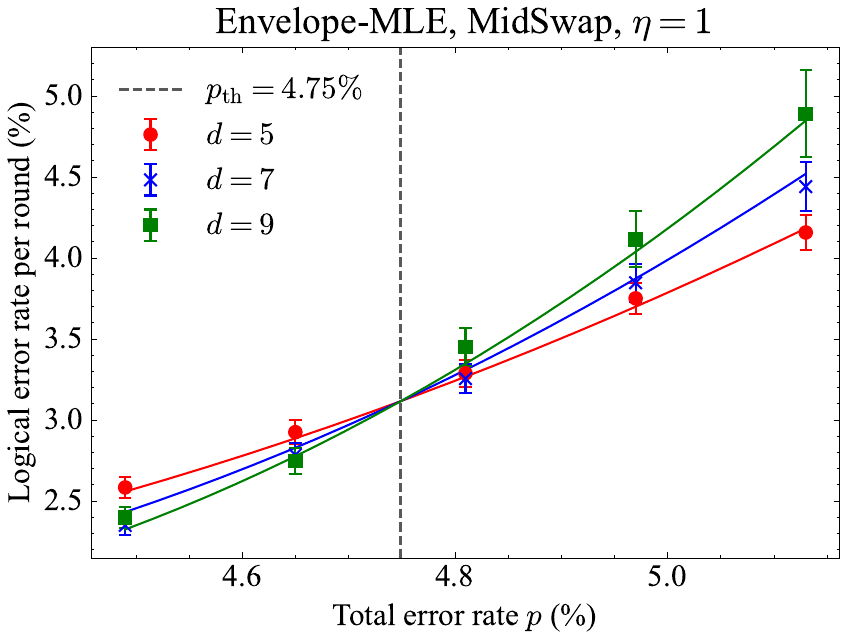}
  \caption{Threshold crossing of the \envelopeMLEDecoder{} on
    \midswapSC{} at loss contribution $\eta=1$ for $d\in\{5,7,9\}$:
    logical error rate per round versus total error rate $p$. The
    solid curves are the common finite-size-scaling fit
    $P_{\mathrm{LER}}=a+bx+cx^2$ with
    $x=(p-p_{\mathrm{th}})\,d^{1/\nu}$, and the dashed line marks the
    fitted threshold. Error bars are one standard error; the $d=9$
  points use $20{,}000$ shots each.}
  \label{fig:mle_midswap_crossing_d579}
\end{figure}

Every shot decoded by the \averageMLEDecoder{} or the \envelopeMLEDecoder{}
is a separate Gurobi MILP, solved with a 60-second time limit and a
relative MIP gap of $1\%$. If the time limit is reached, the best
feasible solution found so far is used as the correction; in
extremely rare cases, the solver is unable to find any feasible
solution, in which case the shot returns a zero correction by default.

In our experiments, time-limit hits are rare. Across all $d=9$ runs
they affect fewer than
$0.02\%$ of the shots, except for the near-threshold data of the
\envelopeMLEDecoder{} on the \midswapSC{}, where $1.2\%$ of the shots
reach the time limit.

To bound the effect of these time-limit hits, we re-decoded
$200{,}000$ near-threshold $d=9$ shots of the \envelopeMLEDecoder{} on
the \midswapSC{} ($\eta\in\{1,0.9,0.8,0.7\}$, five values of $p$ each,
$10{,}000$ shots per point) from identical cached samples under the
two solver settings of \zcref{tab:solver_sensitivity}. Raising the
time limit to $300$~s removes all time-limit hits but changes the
correction of only $51$ shots and the number of logical failures by
five, which is not statistically significant.

\begin{table}[htbp]
  \caption{Re-decoding the same $200{,}000$ near-threshold $d=9$ shots
    of the \envelopeMLEDecoder{} on the \midswapSC{} under two Gurobi
    settings. The ``Changed corrections'' row counts shots whose
    correction differs from
  the $60$~s, $1\%$ MIP gap setting used in the paper.}
  \label{tab:solver_sensitivity}
  \centering
  \begin{tabular}{lcc}
    \hline\hline
    Time limit & $60$~s & $300$~s \\
    Relative MIP gap & $1\%$ & $1\%$ \\
    \hline
    Time-limit hits & $126$ & $0$ \\
    Logical failures & $89{,}297$ & $89{,}302$ \\
    Changed corrections & -- & $51$ \\
    \hline\hline
  \end{tabular}
\end{table}

\subsection{Threshold crossings}
We show an example of the fitted thresholds to show the threshold
crossing behavior. \zcref{fig:mle_midswap_crossing_d579} shows the
underlying crossing for the \envelopeMLEDecoder{} on the \midswapSC{}
at $\eta=1$.

\clearpage
\onecolumngrid
\section{Summary of Notation}
\label{app:notation}
\zcref{tab:notation} collects the main symbols, grouped by topic.
Auxiliary symbols used only in individual proofs are defined locally.

\begin{table}[!ht]
  \providecommand{\notationcell}[1]{\parbox[t]{0.61\textwidth}{\raggedright
  #1\strut}}
  \centering
  \footnotesize
  \begin{tabular}{l l l}
    \toprule
    \textbf{Symbol} & \textbf{Meaning} & \textbf{Introduced in} \\
    \midrule
    \multicolumn{3}{l}{\textit{Circuits and error configurations}} \\
    $\mathcal{C}$, $\mathcal{C}_i$ & \notationcell{Clifford circuit;
    auxiliary circuits used in proofs} & \zcref{sec:pauli-envelope},
    \zcref{app:proof-detector-equivalence} \\
    $p$ & \notationcell{Pauli error configuration: a Pauli operator
    assigned to each faulty space-time location} & \zcref{sec:pauli-envelope} \\
    $l$, $l_i$ & \notationcell{loss configuration: set of space-time
      locations at which atoms are lost; single-atom loss
    configuration} & \zcref{sec:pauli-envelope} \\
    $\emptyset$ & \notationcell{the loss-free configuration} &
    \zcref{def:pauli-bounded} \\
    $\oplus$ & \notationcell{Pauli product with phases discarded; XOR
    for bits or chains; union for loss configurations or readouts} &
    \zcref{sec:pauli-envelope} \\
    $P$, $L$ & \notationcell{sets of Pauli error configurations and
    of loss configurations} & \zcref{sec:pauli-envelope} \\
    $\mathcal{P}_p$, $\mathcal{P}_l$ & \notationcell{distributions of
    Pauli error and loss configurations} & \zcref{sec:pauli-envelope} \\
    $R$, $\mathcal{N}$ & \notationcell{reset gate; completely
    depolarizing channel modeling the random bit of a lost atom} &
    \zcref{app:proof-detector-equivalence} \\
    \midrule
    \multicolumn{3}{l}{\textit{Readouts}} \\
    $\tilde b_m\in\{0,1,\bot\}$ & \notationcell{physical
    three-outcome readout of atom $m$; $\bot$ means the atom is lost}
    & \zcref{sec:pauli-envelope} \\
    $b_m$ & \notationcell{binary bit used to evaluate detectors and
    observables; uniformly random if $\tilde b_m=\bot$} &
    \zcref{sec:pauli-envelope} \\
    $r=(r_m)_m$ & \notationcell{loss-resolving readout, with $r_m=1$
    if and only if $\tilde b_m=\bot$} & \zcref{sec:pauli-envelope} \\
    $\mathcal{R}$ & \notationcell{readout map, $r=\mathcal{R}(l)$} &
    \zcref{sec:pauli-envelope} \\
    $\mathcal{L}(r)$ & \notationcell{preimage of $r$: all loss
    configurations with $\mathcal{R}(l)=r$} & \zcref{sec:pauli-envelope} \\
    $r^{(m)}$ & \notationcell{readout in which only atom $m$ is
    reported as lost} & \zcref{lem:loss-resolving-readout} \\
    \midrule
    \multicolumn{3}{l}{\textit{Detectors, outcomes, and Pauli envelopes}} \\
    $\mathcal{D}$, $\mathcal{O}$ & \notationcell{detector pattern;
    logical observable flip} & \zcref{sec:background},
    \zcref{sec:pauli-envelope} \\
    $\mathcal{S}(p,l)$, $\mathcal{S}(P,L)$ & \notationcell{support of
    detector-observable outcomes; union over sets of configurations}
    & \zcref{sec:pauli-envelope} \\
    $E$, $E(l)$, $E(r)$ & \notationcell{Pauli envelope (a set of
    Pauli error configurations) of a loss configuration or readout} &
    \zcref{def:pauli-bounded} \\
    $\Pi_i$, $\Pi_q^{\tau}$ & \notationcell{Pauli set $\{I,X,Y,Z\}$
    at space-time location $i$ (on qubit $q$ at time $\tau$)} &
    \zcref{lem:detector-equivalence}, \zcref{app:css} \\
    $\textsf{Dec}$ & \notationcell{decoder mapping $(\mathcal{D},r)$
    to a predicted observable} & \zcref{lem:sufficiency} \\
    $P_{\mathrm{LER}}$, $P_{\mathrm{fail}}$ & \notationcell{logical
    error rate; failure probability with Pauli envelope} &
    \zcref{def:effective-distance}, \zcref{sec:pauli-envelope} \\
    \midrule
    \multicolumn{3}{l}{\textit{Distances, error rates, and counts}} \\
    $d$, $d_{\mathrm{loss}}$ & \notationcell{code distance; loss
    distance} & \zcref{def:effective-distance} \\
    $\sim f(d)$ & \notationcell{$f(d)\pm O(1)$} & \zcref{sec:introduction} \\
    $d_{\mathrm{eff}}$ & \notationcell{fitted effective distance
    (exponent of the low-error scaling)} & \zcref{sec:evaluation} \\
    $p_{\mathrm{loss}}$, $p_{\mathrm{pauli}}$, $p_{\mathrm{readout}}$
    & \notationcell{physical loss, Pauli, and loss-readout error
    rates} & \zcref{def:effective-distance}, \zcref{sec:evaluation} \\
    $p$, $\eta$ (\zcref{sec:evaluation}) & \notationcell{total
      physical error rate and loss contribution:
    $p_{\mathrm{loss}}=\eta p$, $p_{\mathrm{pauli}}=(1-\eta)p$} &
    \zcref{sec:evaluation} \\
    $p$, $\eta$ (\zcref{app:correlated}) & \notationcell{loss-rate
      parameter (marginal loss probability $p/2$ per atom) and
    correlated-loss contribution} & \zcref{app:correlated} \\
    $\bar p$ & \notationcell{unified rate in scaling proofs,
    e.g.\ $\max(p_{\mathrm{loss}},\sqrt{p_{\mathrm{pauli}}})$} &
    \zcref{app:proof-failure-probability} \\
    $p_{\mathrm{th}}$ & \notationcell{threshold error rate} &
    \zcref{def:effective-distance} \\
    $n_l$, $n_p$, $n_f$ & \notationcell{numbers of atom losses, Pauli
    errors, and undetected (false-negative) losses} &
    \zcref{thm:mle-optimal-distance}, \zcref{thm:measurement-errors} \\
    $\Lambda$ & \notationcell{error suppression factor} &
    \zcref{sec:evaluation} \\
    \midrule
    \multicolumn{3}{l}{\textit{Decoders}} \\
    $E_m$, $E_{m,n}$ & \notationcell{Pauli envelope of $r^{(m)}$; of
    a correlated loss of atoms $m$ and $n$} &
    \zcref{alg:mle-decoder}, \zcref{alg:mle-correlated-decoder} \\
    $(\mathcal{D}_{m,j},\mathcal{O}_{m,j})$ & \notationcell{$j$-th
    detector-observable pair in $\mathcal{S}(E_m,\emptyset)$} &
    \zcref{alg:mle-decoder} \\
    $x_{m,j}$, $x_{(m,n),j}$ & \notationcell{binary
    envelope-selection variables} & \zcref{alg:mle-decoder},
    \zcref{alg:mle-correlated-decoder} \\
    $y_i$, $w_i$, $w$ & \notationcell{binary variable and
      log-likelihood weight of Pauli error mechanism $i$; common
    weight} & \zcref{alg:mle-decoder} \\
    $\lambda_{m,j}$ & \notationcell{negative log-probability of
      envelope pattern $j$ for atom $m$ reported as lost in the decoder
    model} & \zcref{alg:mle-decoder} \\
    $u$, $s_u$ & \notationcell{detector index; observed value of
    detector $u$} & \zcref{alg:mle-decoder} \\
    $D_{i,u}$, $D_{m,j,u}$ & \notationcell{whether Pauli error $i$ /
    loss pattern $(m,j)$ flips detector $u$} & \zcref{alg:mle-decoder} \\
    $G=(V,\mathcal{E})$, $w_e$ & \notationcell{matching graph and its
    edge weights} & \zcref{alg:envelope-matching-decoder} \\
    $S_1,\ldots,S_4$ & \notationcell{detector pairs (edges) that an
    ancilla loss can trigger in \midswapSC{}} & \zcref{lem:z-shaped-loss} \\
    \midrule
    \multicolumn{3}{l}{\textit{Proofs}} \\
    $F_{m,j}$, $\partial F_{m,j}$ & \notationcell{representative
      chain for $(\mathcal{D}_{m,j},\mathcal{O}_{m,j})$; its detector
    boundary $\mathcal{D}_{m,j}$} & \zcref{app:proof-loss-pattern-reduction} \\
    $C=C_P\oplus C_L$ & \notationcell{difference of actual and
      decoded representative chains, split into Pauli and loss
    components} & \zcref{app:proof-loss-pattern-reduction} \\
    $F$, $\widehat F$ & \notationcell{actual error chain and MWPM
    correction chain in the matching graph} &
    \zcref{app:proof-failure-condition} \\
    \bottomrule
  \end{tabular}
  \caption{Summary of notation.}
  \label{tab:notation}
\end{table}

\end{document}